\documentclass{amsart}
\usepackage[letterpaper,
            top=1in, bottom=1in,
            left=.65in, right=.65in,
            columnsep=1.75em]{geometry}

\usepackage[utf8]{inputenc}
\usepackage{etoolbox}
\usepackage{booktabs} %

\usepackage[table]{xcolor}
\usepackage{libertine}
\usepackage{amssymb}

\usepackage[libertine]{newtxmath}

\usepackage[font=small,labelfont=bf,margin=0pt,justification=justified,singlelinecheck=false]{caption}

\usepackage{hyperref}
\hypersetup{
	colorlinks,
	citecolor=black,
	filecolor=black,
	linkcolor=black,
	urlcolor=black,
	linktoc=all
}
\hypersetup{colorlinks=true}

\usepackage[square,numbers]{natbib}

\usepackage{picinpar,moresize,xfrac,graphpap,dcolumn,wrapfig,graphicx}

\DeclareGraphicsExtensions{.png,.pdf,.jpg}
\usepackage{subcaption}
\usepackage{stackengine}

\usepackage{tikz,pgfplots}
\usepackage{tikz-cd}
\usepackage{pgf}
\usepgfplotslibrary{colormaps,groupplots,dateplot}
\usetikzlibrary{pgfplots.colormaps}
\usetikzlibrary{arrows,automata,positioning,backgrounds}
\usepackage{rotating}
\pgfplotsset{compat=newest}
\pgfplotsset{plot coordinates/math parser=false}
\newlength\figureheight
\newlength\figurewidth

\usetikzlibrary{patterns,shapes.arrows}

\usepackage{soul,array,calc,url,ragged2e,graphpap}
\usepackage{booktabs} %
\usepackage{tabularx}
\usepackage{colortbl}
\usepackage{multirow}
\usepackage{hyphenat}
\usepackage{transparent}
\usepackage{enumerate}
\usepackage{mathrsfs}
\usepackage{mdframed}
\usepackage{acro}

\usepackage[figure]{hypcap}
\usepackage{mathtools}
\mathtoolsset{centercolon}
\usepackage{nicefrac}
\usepackage{units}
\usepackage[normalem]{ulem}
\usepackage{cancel}
\usepackage{pbox}
\usepackage{multicol}
\usepackage{cases}
\usepackage[all]{xy}
\usepackage{nicematrix}

\usepackage{adjustbox}
\usepackage{wrapfig}
\AfterEndEnvironment{wrapfigure}{\setlength{\intextsep}{0mm}}

\usepackage{accents} %

\usepackage{cleveref}
\crefmultiformat{subequation}%
{\edef\crefstripprefixinfo{#1}(#2#1#3}%
{,#2\crefstripprefix{\crefstripprefixinfo}{#1}#3)}%
{,#2\crefstripprefix{\crefstripprefixinfo}{#1}#3}%
{,#2\crefstripprefix{\crefstripprefixinfo}{#1}#3)}

\makeatletter
\DeclareFontFamily{OMX}{MnSymbolE}{}
\DeclareSymbolFont{MnLargeSymbols}{OMX}{MnSymbolE}{m}{n}
\SetSymbolFont{MnLargeSymbols}{bold}{OMX}{MnSymbolE}{b}{n}
\DeclareFontShape{OMX}{MnSymbolE}{m}{n}{
    <-6>  MnSymbolE5
   <6-7>  MnSymbolE6
   <7-8>  MnSymbolE7
   <8-9>  MnSymbolE8
   <9-10> MnSymbolE9
  <10-12> MnSymbolE10
  <12->   MnSymbolE12
}{}
\DeclareFontShape{OMX}{MnSymbolE}{b}{n}{
    <-6>  MnSymbolE-Bold5
   <6-7>  MnSymbolE-Bold6
   <7-8>  MnSymbolE-Bold7
   <8-9>  MnSymbolE-Bold8
   <9-10> MnSymbolE-Bold9
  <10-12> MnSymbolE-Bold10
  <12->   MnSymbolE-Bold12
}{}
\let\llangle\@undefined
\let\rrangle\@undefined
\DeclareMathDelimiter{\llangle}{\mathopen}%
                     {MnLargeSymbols}{'164}{MnLargeSymbols}{'164}
\DeclareMathDelimiter{\rrangle}{\mathclose}%
                     {MnLargeSymbols}{'171}{MnLargeSymbols}{'171}
\makeatother

\usepackage{fancyvrb,listings}
\usepackage{algorithm}
\usepackage{algorithmicx}
\usepackage[noend]{algpseudocode}

\algrenewcommand\alglinenumber[1]{\footnotesize #1:}
\makeatletter  
 \renewcommand{\ALG@name}{\small Algorithm} 
\makeatother 
\algdef{SE}[DOWHILE]{Do}{doWhile}{\algorithmicdo}[1]{\algorithmicwhile\ #1}%

\newtheoremstyle{mine}{3pt}{3pt}{\itshape}{}{\bfseries}{.}{.5em}{}
\theoremstyle{mine}
\newtheorem{theorem}{Theorem}[section]
\newtheorem{lemma}{Lemma}[section]
\newtheorem{proposition}{Proposition}[section]
\newtheorem{corollary}{Corollary}[section]
\newtheorem{definition}{Definition}[section]
\newtheorem{example}{Example}[section]

\newtheorem{remark}{Remark}[section]

\newcommand{\figref}[1]{\textup{Fig.~\ref{#1}}}
\newcommand{\tabref}[1]{\textup{Table~\ref{#1}}}
\newcommand{\teqref}[1]{\textup{Eq.~(\ref{#1})}}

\newcommand{\secref}[1]{\textup{Section~\ref{#1}}}

\renewcommand{\algref}[1]{\textup{Alg.~\ref{#1}}}

\newcommand{\lemref}[1]{\textup{Lemma~\ref{#1}}}

\def\cf{\emph{cf.}}
\def\ie{\emph{i.e.}}
\def\eg{\emph{e.g.}}

\def\CC{\mathbb{C}}

\def\RR{\mathbb{R}}

\def\TT{\mathbb{T}}

\def\ZZ{\mathbb{Z}}

\def\bA{\mathbf{A}}
\def\bB{\mathbf{B}}
\def\bC{\mathbf{C}}
\def\bD{\mathbf{D}}

\def\bF{\mathbf{F}}

\def\bI{\mathbf{I}}
\def\bJ{\mathbf{J}}
\def\bK{\mathbf{K}}
\def\bL{\mathbf{L}}
\def\bM{\mathbf{M}}

\def\bP{\mathbf{P}}

\def\bR{\mathbf{R}}
\def\bS{\mathbf{S}}
\def\bT{\mathbf{T}}
\def\bU{\mathbf{U}}
\def\bV{\mathbf{V}}
\def\bW{\mathbf{W}}
\def\bX{\mathbf{X}}
\def\bY{\mathbf{Y}}
\def\bZ{\mathbf{Z}}

\def\cA{\mathcal{A}}
\def\cB{\mathcal{B}}

\def\cD{\mathcal{D}}

\def\cF{\mathcal{F}}

\def\cH{\mathcal{H}}
\def\cI{\mathcal{I}}

\def\cL{\mathcal{L}}

\def\cP{\mathcal{P}}

\def\cR{\mathcal{R}}
\def\cS{\mathcal{S}}

\def\cV{\mathcal{V}}

\def\sfE{\mathsf{E}}
\def\sfF{\mathsf{F}}

\def\sfM{\mathsf{M}}

\def\sfP{\mathsf{P}}

\def\sfe{\mathsf{e}}
\def\sff{\mathsf{f}}

\def\sfp{\mathsf{p}}
\def\sfq{\mathsf{q}}

\def\ba{\mathbf{a}}

\def\bff{\mathbf{f}}
\def\bg{\mathbf{g}}
\def\bh{\mathbf{h}}

\def\bj{\mathbf{j}}
\def\bk{\mathbf{k}}

\def\bm{\mathbf{m}}
\def\bn{\mathbf{n}}

\def\bs{\mathbf{s}}

\def\bu{\mathbf{u}}
\def\bv{\mathbf{v}}
\def\bw{\mathbf{w}}
\def\bx{\mathbf{x}}
\def\by{\mathbf{y}}

\def\bdelta{\boldsymbol{\delta}}

\def\blambda{\boldsymbol{\lambda}}

\def\bmu{\boldsymbol{\mu}}

\def\brho{\boldsymbol{\rho}}

\def\bpsi{\boldsymbol{\psi}}

\def\rV{{\rm V}}

\def\tbf{\skew{2}\tilde{\mathbf{f}}}
\def\tbh{\vphantom{\tilde{\mathbf{f}}}\tilde{\mathbf{h}}}

\usepackage{bbm}
\DeclareSymbolFont{bbold}{U}{bbold}{m}{n}
\DeclareSymbolFontAlphabet{\mathbbold}{bbold}
\newcommand{\ii}{\mkern1.5mu\mathbbold{i}\mkern1.5mu}

\renewcommand{\Im}{\operatorname{Im}}
\renewcommand{\Re}{\operatorname{Re}}

\def\tr {\operatorname{tr}}
\def\det{\operatorname{det}}

\def\im {\operatorname{im}}
\def\ker{\operatorname{ker}}
\def\spanset{\operatorname{span}}
\def\id{\operatorname{id}}

\DeclareMathOperator{\GL}{GL}

\let\O\relax
\DeclareMathOperator{\O}{O}
\DeclareMathOperator{\SO}{SO}

\renewcommand{\so}{\mathfrak{so}} %
\DeclareMathOperator{\su}{\mathfrak{su}}

\DeclareMathOperator{\Diff}{Diff}
\DeclareMathOperator{\SDiff}{SDiff}

\DeclareMathOperator{\Ad}{Ad}

\DeclareMathOperator{\bad}{\mathbf{ad}}
\DeclareMathOperator{\diff}{\mathfrak{diff}}
\DeclareMathOperator{\sdiff}{\mathfrak{sdiff}}

\DeclareMathOperator{\Adv}{Adv}
\DeclareMathOperator{\adv}{adv}

\DeclarePairedDelimiter\parens{\lparen}{\rparen}

\DeclarePairedDelimiterX\braket[2]{\langle}{\rangle}{#1\,\delimsize\vert\,\mathopen{}#2}

\newcommand{\grad}{\mathop{\mathrm{grad}}\nolimits}

\renewcommand{\div}{\mathop{\mathrm{div}}\nolimits}
\newcommand{\LD}{\mathop{\mathscr{L}}\nolimits}

\newcommand{\rank}{\mathop{\mathrm{rank}}\nolimits}
\newcommand{\Cay}{\mathrm{Cay}}

\DeclareMathOperator{\Vol}{Vol}

\def\bpsi{\boldsymbol{\psi}}

\def\bigmathring{\accentset{\circ}}

\def\HD{\mathcal{HD}}

\numberwithin{equation}{section}

\begin{document}
\tolerance=9999
\emergencystretch=3em

\title{Vakonomic Fluids}

\author{Ritoban Roy-Chowdhury$^{1,2,*}$}
\author{Mohammad Sina Nabizadeh$^{2,1,*}$}
\author{Oliver Gross$^{1}$}
\author{Anthony Gruber$^{3,\dagger}$}
\author{Albert Chern$^{1,\dagger}$}

\subjclass[2020]{Primary XXXX; Secondary YYYY}
\date{Placeholder}

\makeatletter
\newcommand{\amsaffiliations}{%
  \vspace{0.6em}%
  {\centering\footnotesize
  $^{1}$University of California San Diego, La Jolla, CA 92093, USA\\
  $^{2}$Massachusetts Institute of Technology, Cambridge, MA 02139, USA\\
  $^{3}$Sandia National Laboratories, Albuquerque, NM 87123, USA\\[0.4em]
  \textit{Email:}
  \texttt{ritoban@mit.edu},
  \texttt{sinabiz@mit.edu},
  \texttt{ogross@ucsd.edu},
  \texttt{adgrube@sandia.gov},
  \texttt{alchern@ucsd.edu}\\[0.35em]
  $^{*}$Both authors contributed equally to this work.\quad
  $^{\dagger}$Both authors are corresponding authors.\par}%
  \vspace{0.4em}%
}
\begin{abstract}
We introduce a novel discretization of the incompressible Euler equations based on their interpretation as geodesic equations on the Lie group of volume-preserving diffeomorphisms.  It is well known that encoding diffeomorphisms and their infinitesimal generators through a discretized Koopman representation places a nonholonomic constraint on  discrete velocities, for which there is no consensus on a variational treatment.  We show that taking the vakonomic perspective, as opposed to the usual perspective of Lagrange--d’Alembert, yields discrete fluid trajectories that remain geodesics on a (sub-)Riemannian manifold.  In particular, the resulting vakonomic dynamics are Lie--Poisson and their solutions admit a discrete relabeling symmetry, leading to machine-precision satisfaction of Casimir invariants along with a discrete analogue of Kelvin’s Circulation Theorem.  Using an efficient momentum map representation based on low-rank Clebsch variables, we show that these vakonomic fluids behave stably and consistently even at low grid resolutions, leading to increased robustness and physical realism in the long term.
\end{abstract}
\twocolumn[{%
\begin{@twocolumnfalse}
\let\ams@setabstract@orig\@setabstract
\renewcommand{\@setabstract}{%
  \amsaffiliations
  \ams@setabstract@orig
}
\maketitle
\vspace{-0.35em}%
\end{@twocolumnfalse}
}]
\global\topskip\normaltopskip
\makeatother
\flushbottom

\setlength{\textfloatsep}{6pt plus 1pt minus 2pt}
\setlength{\floatsep}{8pt plus 2pt minus 2pt}
\setlength{\intextsep}{8pt plus 2pt minus 2pt}
\setlength{\dbltextfloatsep}{8pt plus 2pt minus 2pt}

\section{Introduction}

Simulating incompressible fluid motion remains a central topic in physics-based simulation, where the   goal is to construct discrete fluid motions that faithfully reproduce the dynamical phenomena of their continuous counterparts.  This task is notoriously challenging due to the information loss inherent in discretization, limiting the resolution at which complex multiscale structures in the flow can be resolved.  This has motivated a long line of work into \emph{structure-preserving} discretizations of the incompressible Euler equations \cite{Elcott:2007:SCP,Mullen:2009:EPI,Pavlov:2011:SPD,Nabizadeh:2022:CF,Nabizadeh:2024:FIP} that aim to account for subgrid-scale effects on the dynamics without explicitly resolving them. Often, this involves discretizing the geometric interpretation of Euler solutions as geodesics on an infinite-dimensional manifold of diffeomorphisms.

\begin{figure}[t]
    \centering
    \includegraphics[width=\linewidth]{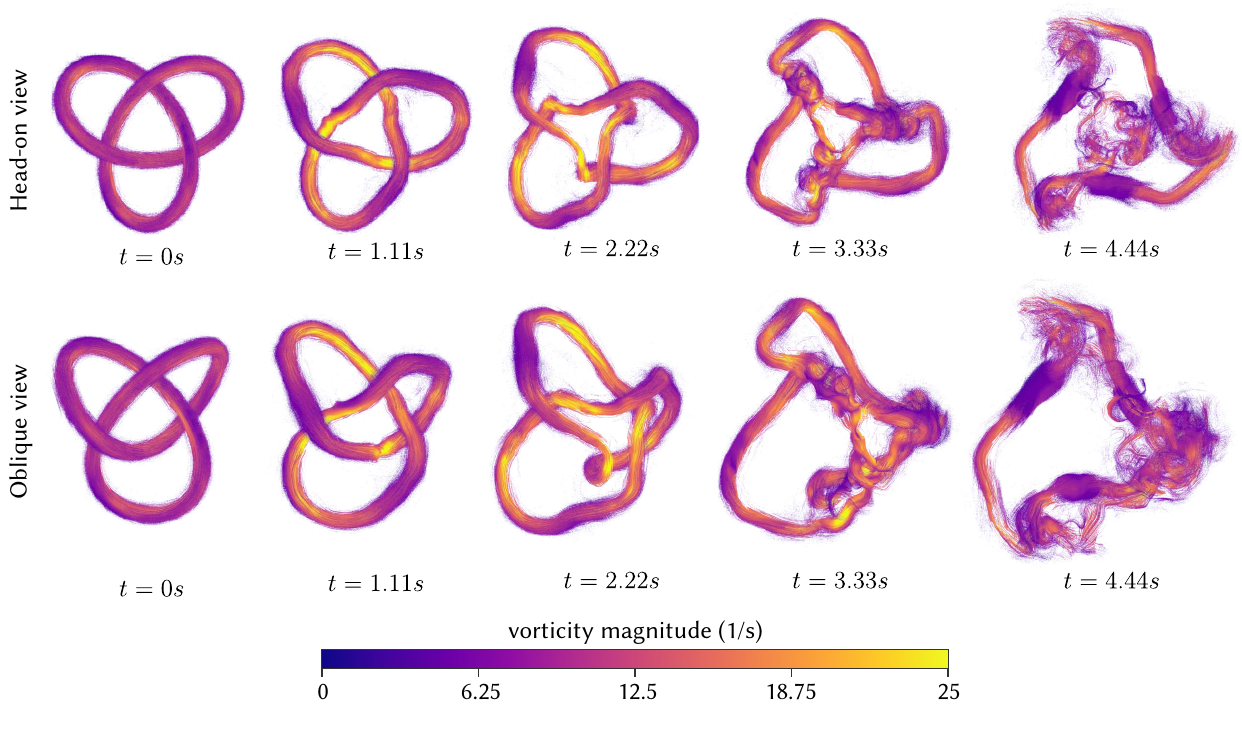}
    \caption{Trefoil-knot experiment computed with our Vakonomic Fluids formulation. Note the faithful capture of the reconnection: the knot advances and stretches before pinching off into a larger and a smaller ring \cite{Kleckner:2013:CDK}.  Prior approaches struggle with different aspects: 
    matrix hydrodynamics \cite{Modin:2025:MHD} is restricted to two-dimensional periodic or spherical domains, while Lagrange--d'Alembert finite element solvers \cite{Pavlov:2011:SPD} preserve only a weak notion of Kelvin's circulation theorem. Our method extends seamlessly to three dimensions on a standard prescription of finite element spaces, here an IGA--FEEC B-spline basis giving pointwise divergence-free velocity, while retaining a Lie--Poisson structure that preserves a discrete analogue of Kelvin's circulation theorem.}
    \label{fig:trefoilknot}
\end{figure}

This Lagrangian perspective on fluid motion is well understood in the continuous setting.  On a Riemannian manifold $M$, Arnold \cite{Arnold:1966:GDG} showed that solutions to the incompressible Euler equations are geodesics (\ie, locally shortest paths) on the space \(\SDiff(M)\) of volume-preserving diffeomorphisms of $M$. Importantly, this picture is fully self-consistent: the equations of incompressible fluid motion are the Euler-Lagrange equations of a certain variational problem posed on the Lie algebra $\sdiff(M)$, and no additional information from outside this particular tangent space to $\SDiff(M)$ is required to establish this fact.

As it turns out, the picture in the discrete setting is far less clear.  Powerful and convenient discretization methods based on, \eg, Koopman representation theory \cite{Koopman:1931:HST,Brunton:2022:MKT}, lead to discrete 
problems with \emph{non-holonomic constraints} for which there is no consensus on a variational treatment.  In short, these constraints arise from representing smooth objects in the Lie algebra $\sdiff(M)$ with discrete ones living in a Lie-algebraically open subspace of matrices \cite{Pavlov:2011:SPD,Gawlik:2011:GVD,Abanov:2025:IDN}.
Largely, there are two competing theories for handling this, known as \emph{Lagrange--d'Alembert} (LdA) \cite{Aalembert:1743:TDD} and \emph{vakonomic} (\textbf{v}ariational \textbf{a}xiomatic \textbf{k}ind of hol\textbf{onomic}) \cite{Kozlov:1983:DSN, Arnold:2006:MAC, Bloch:2015:NMC}, respectively.  While these theories agree when constraints are holonomic (\ie, integrable), and also when certain initial data are prescribed ~\cite{Favretti:1998:EDN, Fernandez:2008:EDN}, 
they produce quite different results in the general non-holonomic case confronted by discrete fluid simulation.  

Until now, the LdA principle has been the dominant approach by far due to its flexibility and strong performance in a variety of computational settings \cite{Marsden:2001:DMV,Pavlov:2011:SPD,Gawlik:2011:GVD,Desbrun:2014:VDR,Liu:2015:MVFS,Bauer:2017:TGV,Bauer:2017:VIA,Natale:2018:VFD,Brecht:2019:VIR,Gawlik:2020:CFE,Gawlik:2021:SPE,Gawlik:2021:VFD,Gawlik:2021:VFD,Gawlik:2022:FEM,Gawlik:2024:VTC,Gawlik:2025:SPT}.  However, LdA has a significant drawback: it is \emph{not} self-consistent, and metric information from outside the solution space is required to formulate well-defined equations of motion. This means the user must make an ad-hoc (non-variational) choice of ``ambient metric'' in order to close the LdA system, rendering its application not fully variational in practice and destroying the symplectic structure of the underlying geodesic equations.  This manifests as incorrect physical behavior in the long term, since violation of discrete Casimir invariants analogous to enstrophy and helicity causes incorrect energy cascading over time.  

The primary goal of this paper is to establish the vakonomic variational principle as an alternative to LdA for the treatment of incompressible fluid simulations.  While LdA allows variations to leave the constraint space,
the vakonomic principle makes a different choice: self-consistency will be maintained by allowing only those variations which obey the non-holonomic constraints. 
This has the enormous positive consequence of \emph{maintaining symplecticity}, so that all invariants of phase space (\ie, Casimirs) will be conserved to machine precision.  However, it also comes with the drawback of reduced flexibility in the discrete variational problem, which will be inherently non-dissipative and constrained to a non-integrable distribution of tangent subspaces where the (derivative of the) fluid is allowed to move.  Ultimately, the vakonomic perspective yields an interpretation of discrete incompressible fluids as \emph{sub-Riemannian geodesics} on a matrix manifold, maintaining a discrete picture consistent with the continuous case and providing practical benefits such as more accurate conserved quantities and more realistic long-term fluid behavior.

To summarize, the contributions of this work include:

\begin{itemize}
    \item  
    Recasting incompressible fluid simulation through the vakonomic variational principle, a self-consistent alternative to the standard Lagrange--d'Alembert approach for handling the non-holonomic constraints inherent to discrete geometric hydrodynamics.
    \item A formulation of the resulting Lax equation as a discrete Lie--Poisson system that preserves energy and Casimir invariants, along with a discrete analogue of Kelvin's Circulation Theorem, to solver tolerance.
    \item A low-rank, momentum map representation of the vakonomic dynamics that dramatically improves computational efficiency while 
    exactly retaining its structure-preserving properties.
    \item Comparisons with previous methods based on functional maps and matrix hydrodynamics, demonstrating how the proposed discretization extends naturally to three dimensions and arbitrary domains.
    \item Numerical experiments involving vortex evolution and turbulence formation that confirm the theory, demonstrating robust, physically faithful long-term behavior, even at low grid resolutions.
\end{itemize}

The remainder of the paper is structured as follows.  \Cref{sec:overview} provides a brief summary of the core ideas driving the proposed approach.  Then, \cref{sec:background} resp. \cref{sec:koopman} review the basics of geometric hydrodynamics resp. Koopman representation necessary for understanding the spatial discretization presented in \cref{sec:spatial-discretization}.  After this introduction to the discrete setup, \cref{sec:EquationsOfMotion} presents the semi-discrete vakonomic equations of motion in Euler--Poincar{\'e} (Lagrangian) and Lie--Poisson (Hamiltonian) forms, before comparing and contrasting them to other modern fluid simulation approaches (including LdA) in \cref{sec:comparison}.  
\Cref{sec:LowRankParameterization} then formulates a computationally efficient low-rank parameterization of the time-continuous dynamical variables analogous to the classical Clebsch representation, before \cref{sec:time-integration} discusses the fully discrete equations of motion along with properties of the time integration scheme employed. Following this, \cref{sec:algorithm} presents the end-to-end simulation algorithm, and a variety of numerical examples follow in \cref{sec:experiments}.  Finally, \cref{sec:conclusion} concludes by outlining some
directions for future work.

\section{Overview}\label{sec:overview}
We begin with a high-level overview of the topics and contributions just discussed. 

\subsection{Lagrangian dynamics with non-holonomic constraints}\label{subsec:non-holonomic}
Variational principles specify the dynamics of a physical system via the extremization of functional data.  Consider the energy functional 
\begin{equation}
\label{eq:GenericQuadraticForm}
\mathcal{E}(\gamma)=\int_0^1 \tfrac{1}{2}\lvert\dot\gamma\rvert^2\,dt,
\end{equation}
defined in terms of a path \(\gamma\colon[0,1]\to M\) on a manifold $M$ with velocity field $\dot\gamma\colon[0,1]\to TM$.  The unconstrained variational problem involves a one-parameter family of curves $\gamma_{\epsilon}\colon[0,1]\to M$ called variations, depending on $\varepsilon\geq 0$, whose associated infinitesimal variations (or variation vectors) are given by the derivative $\bigmathring{\gamma} = \tfrac{d}{d\epsilon}\big|_{\epsilon=0}\,\gamma_{\epsilon}$ (see \Cref{fig:VariationalCartoon}).  
The dynamics generated by $\mathcal{E}$ coincide with the conditions required for its stationarity 
under arbitrary variations $\gamma_{\epsilon}$
with compact support in the interior of \([0,1]\). 

In the presence of constraints, admissible trajectories for $\mathcal{E}$ are additionally required to satisfy
\(\dot\gamma(t)\in \cD_{\gamma(t)}\) for a given distribution \(\cD\subset TM\), \ie, a smooth assignment of
linear subspaces \(\cD_\gamma\subset T_\gamma M\) contained in each tangent space.  \emph{Holonomic} (\ie, integrable) constraints are those for which the
admissible velocities $\dot\gamma$ are tangent to a foliation of \(M\) by submanifolds, meaning that
the constrained variational problem reduces to an unconstrained problem on each leaf where these admissible trajectories lie.
Conversely, when the distribution \(\cD\) is \emph{non-holonomic} (non-integrable), admissible paths \(\gamma\) need not lie in any submanifold of $M$.  Instead, they are merely required to be tangent to the distribution $\cD$ at any given time, \ie, \(\dot\gamma\in\cD_\gamma\).  As mentioned previously, this forces a choice between different principles for imposing these constraints.

\begin{figure}
    \centering
    \includegraphics[width=\linewidth]{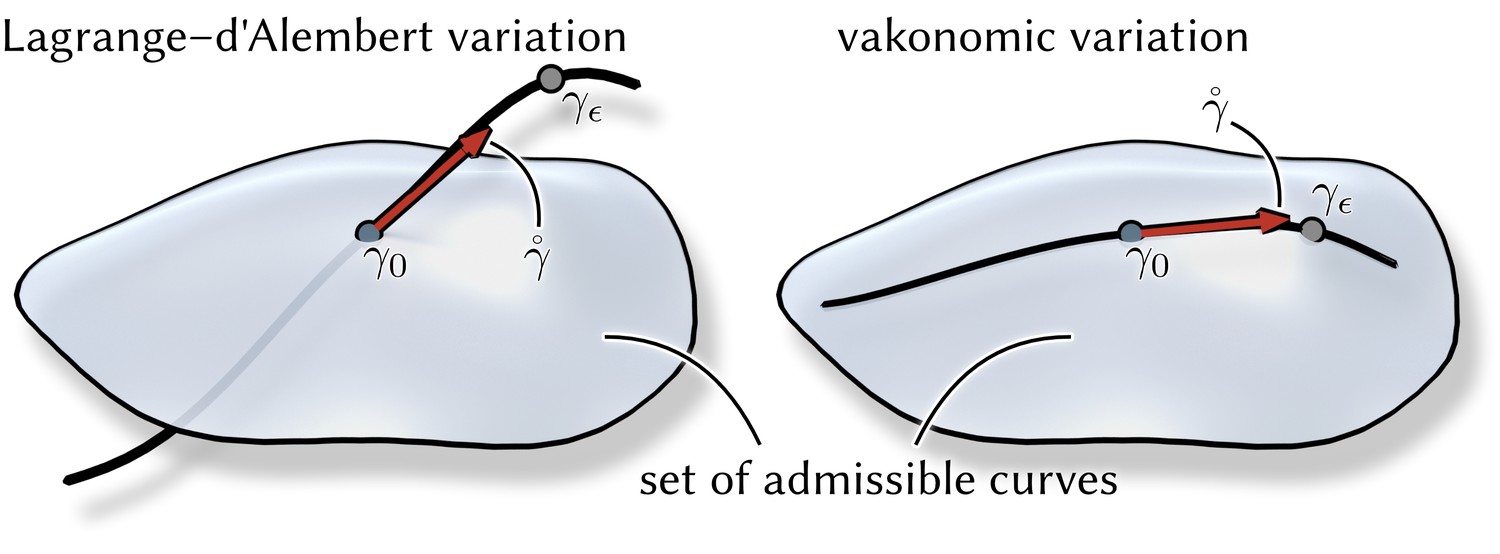}
    \caption{A cartoon of ``variation space'': each point is a path $\gamma$, the shaded region collects the \emph{admissible} paths ($\dot\gamma\in\cD_\gamma$), each black curve is a variation family $\epsilon\mapsto\gamma_{\epsilon}$, and the red arrow is the infinitesimal variation $\mathring\gamma$. Lagrange--d'Alembert (left) lets the family leave the admissible set, so the arrow escapes the shaded region; the vakonomic principle (right) keeps every $\gamma_{\epsilon}$ admissible, confining the arrow to it. The restrictions are complementary: pointwise on $M$ (see \Cref{fig:LdAvsVak}) the picture reverses, with LdA confining $\mathring\gamma(t)$ to $\cD_{\gamma(t)}$ and the vakonomic principle instead confining the velocities $\dot\gamma_{\epsilon}(t)$.
    }
    \label{fig:LdAvsVakAbstract}
\end{figure}

\begin{figure}
    \centering
    \includegraphics[width=\linewidth]{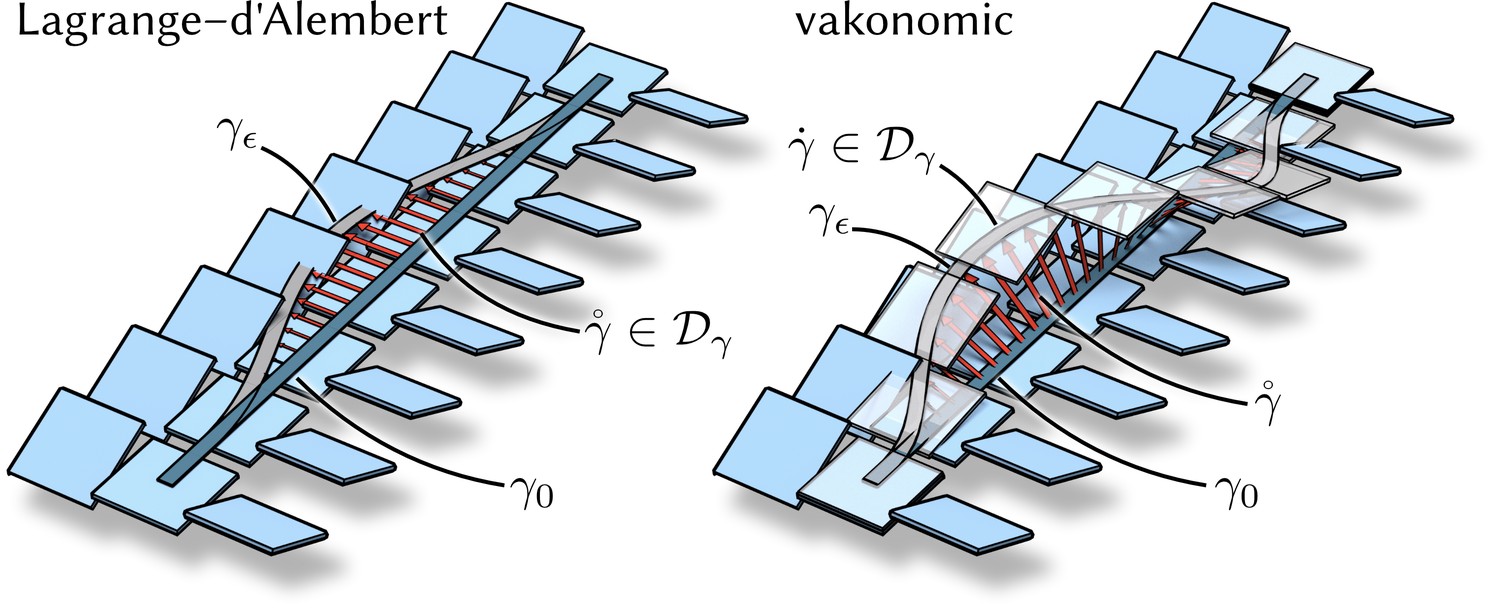}
    \caption{Under the vakonomic variational principle (right), variation curves $t\mapsto \gamma_{\epsilon}(t)$ must stay tangent to the constraint distribution. Conversely, the Lagrange–d'Alembert principle (left) is less restrictive, requiring only that the variation vector field $t\mapsto \mathring{\gamma}(t)$ belong to this distribution.}
    \label{fig:LdAvsVak}
\end{figure}

Of the two standard approaches to 
enforcing non-holonomic constraints, the most common choice in computational mechanics is the Lagrange--d'Alembert principle (see \Cref{sec:ComparisonToLdA}). LdA states that non-stationary variations $\gamma_\epsilon$ in the family should not necessarily be admissible themselves, provided that the ``virtual work'' done on them by the constraint forces is zero.  This means that LdA variations can leave the space of admissible paths (see \Cref{fig:LdAvsVakAbstract} left), \ie, that the velocity $\dot{\gamma}_{\epsilon}$ is not constrained to the distribution $\cD$ (see \Cref{fig:LdAvsVak} left), provided the variation vector field $\bigmathring\gamma$ remains constrained to $\cD$ at each point in time.  This condition is equivalent to the so-called ``principle of virtual work'', also known as Cetaev's condition~\cite{chetaev1932gauss}, which states that LdA-admissible variations $\gamma_{\epsilon}$ must satisfy \(\bigmathring\gamma\in \cD\).
Importantly, the unrestricted nature of the variations $\gamma_{\epsilon}$ renders this  process not intrinsic to the distribution $\cD$, and metric information from the ambient space $TM$ is required to carry out the variational program.  Consequently, the resulting dynamics are not generally the stationarity condition of an energy functional on the space of admissible curves.

The alternative to LdA, standard in control theory and natural from the viewpoint of constrained optimization, is the vakonomic variational principle (\Cref{sec:VakonomicVariationalPrinciple}).  This approach maintains self-consistency by requiring that all variations $\gamma_{\epsilon}$ remain tangent to the constraint distribution, \ie, $\dot{\gamma}_\epsilon(t) \in \cD_{\gamma_\epsilon(t)}$ for all $\epsilon$ in the variational family.  This restriction lends 
the vakonomic principle its fully variational character, since no additional information from outside the distribution $\cD$ is ever required.  By varying over all such admissible curves using, \eg, Lagrange multipliers to enforce the tangency requirement, the resulting vakonomic stationary conditions 
naturally embed notions of sub-Riemannian geometry (\Cref{sec:SubRiemannianStructure}), and the equations of motion describe sub-Riemannian geodesics.

\subsection{Discretizing Incompressible Fluid Motions}

Recall that incompressible fluid motion in the continuous setting is described by an unconstrained Lagrangian least action principle for curves \(\gamma\colon[0,1]\to\SDiff(M)\) in the volume-preserving diffeomorphism group (\Cref{sec:FluidsAsGeodesicEquation}).  Since discretizing $\SDiff(M)$ directly is challenging, many fluid simulation methods (\eg, \cite{Azencot:2014:FFS,Pavlov:2011:SPD,Modin:2024:TDF}) consider a discretized Koopman representation instead.  In the continous Koopman theory,  diffeomorphisms in $\SDiff(M)$ are represented as unitary or orthogonal linear operators acting on a space of smooth functions or half-densities (\Cref{sec:KoopmanRepresentation}). After choosing a finite-dimensional discretization of this space, these continuous linear operators become orthogonal matrices acting on a finite-dimensional vector space (\Cref{sec:spatial-discretization}).
Ultimately, this results in a simple matrix representation space $\SO(F)$ for some integer $F>0$
(\Cref{sec:FEECandVdiv}), along with a non-holonomically constrained variational problem for curves \(\bR\colon[0,1]\to \SO(F)\) of rotation matrices representing the action of fluid motions in $\SDiff(M)$.  Moreover, the velocities of these paths will be skew-symmetric matrices $\bX\colon[0,1]\to\so(F)$ representing skew-adjoint linear operators in the Lie algebra $\sdiff(M)$.

To understand how this non-holonomic constraint arises, consider that 
the continuous Koopman representation just mentioned is a Lie algebra homomorphism at the level of velocities (\Cref{sec:GroupAndAlgebraRepresentations}).  That is, there is a map from vector fields to linear transformations of half-densities, $\vec u \mapsto \adv_{\vec u}$, representing the infinitesimal advection (\ie, directional derivative) along $\vec u$, and the Lie bracket \([{\vec u}, {\vec v}]\) of two vector fields \(\vec u, \vec v\) is represented by the commutator of their associated linear operators: 
\begin{equation}
	\label{eq:LieAlgebraHom}
	\adv_{[\vec u,\vec v]} = [\adv_{\vec u}, \adv_{\vec v}].
\end{equation}
This can be understood as a ``smallness condition" on the image space $\im(\adv)$: not many skew-adjoint linear operators on half-densities represent the directional derivative along some vector field (\Cref{fig:HDLieAlgebraCartoon}).
However, after discretizing the half-densities being acted on, the target space becomes an algebra \(\so(F)\) of matrices, and the image $\im(\bA)$ of the associated discrete map $\bA\approx\adv$ is no longer closed under the action of the matrix commutator (\Cref{ex:LieNonsubalgebra}).  This means that the Lie brackets of elements contained there may not remain there, \ie, $[\im(\bA),\im(\bA)]\not\subset\im(\bA)$. Therefore, the homomorphism property is lost due to discretization of the target space, and admissible discrete motions $\bR$ at the group level must include a linear constraint on the body velocity, $\bX\in\im(\bA)$, at the algebra level.  Since this constraint defines a non-integrable distribution $\cD$ 
in the tangent bundle $T\SO(F)$, this places the discrete dynamics within the non-holonomic framework introduced in \Cref{subsec:non-holonomic}.

\subsubsection{Equations of motion}
Given the difference in interpretation between the LdA and vakonomic cases, it is not surprising that their resulting equations of motion are quite different as well.  In the case of LdA (\Cref{sec:ComparisonToLdA}), invoking the principle of virtual work leads to the following system on $\so(F)\cong\so(F)^*$ 
\begin{equation}\label{eq:LdA-overview}
\dot{\bZ} = \pi_{\cD}\!\big([\bX,\bZ]\big),
\end{equation}
where $\bX=\bX(\bZ)$ is an element of $\so(F)$ depending on $\bZ$ and  \(\pi_{\cD}\) denotes projection onto the constraint distribution.  This projection requires a choice of metric information from outside $\cD\subset\so(F)$, inducing a nontrivial design decision leading to infinitely many different LdA equations of motion.  While this offers a degree of flexibility in the LdA formulation, it is far from obvious what choice is the most appropriate in any given computational setting.

Conversely, the vakonomic equations of motion  (\Cref{sec:VakonomicEquationsOfMotion}) take the usual form of a matrix Lax equation~\cite{Lax:1968:INE}:
\begin{equation}\label{eq:vak-overview}
\dot{\bZ} = [\bX,\bZ].
\end{equation}
Since these equations are truly the stationarity condition of a discrete functional, they inherit the benefits of variational structure, including conservation laws implied by discrete symmetries, and, in particular, the isospectral (Casimir-preserving) character of the Lax evolution (\Cref{sec:CasimirsAndNoetherCharges}).  This is not the case for LdA, since projecting the Lie bracket action using $\pi_\cD$ violates skew-adjointness, leading to a loss of Hamiltonian structure and an un-physical drift in Casimirs over time.

\subsection{A structure-preserving low-rank approach}
The proposed vakonomic fluid discretization relies on an $F$-dimen\-sional subspace of scalar functions or half-densities, along with a $V$-dimensional subspace of vector fields.  However, the resulting equations of motion are largely independent of how these spaces are chosen.  This means that a variety of standard (or exotic) constructions are applicable, including those based on compatible finite elements, mimetic B-splines, or spectral methods.  Remarkably, and in contrast to previous work such as \cite{Pavlov:2011:SPD}, it is not strictly necessary to use a compatible or mimetic discretization, since the advection operators employed here are structure-preserving by design
(\Cref{sec:FEECandVdiv}).

With this said, the vakonomic equations \teqref{eq:vak-overview} are still an isospectral Lax-type evolution in a dense, skew-symmetric matrix variable $\bZ\in\mathbb{R}^{F\times F}$.  As such, na\"{i}ve simulation strategies incur a high cost due to the $O(F^2)$ degrees of freedom which must be managed at each iteration.  Fortunately, this can be mitigated without sacrificing structure-preservation by using a low-rank parameterization of $\bZ$ in terms of \emph{Clebsch variables} (\Cref{sec:LowRankParameterization}).  This Clebsch representation arises via a group-equivariant momentum map and expresses the full matrix variable in terms of a small number (say $m$) of rank-two skew-symmetric matrices.  Evolving the corresponding factor pairs yields a closed system whose reconstruction reproduces the same isospectral flow, forming an analytical reduced-order model 
for the vakonomic evolution.  This reduces the complexity at each iteration from $O(F^2)$ to $O(mF)$ while retaining the variational character of the scheme and its associated conservation properties.  

Ultimately, the incompressible fluid discretization presented in the remainder of the paper provides a ``proof of concept'' for the vakonomic point of view.  We hope that the flexibility, self-consistency, and long-term geometric accuracy of this perspective will inspire further work towards a fully viable and general-purpose fluid solver based on these ideas.

\section{Review: Geometric Fluid Dynamics}\label{sec:background}
This section briefly reviews the geometric theory of fluids necessary for the present approach.
A summary of the corresponding notational analogues is given in~\tabref{tab:Notation}.

\subsection{Flow Maps}
Fluid motion on a manifold \(M\) can be described geometrically in terms of flow maps that transport material points over time. Collectively, these maps form the diffeomorphism group \(\Diff(M)\) of the manifold, which provides a natural configuration space for fluid flows. Time-varying fluid motions on $M$ correspond to trajectories through this space, with trajectory velocities represented by time-varying vector fields on 
$M$.

\subsubsection{Flow maps and Eulerian velocities.}

In this work, the configuration space of the fluid is modeled as the infinite-dimensional Lie group of diffeomorphisms  
\begin{equation*}
\label{eq:Diff}
    \Diff(M)\coloneqq\left\{\phi\colon M\to M\mid \text{smooth}\,\phi^{-1}\,\text{exists}\right\}\subset C^\infty(M;M),
\end{equation*}
where \(M\) is a fluid domain with boundary $\partial M$. This group is a proper subset of the space $C^{\infty}(M;M)$ of smooth maps from $M$ to itself. 
Its Lie algebra,
\begin{equation*}
    \diff(M)\coloneqq T_{\id}\Diff(M) = \left\{ \vec u \in \Gamma(TM) \mid \vec u \,\,\text{tangent to}\, \partial M  \right\},
\end{equation*}
consists of smooth vector fields on \(M\) satisfying the no-through boundary condition on \(\partial M\). Right-translation \(dR_\phi: \diff(M) \to T_{\phi}\Diff(M)\) corresponds to pre-composition with a diffeomorphism \(\phi\): for a vector field \(\vec u\in \diff(M)\),
\begin{equation*}
dR_\phi(\vec u) = \vec u \circ \phi.
\end{equation*}

A time-dependent flow map is a smooth path \(t\mapsto\phi_t\in\Diff(M)\) that assigns to each material point its position at time \(t\). The time derivative \(\dot\phi_t = {d\over dt}\phi_t\) is a tangent vector along this path, and can be represented by a time-dependent vector field \(\vec u_t\in\diff(M)\) via right translation 
\begin{equation}
    \label{eq:RepresentPhiDotByU}
	\dot\phi_t=\vec u_t\circ\phi_t.
\end{equation}
This relation expresses the Lagrangian velocity $\dot\phi_t$ of the flow map in terms of an Eulerian velocity field $\vec u_t$ defined on the fixed spatial domain \(M\).  For simplicity of notation, we will omit the subscript indicating time dependence when it is clear from context.

\subsubsection{Incompressible fluid flows}
To discuss incompressible fluid flows on a manifold \(M\) equipped with a fixed volume form \(\Vol\),
we consider the distinguished subgroup \(\SDiff(M)\subset\Diff(M)\) of volume-preserving diffeomorphisms,
\begin{equation*}
\SDiff(M)\coloneqq\{\phi\in\Diff(M)\mid \det(d\phi)=1 \}. %
\end{equation*} 
Its Lie algebra consists of those vector fields in \(\diff(M)\) that are divergence-free, \ie,
\begin{equation*}
\sdiff(M)\coloneqq \{ \vec u\in\diff(M) \mid %
\div\vec u\coloneqq \LD_{\vec u}\Vol %
= 0\}
\subset\diff(M),
\end{equation*}
where $\LD_{\vec u}$ indicates the Lie derivative along $\vec u$.
Accordingly, the velocities of incompressible flow maps in \(\SDiff(M)\) admit a representation as in \eqref{eq:RepresentPhiDotByU} by right translations of divergence-free vector fields \(\vec u\in\sdiff(M)\).  This is illustrated in \Cref{fig:DivFreeCartoon}.

\begin{figure}[h]
    \centering
    \includegraphics[width=\linewidth]{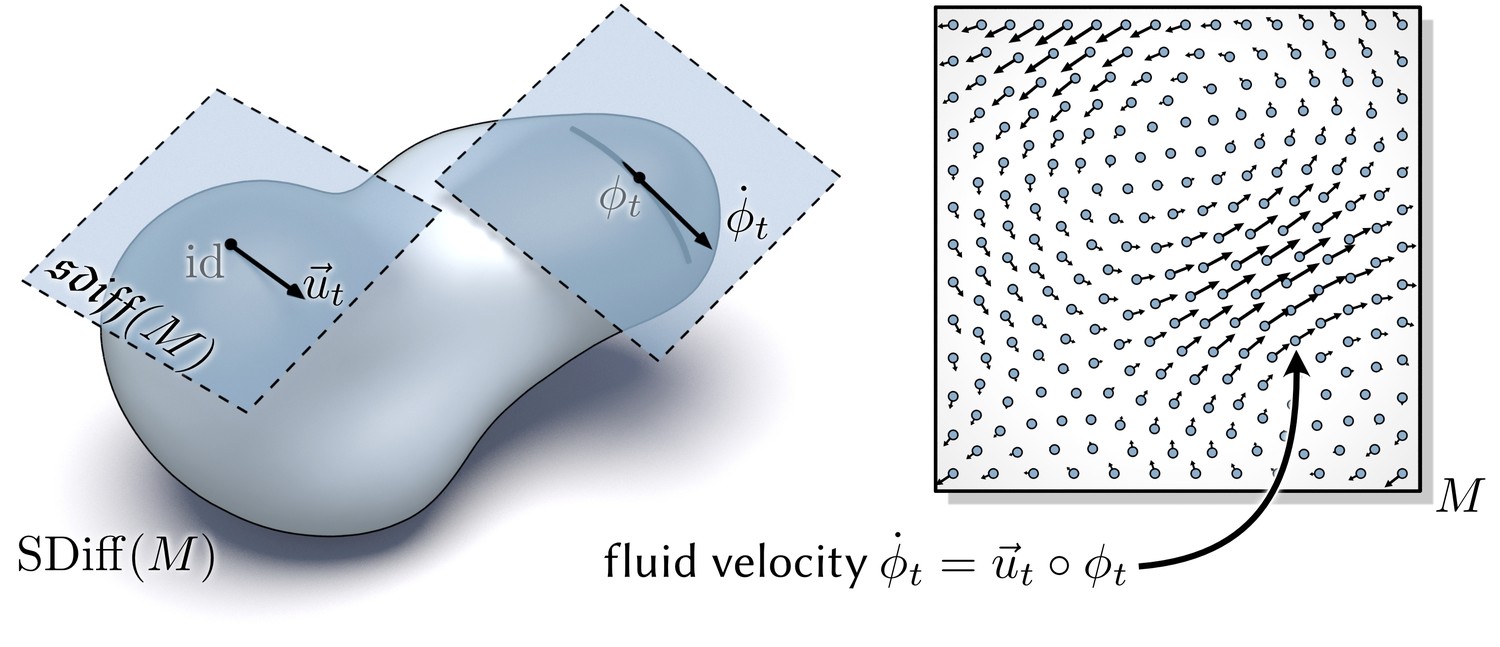}
    \caption{The geometry of the Lie group $\SDiff(M)$. Here, the Lagrangian velocity $\dot\varphi$ at $\varphi\in\SDiff(M)$ is right-translated to the identity $\id$, yielding a representation by the divergence-free Eulerian velocity $\vec u\in\sdiff(M)$. %
    }%
    \label{fig:DivFreeCartoon}
\end{figure}

\subsection{Fluids as Geodesics}
\label{sec:FluidsAsGeodesicEquation}
The variational formulation of the incompressible Euler equations follows Hamilton's standard least-action principle, also known as the stationary action principle, applied to the fluid body.  %
For free fluid motion, the action functional consists solely of the kinetic energy of the fluid, and the path of the least action corresponds to a geodesic in the space of flow maps.

\subsubsection{Fluid flows on a Riemannian manifold}
Recall that the kinetic energy of a vector field \(\vec u\in \sdiff(M)\) is defined as one-half of its \(L^2(M)\)-norm: 
\begin{align}
    {1\over 2}\|\vec u\|^2\coloneqq {1\over 2}\int_M |\vec u|^2 \Vol,
\end{align}
where $|\vec u|^2$ denotes the norm of $\vec u$ with respect to the metric on $M$.
This kinetic energy defines a Riemannian metric on the Lie group \(\SDiff(M)\) 
that is invariant under right translation, \ie,
\begin{align}
    \|\dot\phi\|^2\coloneqq \|\vec u\|^2,\quad\text{where \(\vec u\coloneqq\dot\phi\circ\phi^{-1}\) \,\,(\cf\@ \eqref{eq:RepresentPhiDotByU})}.
\end{align}
A geodesic in \(\SDiff(M)\) is a path \(\phi\colon [0,T]\to\SDiff(M)\) that is a stationary point of %
the squared metric length (\ie, action) functional
\begin{align}
\label{eq:ContinuousAction}
    \cS(\phi) \coloneqq \int_{0}^T {1\over 2}\|\dot\phi\|^2\, dt,
\end{align}
where the Lagrangian is simply the kinetic energy of the fluid.

\subsubsection{Variational characterization of incompressible Euler flows}
It turns out that stationary points of the action \eqref{eq:ContinuousAction} satisfy the incompressible Euler equation. 
To see this, consider a flow map \(\phi = \phi_{t,\epsilon} \subset \SDiff(M)\), depicted in \Cref{fig:VariationalCartoon}, that depends on both time \(t\) and a variational parameter \(\epsilon\).  Assume that $\varphi_{t,\epsilon}$ satisfies the fixed endpoint conditions $\varphi_{0,\epsilon} = \id$ and $\varphi_{T,\epsilon}=\varphi_{T,0}$ for all $\varepsilon$.

\begin{figure}[h]
    \centering
    \includegraphics[width=\linewidth]{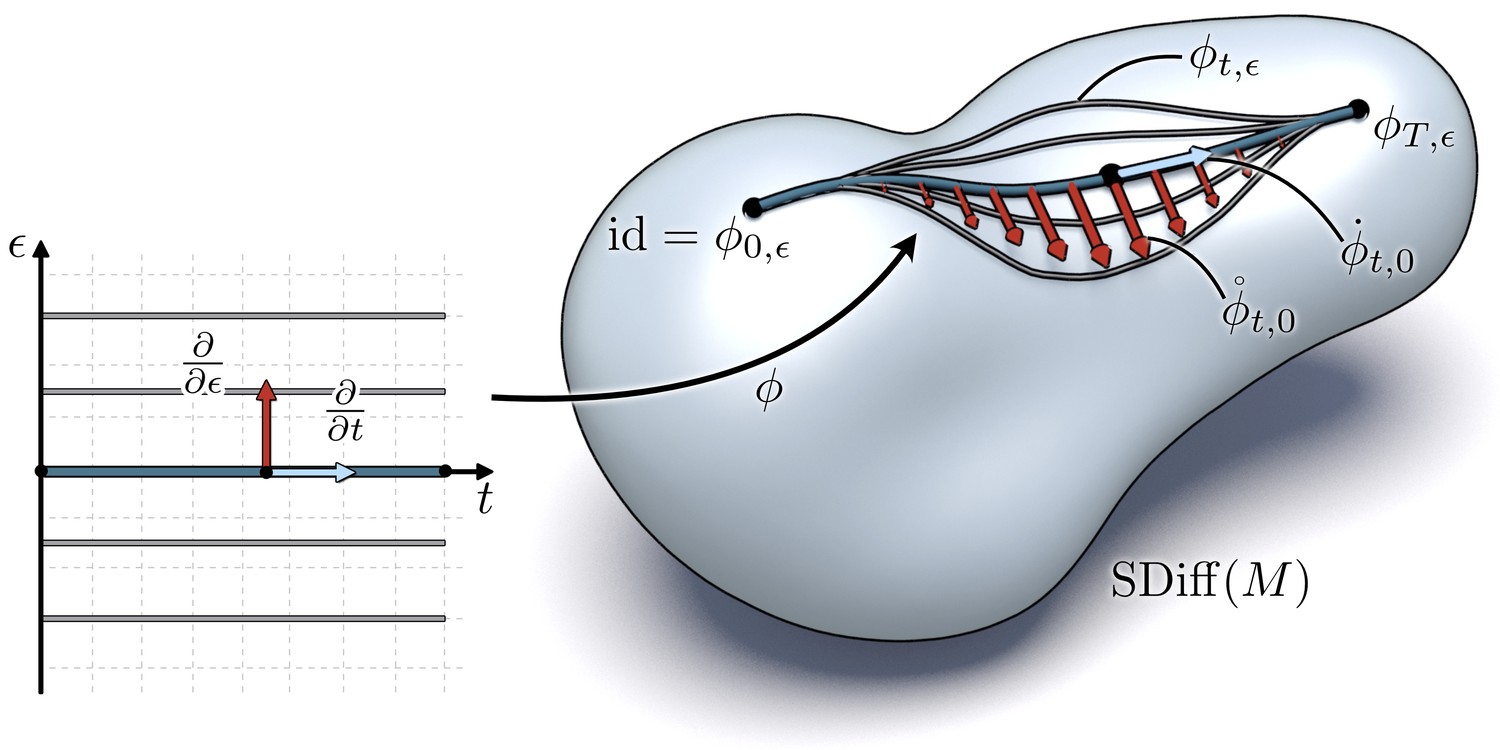}
    \caption{A depiction of the variational calculus on $\SDiff(M)$.  The flow map $\varphi=\varphi_{t,\varepsilon}$ depends on time $t\in[0,T]$ and a variational parameter $\epsilon\in[0,1]$. The $\SDiff(M)$  elements $\id = \varphi_{0,\epsilon}$ and $\varphi_{T,\varepsilon}$ denote fixed endpoints shared between all variations, while the vectors $\dot{\varphi}_{t,0},\mathring{\varphi}_{t,0}$ in the tangent space to $\SDiff(M)$ at $\varphi_{t,0}$ denote velocities of the flow map in the temporal and variational directions.} %
    \label{fig:VariationalCartoon}
\end{figure}
It is convenient to express the Euler--Lagrange equation for \(\cS\) in terms of velocity fields in Eulerian coordinates.  To that end, 
we express the partial derivatives
\(\dot\phi_{t,\epsilon}\coloneqq {\partial\over\partial t}\phi_{t,\epsilon}\) and
\(\bigmathring\phi_{t,\epsilon}\coloneqq {\partial\over\partial\epsilon}\phi_{t,\epsilon}\)
by Eulerian vector fields \(\vec u,\vec v\in \sdiff(M)\) via right translation as in \eqref{eq:RepresentPhiDotByU},
\begin{equation*}
\dot\phi_{t,\epsilon} = \vec u_{t,\epsilon}\circ\phi_{t,\epsilon},\quad
\bigmathring\phi_{t,\epsilon} = \vec v_{t,\epsilon}\circ\phi_{t,\epsilon}.
\end{equation*}
Computing the variation of the action \eqref{eq:ContinuousAction} with respect to the volume-preserving flow map $\varphi_{t,\epsilon}$ then yields the expression
\begin{equation}
	\label{eq:VariationWoLin}
    \bigmathring{\cS}=\left.{d\over d\epsilon}\right|_{\epsilon=0} \int_0^T{1\over 2}\|\dot\phi_{t,\epsilon}\|^2\, dt
    =\int_0^T\llangle\vec u,\bigmathring{\vec u}\rrangle\, dt,
\end{equation}
where \(\llangle\cdot,\cdot\rrangle\) denotes the \(L^2\) inner product between vector fields. 
Further manipulation of \eqref{eq:VariationWoLin} requires knowledge of the relationship between the vector fields associated with the \(t\) and \(\epsilon\) derivatives of the flow map $\phi_{t,\epsilon}$. Since the $t\epsilon$-coordinate vector fields \(\nicefrac{\partial}{\partial t}\) and \(\nicefrac{\partial}{\partial\epsilon}\) commute, 
it follows that
\begin{equation*}
\tfrac{\partial}{\partial\epsilon}\big(\vec u_{t,\epsilon}\circ\phi_{t,\epsilon}\big)
=
\tfrac{\partial}{\partial t}\big(\vec v_{t,\epsilon}\circ\phi_{t,\epsilon}\big).
\end{equation*}
This implies the Eulerian expression
\begin{equation*}
\bigmathring{\vec u} + \nabla_{\vec v}\vec u = \dot{\vec v} + \nabla_{\vec u}\vec v,
\end{equation*}
where \(\nabla\) is the Levi-Civita connection on the tangent bundle $TM$.
Since \(\nabla\) is torsion-free, the difference of directional derivatives \(\nabla_{\vec u}\vec v - \nabla_{\vec v}\vec u = [\vec u,\vec v]\) is the Lie bracket of the vector fields $\vec u$ and $\vec v$, yielding the relationship
\begin{equation}
\label{eq:ContinuousLinConstraint}
\bigmathring{\vec u} = \dot{\vec v} + [\vec u,\vec v].
\end{equation}
Equation \eqref{eq:ContinuousLinConstraint} is known as the \emph{Lin constraint}, and expresses the \(\epsilon\)-variation of the Eulerian velocity \(\vec u\) in terms of the time derivative of the Eulerian variation field \(\vec v\) and its Lie bracket with \(\vec u\).

Using \teqref{eq:ContinuousLinConstraint} and integrating by parts in time, 
the variation \eqref{eq:VariationWoLin} becomes
\begin{equation}
    \bigmathring{\cS}= \int_0^T\llangle \vec u^\flat\,|\,\dot{\vec v} + [\vec{u},\vec{v}]\rrangle\, dt = -\int_0^T\llangle \dot{\vec u}^\flat + \LD_{\vec u}\vec u^\flat \,|\,\vec v\rrangle\,dt,
    \label{eq:VariationOfContinuousAction}
\end{equation}
where it was used that \(\vec v_{0,\epsilon}\) and \(\vec v_{T,\epsilon}\) vanish for variations with fixed endpoints.
Here, \(\llangle\cdot|\cdot\rrangle\) denotes the dual pairing between 1-forms and vector fields: \(\llangle\alpha\,|\,\vec v\rrangle\coloneqq \int_M \langle \alpha\,|\,\vec v\rangle\,dV\).  The flat operator \(\flat\) maps a vector field \(\vec u\) to its corresponding 1-form \(\vec u^\flat\) obtained via Riesz representation, \ie, \(\langle\vec u^\flat \,|\,\vec v\rangle = \langle \vec u, \vec v\rangle\).
The operator \(\LD_{\vec u} = -[\vec u,\cdot]^*\) is the Lie derivative of 1-forms along the vector field \(\vec u\), which coincides with the negative adjoint of the Lie bracket 
on vector fields.

By the fundamental lemma of the calculus of variations \cite{huke1931historical}, a flow map is a stationary point of the action \eqref{eq:ContinuousAction} under 
volume-preserving variations with fixed endpoints if and only if \eqref{eq:VariationOfContinuousAction} vanishes for all divergence-free vector fields \(\vec v\in\sdiff(M)\). We conclude that the one-form $\dot{\vec u}^\flat + \LD_{\vec u}\vec u^\flat$
must vanish when paired with any divergence-free \(\vec v\in\ker(\div)\). Equivalently, its Riesz-associated vector field must lie in the \(L^2\)-orthogonal complement of \(\ker(\div)\), which is dual to the space of exact one-forms. In other words,
\begin{equation*}
\dot{\vec u}^\flat + \LD_{\vec u}\vec u^\flat \in \im(d).
\end{equation*}
Therefore, there exists a scalar function \(p_{\rm L}\), acting as a Lagrange multiplier enforcing incompressibility, such that the optimality condition for the least-action principle becomes the Euler equations in covector form
\begin{align}
\dot{\vec u}^\flat + \LD_{\vec u}\vec u^\flat = -\,dp_{\rm L}.
\end{align}
Using the fact that the Lie derivative \(\LD_{\vec u}\vec u^\flat = ( \nabla_{\vec u} \vec u + \frac{1}{2}\nabla|\vec u|^2)^\flat\), this is equivalent to the more familiar vector form of the incompressible Euler equations,
\begin{align}\label{eq:inc_euler}
\dot{\vec u} + \nabla_{\vec u}\vec u = -\nabla p,
\end{align}
where the pressure is given by \(p = p_{\rm L} + \tfrac12 |\vec u|^2\).

\section{Representation on half-densities}\label{sec:koopman}
Several options are available for 
discretizing the theory of \autoref{sec:FluidsAsGeodesicEquation}, each with its own benefits and drawbacks.
The present approach will implicitly express the flow maps of fluid motions as linear operators applied to a space of functions. The encoding of diffeomorphisms through linear operators on scalar functions is known as the \emph{Koopman representation}, which serves as a powerful tool for modeling and simulation.  It also appears in the functional maps methods in geometry processing \cite{Azencot:2013:OAT,Azencot:2015:DDV}, and is employed in the functional fluids method for fluid simulations \cite{Azencot:2014:FFS}.  

In this section, we introduce an adapted Koopman representation of the volume-preserving diffeomorphisms on the alternative space of \emph{half-densities}. After discretization, this has the advantage of yielding operators in the standard orthogonal group. 

\subsection{Koopman representation}
\label{sec:KoopmanRepresentation}
A (left) representation of a group \(G\) on a vector space \(V\) is a group homomorphism from \(G\) to \(\GL(V)\), the general linear group on \(V\). That is, a map \(\rho\colon G\to \GL(V)\) such that \(\rho(g_1g_2)=\rho(g_1)\rho(g_2)\) for all \(g_1,g_2\in G\). 
In the standard Koopman theory, the diffeomorphism $\phi\in\Diff(M)$ is represented via the linear pullback action $f\mapsto f\circ \phi^{-1}$ on smooth functions $f\in C^{\infty}(M)$, which is an element of the general linear group $\GL(C^{\infty}(M))$.  More precisely, the Koopman map
\[([\cdot]^{-1})^*\colon\Diff(M)\to \GL(C^{\infty}(M)), \quad (\phi^{-1})^*f = f\circ\phi^{-1},\] 
forms a left representation for the diffeomorphism group, since $((\psi\circ\phi)^{-1})^*f = (\psi^{-1})^*(\phi^{-1})^*f$ for smooth functions $f$ and diffeomorphisms $\phi,\psi$. 

\begin{remark}
    The map $f\mapsto f\circ\phi$ is an equally valid group action, leading to a right Koopman representation instead.
\end{remark}

While the setting just described is a perfectly valid way to encode diffeomorphisms as linear operators, it is advantageous 
to consider a representation on half-densities instead.  To that end, recall that the metric on an oriented Riemannian manifold $M$ induces a canonical volume form \(\Vol\).%
\footnote{The orientability assumption is not crucial to our theory. A non-orientable Riemannian manifold also gives rise to a canonical volume form \(\Vol\) regarded as a \emph{density} (a de Rham pseudo-differential form) rather than a differential form. 
} 
This defines an \(L^2\)-inner product on functions \(f,h\in C^\infty(M)\) through integration against $\Vol$,
\begin{equation*}
\llangle f,h \rrangle_{\Vol}=\int_M fh\Vol.
\end{equation*}
A half-density then satisfies the following definition.

\begin{definition}
On an oriented Riemannian manifold \(M\), a \emph{half-density} is an object with the local representation
\begin{equation*}
\tilde f = f\Vol^{\frac{1}{2}},
\quad f\in C^\infty(M),
\end{equation*}
where \(\Vol^{\frac{1}{2}}\) denotes a formal square root of the volume form, defined so that the product of two half-densities yields a density integrable over \(M\). 
\end{definition}
The space of half-densities, denoted \(\HD(M)\), is an infinite dimensional vector space similar to the smooth functions, but
carries a canonical pairing given by
\begin{equation}
	\label{eq:CanonicalPairing}
\llangle \tilde f,\tilde h \rrangle \coloneqq \int_M f h\Vol.
\end{equation}
Note that the canonical pairing of half-densities coincides with the \(L^2\)-inner product \(\llangle f,h \rrangle_{\Vol}\) of the associated functions, providing a useful link between these two distinct objects. In particular, \(\HD(M)\) admits a natural Hilbert space structure after \(L^2\)-completion without requiring an additional choice of inner product. Hence, the map 
\begin{equation}
	\label{eq:FunctionToHalfDensity}
	\Phi_{\Vol}\colon C^\infty(M)\to \HD(M),\quad f\mapsto f\Vol^{\frac{1}{2}} = \tilde f,
\end{equation}
sending a smooth function to its associated half-density, is an isometry.

\subsubsection{Linear operators on half-densities} 
Describing the proposed Koopman representation on half-densities requires understanding the linear transformations that act on them. 
The general linear group over \(\HD(M)\) is denoted by
\begin{equation*}
\GL(\HD(M)) \coloneqq
\big\{\tilde{\cR}\colon \HD(M)\xrightarrow{\rm linear}\HD(M)\mid \tilde{\cR} \text{ is invertible}\big\}.
\end{equation*}
The (sub-)group of linear isometries of \(\HD(M)\) preserving the canonical pairing is similarly
\begin{equation*}
\mathrm{O}(\HD(M)) \coloneqq
\big\{\tilde{\cR}\in\GL(\HD(M))\mid \llangle \tilde{\cR} f,\tilde{\cR} h\rrangle = \llangle f, h\rrangle\big\}.
\end{equation*}
Analogous to the Lie algebra \(\so(n)=\mathfrak{o}(n)\) of the group \(\SO(n)\) of special orthogonal transformations of \(\RR^n\), the Lie algebra \(\so(\HD(M))=\mathfrak{o}(\HD(M))\) of \(\O(\HD(M))\) consists of skew symmetric operators\footnote{
The distinction between $\SO(\HD(M))$ and $\O(\HD(M))$ in the infinite-dimensional setting is less apparent than in the finite-dimensional case of $\SO(n)\subsetneq \O(n)$, as there is no analogue of the determinant operation clearly separating $\O(\HD(M))$ into connected components \cite{Putnam:1952:OGH, Kuiper:1965:HTU}.  %
}.

\subsubsection{Koopman representation on half-densities}
Similar to the Koopman representation on functions, there is an analogous representation of the diffeomorphism group
\[([\cdot]^{-1})^*:\Diff(M)\to\O(\HD(M)),\quad \tilde{f}\mapsto (\phi^{-1})^*\tilde{f},\]
that acts on smooth half-densities $\tilde{f}\in\HD(M)$.
Using half-densities as the target space has the noteworthy advantage of representing diffeomorphisms $\phi\in\Diff(M)$ with \textit{orthogonal} operators in $\O(\HD(M))$. 
In particular, for any half-densities \(\tilde f,\tilde h\in\HD(M)\),
\begin{align*}
\llangle (\phi^{-1})^\ast\tilde f, (\phi^{-1})^\ast\tilde h \rrangle
&= \int_M (\phi^{-1})^\ast\tilde f\,\,(\phi^{-1})^\ast\tilde h \\
&= \int_M (f\circ\phi^{-1})(g\circ\phi^{-1})\det(d\phi^{-1})\Vol \\
&= \int_M (\phi^{-1})^\ast(fh\Vol) = \int_M fh\Vol
= \llangle \tilde f, \tilde h \rrangle,
\end{align*}
meaning that $(\phi^{-1})^*$ preserves the canonical pairing.  This differs from the usual case of smooth functions $C^{\infty}(M)$, where the Koopman representation yields only invertible operators in $\GL(C^{\infty}(M))$.

\subsection{Group and Algebra Representations}
\label{sec:GroupAndAlgebraRepresentations}

In the case of fluids, Koopman operators represent fluid motions through the advection of other mathematical objects.  To that end, we define the advection operator 
\begin{equation*}
	\Adv\colon \Diff(M)\to \O(\HD(M)),\quad \phi\mapsto \Adv_\phi,
\end{equation*}
which takes a 
diffeomorphism $\phi\in\Diff(M)$ 
to its Koopman representation on half-densities.  More precisely, for each $\phi\in\Diff(M)$ and $\tilde f\in\HD(M)$, 
\begin{equation*}
	\Adv_\phi \tilde f\coloneqq (\phi^{-1})^\ast\tilde f.
\end{equation*}
It follows immediately from the Koopman representation of functions that \(\Adv\) is a group homomorphism from diffeomorphisms to orthogonal operators on half-densities, \ie, for all \(\phi,\psi\in\Diff(M)\) it holds that 
\begin{equation*}
	\Adv_{\phi\circ\psi} = \Adv_{\phi}\circ\Adv_\psi.
\end{equation*}
The Lie group representation \(\Adv\) of $\Diff(M)$ induces a corresponding representation for the Lie algebra $\diff(M)$ through its differential at the identity: 
\begin{equation*}
    \adv \coloneqq (d\Adv)_{\id}\colon\diff(M)\xrightarrow{\rm hom}\so(\HD(M)), \quad \adv_{\vec u} = -\LD_{\vec u}. \label{eq:SmallAdvOperator}
\end{equation*}
That is, \(\adv_{\vec u}=-\LD_{\vec u}\) is the infinitesimal advection of half-densities $\tilde{f}\in\HD(M)$ along the vector field $\vec u$. 

\subsubsection{Making the advection explicit}
Just as the Lie derivative operators \(\LD_{\vec u}\) on \(C^\infty(M)\) are skew-adjoint with respect to the \(L^2\)-inner product for any vector field $\vec u \in \diff(M)$, the operators \(\adv_{\vec u}\) are skew-adjoint on half-densities \(\tilde f,\tilde h \in \HD(M)\) with respect to the canonical pairing.  To see this, recall the Leibniz rule for the Lie derivative,
\begin{equation*}
    \label{eq:LieDerLeibniz}
\LD_{\vec u}(\tilde f\tilde h) = (\LD_{\vec u}\tilde f)\,\tilde h + \tilde f\,(\LD_{\vec u}\tilde h).
\end{equation*}
Integrating over \(M\) and applying Cartan's formula $\LD_{\vec u} = d\iota_{\vec u} + \iota_{\vec u}d$ immediately yields 
\begin{align*}
\llangle \adv_{\vec u}\tilde f,\tilde h \rrangle + \llangle \tilde f,\adv_{\vec u}\tilde h \rrangle &= - \int_M \LD_{\vec u}(\tilde f\,\tilde h) \\
&= -\int_M d\iota_{\vec u}(fh\Vol) 
= -\int_{\partial M} fh\, \iota_{\vec u}\Vol,
\end{align*}
which vanishes, since either \(M\) has no boundary or $\vec u\in\diff(M)$ satisfies no-through boundary conditions. Hence,
\begin{equation*}
\llangle \adv_{\vec u}\tilde f,\tilde h \rrangle = -\,\llangle \tilde f,\adv_{\vec u}\tilde h \rrangle,
\end{equation*}
verifying that $\adv_{\vec u}\in\so(\HD(M))$ is skew-adjoint. 

In fact, it turns out that the infinitesimal advection 
is explicitly computable:
using the Leibniz rule and \(\LD_{\vec u}\Vol = \div(\vec u)\Vol\), observe that
\begin{equation*}
\LD_{\vec u}\tilde{f} = \LD_{\vec u}(f\Vol^{\frac{1}{2}}) = df(\vec u)\Vol^{\frac{1}{2}} + f\LD_{\vec u}\Vol^{\frac{1}{2}}.
\end{equation*}
Since differentiating $\Vol = (\Vol^{\frac{1}{2}})^2$ yields \(\LD_{\vec u}\Vol^{\frac{1}{2}} = \tfrac{1}{2}\div(\vec u)\Vol^{\frac{1}{2}}\), rearrangement in view of \(f\div(\vec u) = \div(f\vec u) - df(u)\) provides the concrete expression  
\begin{equation}\label{thm:AdvectHalfDensityFormula}
    \adv_{\vec u} \tilde{f} = -\LD_{\vec u}\tilde f = -\tfrac{1}{2}\big(df(\vec u) +  \div(f\vec u)\big)\Vol^{\frac{1}{2}}.
\end{equation}
In particular, when $\vec u\in\sdiff(M)$ is divergence free, it follows from \eqref{thm:AdvectHalfDensityFormula} that 
\[\adv_{\vec u}\tilde{f} = -df(\vec u)\Vol^{\frac{1}{2}} = (-\LD_{\vec u}f)\Vol^{\frac{1}{2}},\]
meaning half-densities are infinitesimally advected identically to scalar-valued functions.  

It will be important to note that the infinitesimal Koopman representation $\adv:\diff(M)\to\so(\HD(M))$ is not surjective, ie, not all skew-adjoint operators on half-densities are of the form $\adv_{\vec u}$ for some vector field $\vec u\in\diff(M)$.  This is illustrated by \Cref{fig:HDLieAlgebraCartoon} and the following example.

\begin{example}
\label{ex:NondiffeoSkewOperator}
Let \(\tilde f_1,\tilde f_2\in\HD(M)\) be two orthonormal half-densities, and define
\begin{equation*}
\cB(\tilde f)
\coloneqq
\llangle \tilde f,\tilde f_1\rrangle\,\tilde f_2
-
\llangle \tilde f,\tilde f_2\rrangle\,\tilde f_1.
\end{equation*}
Then \(\cB\) is skew-adjoint and hence lies in \(\so(\HD(M))\). However, \(\cB\) does
not arise from any vector field \(\vec u\in\diff(M)\), since operators of the form
\(-\LD_{\vec u}\) act locally on half-densities, whereas \(\cB\) depends on global inner
products and is therefore nonlocal.
\end{example}
\begin{figure}
    \centering
    \includegraphics[width=\linewidth]{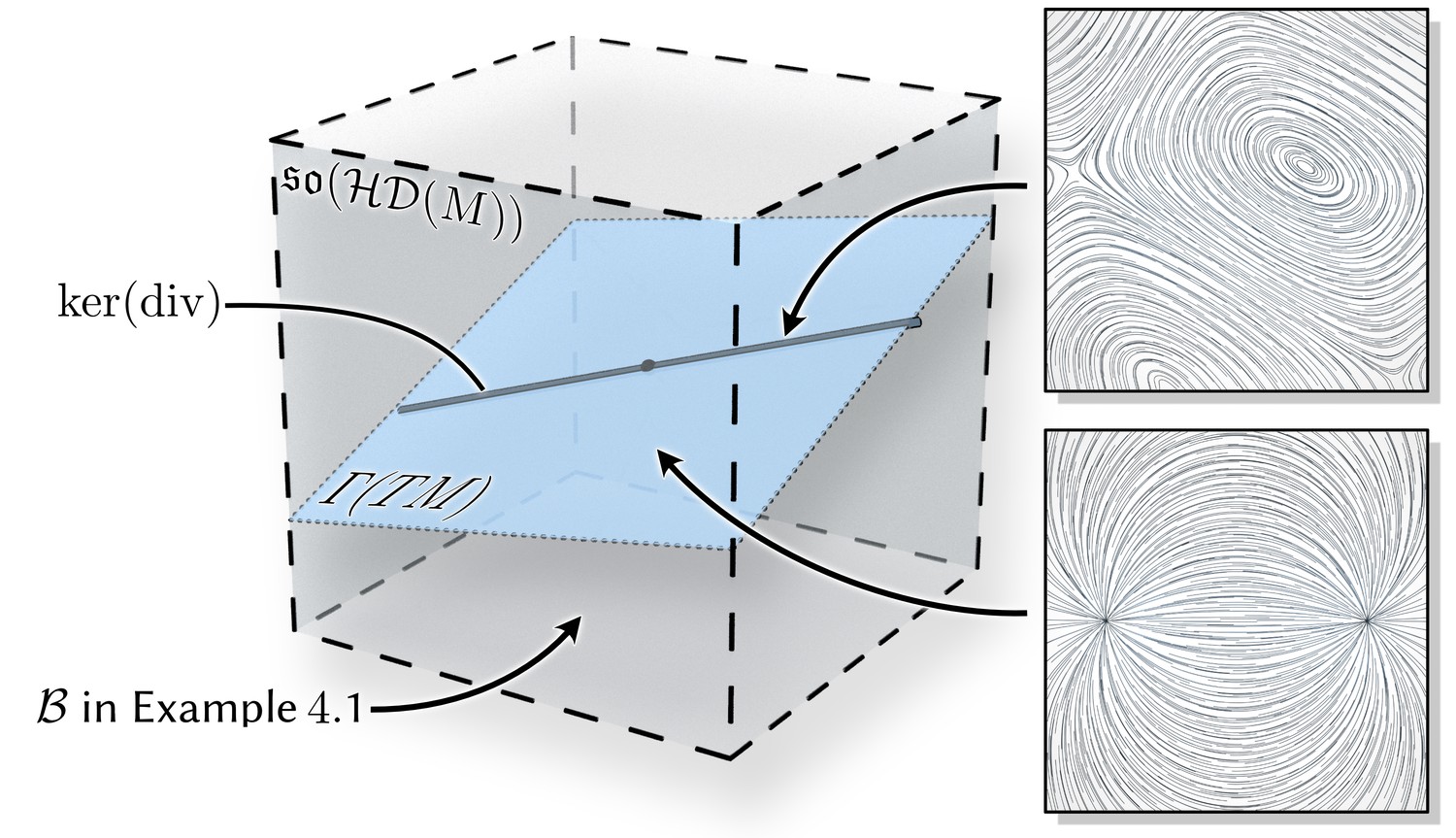}
    \caption{An illustration of the Lie algebra $\so(\HD(M))$.  Each vector field $\vec u\in\Gamma(TM)\cong \diff(M)$ corresponds to a skew-adjoint Lie derivative operator $\adv_{\vec u}=-\LD_{\vec u}\in\so(\HD(M))$, but not every skew-adjoint operator $\cB\in\so(\HD(M))$ takes this form. %
    }
    \label{fig:HDLieAlgebraCartoon}
\end{figure}

The fact that \(\im(\adv)\subsetneq\so(\HD(M))\) and only a small fraction of skew-adjoint operators represent fluid velocities is critical for formulating an appropriate discretization using Koopman representations.  We will next show that generating realistic fluid behavior in the discrete setting requires respecting the ``smallness'' of $\im(\adv)$, even if its homomorphism property cannot be exactly preserved.

\section{Spatial Discretization}
\label{sec:spatial-discretization}
Motivated by the previous section, we will  discretize fluid motions through their action on a finite-dimensional space of rotation matrices arising from a corresponding discretization of half-densities on $M$.  An appropriate least-action principle then provides a path through these matrices representing the incompressible Euler equation \eqref{eq:inc_euler}.  In this way, flow maps describing fluid motions become rotation matrices acting on (the coefficients of) discrete half-densities, and vector fields become skew-symmetric matrices which act infinitesimally.  Much like the image of $\adv$ is a proper subspace of the space $\so(\HD(M))$ of all skew-symmetric operators on half-densities, our finite-dimensional subspace representing velocity fields will also form a proper subspace of skew-symmetric matrices. But, unlike the infinite-dimensional setting, it will prove \textit{not} to be a Lie algebra. 

\subsection{The discrete function space}
First, the space of scalar functions \(C^\infty(M)\) is replaced by a finite dimensional vector space $\cF$ of dimension $F\in\mathbb{N}$, 
\begin{align}
    \cF \subset H^1(M)\subset L^2(M).
\end{align}
 Here, the Sobolev space \(H^1(M)\) contains $L^2(M)$ functions whose first (weak) derivatives also lie in \(L^2(M)\).  This choice is natural, as \(L^2(M)\) is the completion of \(C^\infty(M)\) under the \(L^2\) metric, and first derivatives are required to treat vector fields as discrete derivations. %
An example of \(\cF\) is the piecewise linear finite element space over a triangle mesh with $F$ vertices.

Letting \(\{\varphi_i\}_{i=1}^{F}\) denote a basis for \(\cF\), each function \(f\in \cF\) is represented by a vector \(\bff = (f_1,\ldots,f_{F})^\intercal \in \mathbb{R}^{F}\) of coefficients via its basis expansion
\begin{align}
    f(x) = \sum_{i=1}^{F} f_i \varphi_i(x).
\end{align}
The vector space \(\cF\), and its coefficient space $\RR^{F}$, inherit an inner product structure from the \(L^2\) metric on $M$.  Specifically, define the symmetric positive definite mass matrix \(\bM \in \RR^{F\times F}\) by
\begin{align}\label{eq:mass}
    [\bM]_{ij} = \llangle \varphi_i,\varphi_j\rrangle = \int_M\varphi_i\varphi_j \Vol.
\end{align}
Then, the inner product between finite-dimensional functions \(f,h\in \cF\) is given by
\begin{equation}
	\label{eq:DiscreteL2}
	\llangle f,h\rrangle = \bff^\intercal\bM\bh.	
\end{equation}

\subsection{Discrete half-densities}
Applying the map defined in \teqref{eq:FunctionToHalfDensity} to the basis $\{\varphi_i\}_{i=1}^{F}$ of \(\cF\), our discrete function space is mapped to a corresponding space $\tilde{\cF}$ of discrete half-densities  
with an associated basis \(\{\tilde\varphi_i \coloneqq \varphi_i\Vol^{\frac{1}{2}}\}_{i=1}^{F}\). Note that the coefficient vector \(\bff\) represents both \(f\in\cF\) and \(\tilde f\in \tilde{\cF} \) 
in their respective bases.

Analogous to the continuous setting, the canonical pairing (\teqref{eq:CanonicalPairing}) for discrete half-densities \(\tilde f, \tilde h\in \tilde{\cF}\) is given by 
\begin{equation*}
	\llangle\tilde f,\tilde h\rrangle = \llangle f,h\rrangle = \bff^\intercal\bM\bh.	
\end{equation*}
for \(\tilde f, \tilde h\in \tilde{\cF}\).  It will be useful to change bases for $\tilde{\cF}$ so that this pairing becomes simpler.
To that end, observe that any $\tilde{f}\in\tilde{\cF}$ has the expression
\[\tilde{f}(x) = \sum_{i=1}^{F} f_i\tilde{\varphi}_i(x) = \sum_{i,j=1}^{F} [\bM^{\frac{1}{2}}\bff]_i[\bM^{-\frac{1}{2}}]_{ij}\tilde{\varphi}_j(x) \eqqcolon \sum_{i=1}^{F}\tilde{f}_i\tilde{\Phi}_i(x),\]
in terms of the convenient change of coordinates
\begin{equation*}
	\bff\mapsto \bM^{\frac{1}{2}}\bff \eqqcolon \tbf, \quad \tilde{\varphi}_i \mapsto \sum_{j=1}^{F}[\bM^{-\frac{1}{2}}]_{ij}\tilde{\varphi}_j \eqqcolon \tilde{\Phi}_i.
\end{equation*}
Conveniently, this choice orthonormalizes the coordinates with respect to the new half-density basis \(\{\tilde{\Phi}_i\}_{i=1}^{F}\) and reduces the canonical pairing to the standard Euclidean dot product, \ie, 
\begin{equation*}
	\llangle\tilde f,\tilde h\rrangle = \sum_{i,j=1}^{F} \llangle\tilde{\Phi}_i,\tilde{\Phi}_j\rrangle \tilde{f}_i\tilde{h}_j = \sum_{i,j=1}^{F} \delta_{ij}\tilde{f}_i\tilde{h}_j =  \tbf^\intercal\tbh,	
\end{equation*}
where $\delta_{ij}$ denotes the Kronecker delta tensor which is one when $i=j$ and zero otherwise.

\subsection{Discrete flow maps} 
As mentioned previously, the goal is to use the Koopman representation, developed in the previous section, to implicitly discretize the flow maps of fluid motions by instead discretizing their action on half-densities.  To that end, 
consider $\SO(\tilde{\cF})$, the identity component of the finite-dimensional space of orthogonal transformations
\[\O(\tilde{\cF}) = \big\{\cR\in\GL(\tilde{\cF})\mid \llangle \cR \tilde{f},\cR \tilde{h}\rrangle = \llangle \tilde{f}, \tilde{h}\rrangle\big\}. \]
It is convenient to identify the space $\SO(\tilde{\cF})$ which acts on finite-dimensional half-densities with the space $\SO(F)$ acting on their coefficients. To accomplish this, consider the Lie group isomorphism defined by 
\begin{equation}\label{eq:discrete-flow-map-representation}
    \cR\mapsto \llangle \cR\tilde{\Phi}_i,\tilde{\Phi}_j\rrangle \eqqcolon [\bR]_{ij}.
\end{equation}
which sends a linear transformation $\cR\in\SO(\tilde{\cF})$ to its corresponding matrix representation $\bR\in\SO(F)$. This is valid since $\{\tilde{\Phi}_i\}$ is a basis for $\tilde{\cF}$, so that $\mathcal{R}\tilde{\Phi}_i = \sum_j\llangle\cR\tilde{\Phi}_i,\tilde{\Phi}_j\rrangle\tilde{\Phi}_j = \sum_j{\bR}_{ij}\tilde{\Phi}_j$ for the matrix $\bR$ of coefficients above.
It immediately follows from the definition of $\O(\tilde{\cF})$ that
\begin{align*}
    \llangle \cR\tilde{\Phi}_i,\cR\tilde{\Phi}_j \rrangle &= \sum_{k,l=1}^{F} [\bR]_{ik}[\bR]_{jl} \llangle \tilde{\Phi}_k,\tilde{\Phi}_l\rrangle \\
    &= \sum_{k=1}^F\bR_{ik}\bR_{jk} = \delta_{ij} = \llangle \tilde{\Phi}_i,\tilde{\Phi}_j \rrangle,
\end{align*}
and therefore $\bR\bR^\intercal = \bI$ as expected.  This ensures that $\SO(F)$ provides a meaningful representation space for fluid motions in terms of their action on the coefficients of discrete half-densities.

\subsection{Advection by discrete velocities}
Now that half-densities and flow maps have been discretized, it remains to
discretize the space of velocity fields and describe its action on discrete half-densities. To that end, the space of square-integrable vector fields \(L^2(M; TM)\) satisfying the no-flux boundary condition is replaced by a finite dimensional vector space 
\begin{equation*}
\cV \subset H_0(\div,M)\subset L^2(M; TM),   
\end{equation*}
where $H_0(\div, M)$ denotes the space of vector fields on $M$ with divergence lying in $L^2(M)$ and satisfying the no-flux boundary condition in the trace sense.  An example is given by the usual Raviart-Thomas finite element space with an essential condition enforcing tangentiality to the boundary~\cite{Raviart:1977:MFE}. 

Letting $V=\dim\cV$ and $\{\vec \psi_1, ..., \vec \psi_{V} \}$ be a basis for $\cV$, each velocity field \(\vec u \in \cV\) is represented by a vector \(\bu = (u_1,\dots,u_{V})^\intercal\in\RR^{V}\) via its basis expansion 
\begin{equation*}
    \vec u (x) = \sum_{k=1}^{V} u_k \vec \psi_k (x).
\end{equation*}
Similar to the case for discrete functions, the vector space \(\cV\) inherits an inner product structure from the metric on \(L^2(M; TM)\).  Defining the symmetric positive definite mass matrix \(\bK\in\RR^{V\times V}\) by
\begin{equation}\label{eq:stiffness}
    [\bK]_{kl} = \int_M \langle \vec \psi_k,\vec \psi_l \rangle\Vol,
\end{equation}
the inner product between vector fields \(\vec u,\vec v\in \cV\) is given by
\begin{equation*}
    \llangle \vec u, \vec v\rrangle = \bu^\intercal\bK\bv.
\end{equation*}

How should we represent the infinitesimal advection of half-densities by finite-dimensional vector fields $\vec u \in \cV$? One reasonable choice comes from the previous formula \eqref{thm:AdvectHalfDensityFormula} for the continuous infinitesimal advection $\adv_{\vec u} = -\LD_{\vec u}$.  Specifically, consider the mapping 
\begin{equation*}
    \cA\colon\cV\to\so(F), \quad \vec u \mapsto \llangle \tilde{\Phi}_i, \adv_{\vec u}\tilde{\Phi}_j\rrangle = \cA_{\vec u},
\end{equation*}

\begin{figure*}%
    \centering
    \includegraphics[width=\linewidth]{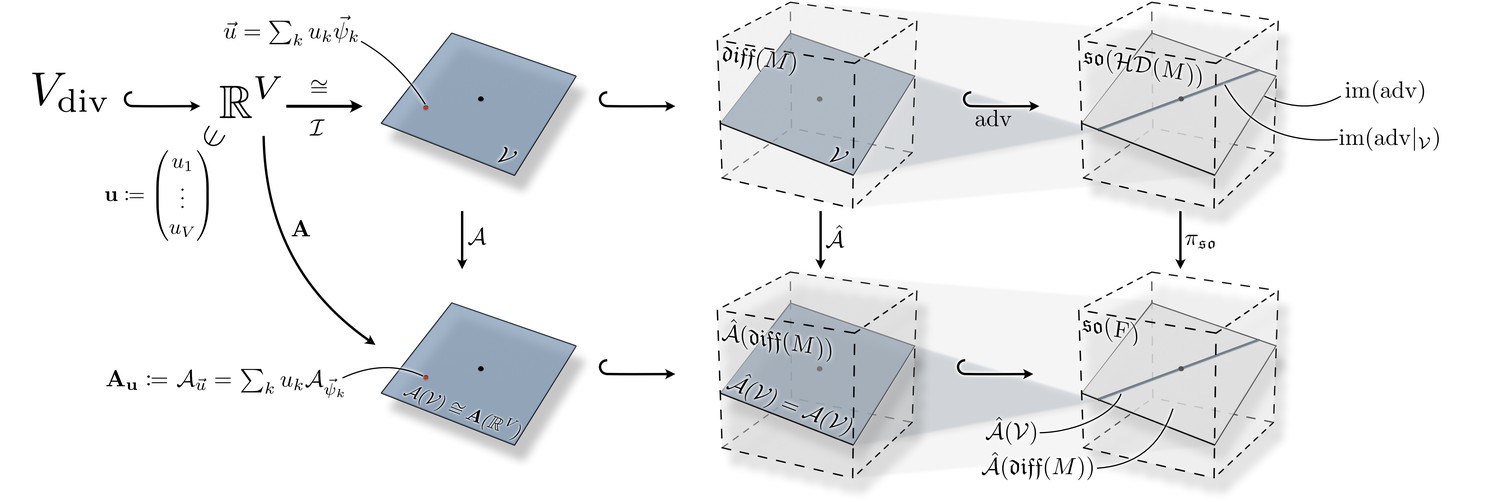}
    \caption{Discretizing the advection operator $\adv:\diff(M)\to\so(\HD(M))$.  A divergence-free field $\vec u\in\cV$ is identified with its coefficients $\bu\in V_{\div}$. The Galerkin projection $\cA_{\vec u}$ of $\adv_{\vec u}$ on the half-density basis defines the fully discrete advection map $\bu\mapsto \bA_{\bu}\in\so(F)$ via $\bA_{\bu} = \cA_{\vec u}$. %
    }
    \label{fig:ContVsDiscretedouble}
\end{figure*}
which sends a vector field $\vec u$ to the Galerkin projection of its negative Lie derivative on the chosen basis $\{\tilde{\Phi}_i\}$ for discrete half-densities.  Similar to the case of $\cF$ before, it is convenient to identify the function space $\cV\cong\mathbb{R}^{V}$ with its coefficients in the $\{\vec\psi_k\}$ basis, so that $\adv_{\vec u}$ is given the alternative representation
\[\bA\colon\RR^{V}\to\so(F), \quad \bu \mapsto \llangle \tilde{\Phi}_i, \adv_{\vec u}\tilde{\Phi}_j\rrangle = \bA_\bu.\]
To compute a useful expression for this mapping, first observe that the right-hand side expands to
\[ -\llangle \tilde{\Phi}_i, \LD_{\vec u}\tilde{\Phi}_j\rrangle = -\sum_{k,l=1}^{V} [\bM^{-\frac{1}{2}}]_{ik}\,\llangle \tilde{\varphi}_{k}, \LD_{\vec u}\tilde{\varphi}_l \rrangle\,[\bM^{-\frac{1}{2}}]_{lj}. \]
Applying \eqref{thm:AdvectHalfDensityFormula} along with integration by parts, the inner pairing simplifies to 
\begin{align*}
    \llangle \tilde{\varphi}_{k}, \LD_{\vec u}\tilde{\varphi}_l \rrangle &= \frac{1}{2}\int_M \varphi_k \big(d\varphi_l(\vec u) + \div(\varphi_l\vec u)\big) \Vol \\
    &= \frac{1}{2}\int_M \big(\varphi_{k}\,d\varphi_l(\vec u) - \varphi_l\,d\varphi_{k}(\vec u)\big)\Vol,
\end{align*}
where the boundary term vanishes since $\vec u \in \diff(M)$ is tangent to $\partial M$.
Putting this together, the discrete advection operator $\bA\colon\RR^{V}\to\so(F)$ takes the computable form 
\begin{equation}\label{eq:discrete-adv-operator}
    \bu\mapsto\bA_\bu, \quad \bA_\bu \coloneqq \sum_{k=1}^{V} u_k \bA_k, 
\end{equation}
in terms of the basis-dependent tensors
\begin{align*}
    \left[\bC_{k}\right]_{ij} &\coloneqq \int_M \varphi_i\, d\varphi_j(\vec \psi_k)\Vol, \\
    {\bA}_{k} &\coloneqq \cA_{\vec{\psi}_k} =  -\tfrac{1}{2}\,\bM^{-\frac{1}{2}}
    \big(\bC_{k}-\bC_{k}^\intercal\big)
    \bM^{-\frac{1}{2}}.
\end{align*}
\begin{wrapfigure}[9]{r}{0.275\columnwidth}
    \centering
    \vspace{0.5\baselineskip}
{\hspace{-1.5em}\includegraphics[width=0.325\columnwidth]{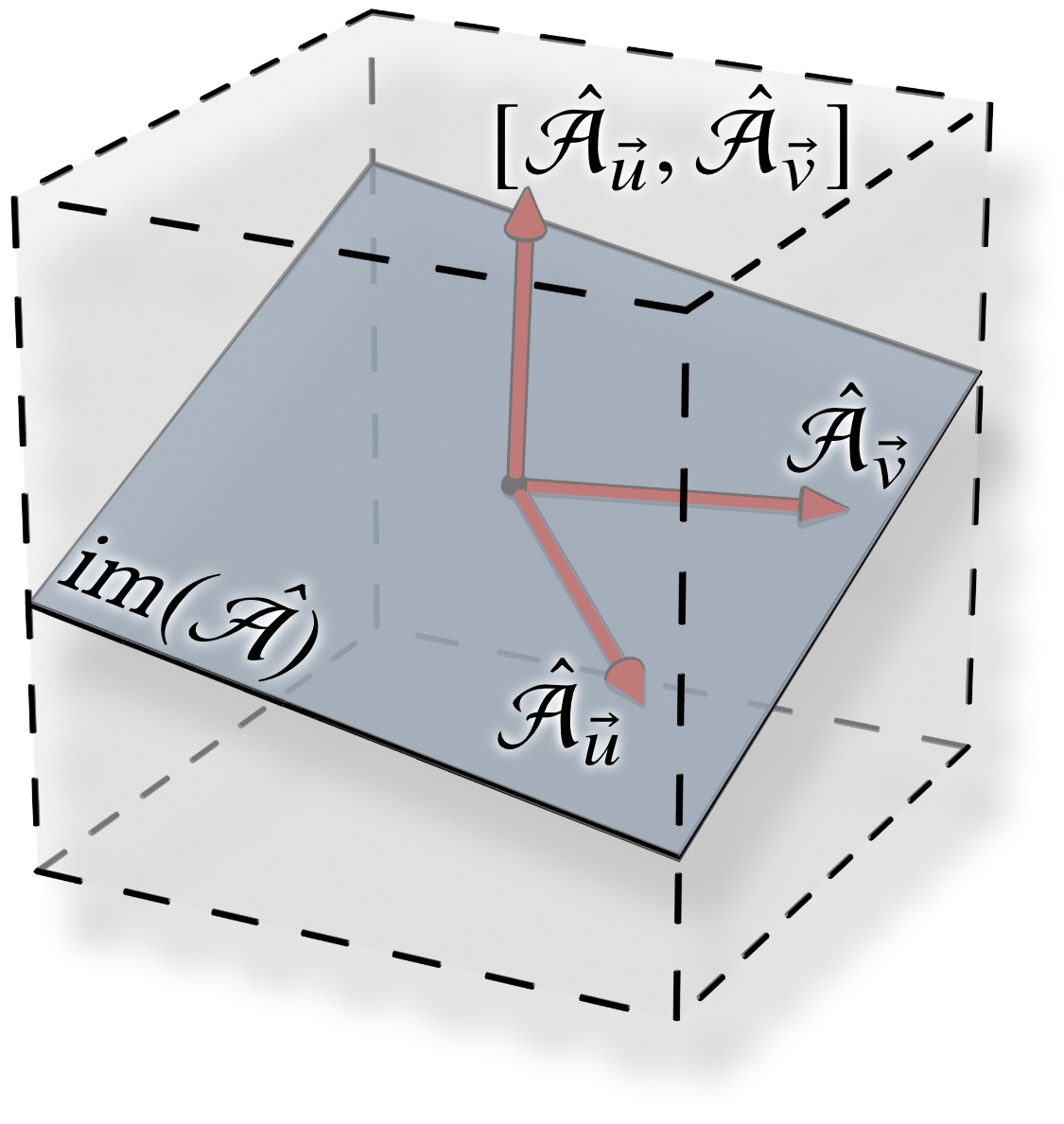}}
\end{wrapfigure}

\begin{table}%
    \centering
    \caption{Discrete and continuous notations. 
    }
    \scriptsize
    \begin{tabular}{lll}
    \toprule
    Continuous Theory & Finite Dimensional Theory & Discrete Theory \\
    \midrule
        \(f\in C^\infty(M)\) & $\sum_i f_i\varphi_i \in \cF$ & \(\bff\in \RR^{F}\) \\
        \(\tilde{f}\in\HD(M)\) & $\sum_i \tilde{f}_i\tilde{\Phi}_i\in\tilde{\cF}$ & \(\tbf\in\RR^{F}\) \\
        \(\Adv_{\phi}\in\O(\HD(M))\) & ${\cR}\in\SO(\tilde{\cF})$ & \(\bR\in\SO(F)\) \\
        \(\vec u \in \diff(M)\) & $\sum_ku_k\vec{\psi}_k\in\cV$ & \(\bu\in \RR^{V}\) \\ 
        $\vec u\in\sdiff(M)$ & $\sum_k u_k\vec\psi_k\in\cV_{\div}$ & $\bu\in V_{\div}\subset \mathbb{R}^{V}$ \\
        \( \adv:\diff(M)\to\so(\HD(M))\) & \(\cA:\cV\to\so(F)\) & \(\bA:\RR^{V}\to\so(F)\) \\
        \( \adv:\sdiff(M)\to\so(\HD(M))\) & \(\bar\cA:\cV_{\div}\to\so(F)\) & \(\bar\bA:V_{\div}\to\so(F)\) \\
         \bottomrule
    \end{tabular}
    \label{tab:Notation}
\end{table}

\begin{remark}\label{rem:nonholonomic}
Similar to the continuous case, one should note that \(\im(\bA) \subsetneq \so(F)\) is a relatively small subset of the ambient Lie algebra. 
Unlike the continuous case, however, the subspace \(\im(\bA)\) is \textit{not} a Lie subalgebra, and, as such, the map \(\bA\) is no longer a Lie algebra homomorphism.  In fact, part of this is a property of discretizing the half-density space $\HD(M)$ (or the function space $C^{\infty}(M)$).  Consider the map $\hat{\cA}\colon\diff(M)\to\so(F)$, analogous to $\cA\colon\cV\to\so(F)$, satisfying \([\hat\cA(\vec u)]_{ij}\coloneqq \llangle\tilde\Phi_i,\adv_{\vec u}\tilde\Phi_j\rrangle\) for each $\vec u \in \diff(M)$.  Then, $\im(\hat{\cA})\subset \so(F)$ is still not a Lie subalgebra, meaning that the composition of $\bar{\cA}$ and the inclusion $\RR^{V}\cong \cV\hookrightarrow \diff(M)$ comprising $\bA$ is not either.  This troublesome fact is further demonstrated in \Cref{ex:LieNonsubalgebra}. 
\end{remark}
\begin{example}[Not a Lie algebra homomorphism]
\label{ex:LieNonsubalgebra}
    Consider two linearly independent vectors $\vec u,\vec v\in\cV$.  It follows that 
    \begin{align*}
        &[\cA_{\vec u},\cA_{\vec v}] = \cA_{\vec u}\cA_{\vec v}-\cA_{\vec v}\cA_{\vec u} \\
        &= \sum_{s} \llangle\tilde{\Phi}_i,\adv_{\vec u}\tilde{\Phi}_s\rrangle\llangle\tilde{\Phi}_s,\adv_{\vec v}\tilde{\Phi}_j\rrangle - \llangle\tilde{\Phi}_i,\adv_{\vec v}\tilde{\Phi}_s\rrangle\llangle\tilde{\Phi}_s,\adv_{\vec u}\tilde{\Phi}_j\rrangle \\
        &= \sum_{s}  \llangle\adv_{\vec v}\tilde{\Phi}_i,\tilde{\Phi}_s\rrangle\llangle\tilde{\Phi}_s,\adv_{\vec u}\tilde{\Phi}_j\rrangle - \llangle\adv_{\vec u}\tilde{\Phi}_i,\tilde{\Phi}_s\rrangle\llangle\tilde{\Phi}_s,\adv_{\vec v}\tilde{\Phi}_j\rrangle \\
        &\neq \llangle \adv_{\vec v}\tilde{\Phi}_i,\adv_{\vec u}\tilde{\Phi}_j\rrangle - \llangle \adv_{\vec u}\tilde{\Phi}_i,\adv_{\vec v}\tilde{\Phi}_j  \rrangle \\
        &= \llangle \tilde{\Phi}_i, [\adv_{\vec u},\adv_{\vec v}]\tilde{\Phi}_j \rrangle = \llangle\tilde{\Phi}_i, \adv_{[\vec{u},\vec{v}]}\tilde{\Phi}_j \rrangle = \cA_{[\vec u,\vec v]},
    \end{align*} 
    where the two separate integrations in $[\cA_{\vec u},\cA_{\vec v}]$ are not equivalent to the one in $\cA_{[\vec u,\vec v]}.$ Therefore,
    $\bA: \mathbb{R}^V\to\so(F)$ is not a Lie algebra homomorphism.
\end{example}

Note that this discussion of advection and its discretization has been general to the entire diffeomorphism group $\Diff(M)$ and its Lie algebra $\diff(M)$.  However, incompressible fluids are restricted to the proper subspace $\SDiff(M)$ and the associated Lie algebra $\sdiff(M)$ of divergence-free vector fields.  To that end, we also consider a finite-dimensional space of incompressible vector fields 
\[\cV_{\div} \coloneqq \{\vec u \in \cV \mid \div\vec u = 0\} \subsetneq \cV,\]
which is contained in the finite-dimensional space $\cV$ of vector fields previously defined.
Identifying $\cV_{\div}$ with set of coefficients similarly to before, a vector field $\vec u\in \cV_{\div}$ can be represented as a vector $\bu\in V_{\div} \subsetneq\mathbb{R}^{V}$. The corresponding advection operators $\bar{\cA}:\cV_{\div}\to \so(F)$ and $\bar{\bA}:V_{\div}\to\so(F)$ are then defined through simple restriction,
\[\bar{\cA} = \cA|_{\cV_{\div}}, \qquad \bar{\bA} = \bA|_{V_{\div}}.\]
For convenience, \Cref{tab:Notation} provides a dictionary between the continuous objects of relevance and their discretizations in the present framework.  The complementary diagram \Cref{fig:ContVsDiscretedouble} displays the relationship between the continuous Koopman representation $\adv:\diff(M)\to\so(\HD(M))$ and its analogy $\bu \mapsto \bA_{\bu}$ in the discrete setting.

\section{Equations of motion}
\label{sec:EquationsOfMotion}
Having established a representation for the configuration space \(\SDiff(M)\) of incompressible fluid motion on the special orthogonal operators in \(\SO(\HD(M))\), it is now prudent to derive the equations of motion governing the discretized fluid. 
It turns out that these vakonomic equations are naturally posed on a \textit{sub-Riemannian} manifold, where the metric is only defined on a part of each tangent space.  These equations admit a variational interpretation leading to a Lax system with provable conservation properties. 

\subsection{Sub-Riemannian structure}
\label{sec:SubRiemannianStructure}

To discuss the variational setting we consider, first observe that the discrete mass matrix $\bK$ on the velocity space $\cV\cong\RR^V$ induces an inner product on the image of the discrete infinitesimal advection operator
\[
\im(\bA)\subsetneq\so(F)=T_{\id}\SO(F),
\]
by identifying elements of $\im(\bA)$ with their corresponding velocity
coefficients.  Explicitly, for $\bu,\bv\in V$, we define the inner product on $\im({\bA})$ through
\begin{equation}
    \label{eq:inner-product-at-identity}
    \langle {\bA}_{\bu}, {\bA}_{\bv} \rangle_{\bK}
    \;\coloneqq\;
    \bu^\intercal \bK \,\bv,
\end{equation}
so that the discrete $\adv$ map
$\bu \mapsto {\bA}_{\bu}$ is an isometry onto its image with respect to $\bK$.  
\begin{remark}
    Note that the $\bK$-metric $\langle\cdot,\cdot\rangle_{\bK}$ is generally quite different from the Frobenius metric $\langle\cdot,\cdot\rangle_{\rm Frob}$. In particular, orthogonal vector fields in $\cV_{\div}$ with coefficients $\bu,\bv\in V_{\div}$ will generally satisfy $\langle\bar\bA_{\bu},\bar\bA_{\bv}\rangle_{\bK}=0$ but $\langle\bar{\bA}_{\bu},\bar{\bA}_{\bv}\rangle_{\rm Frob}\neq 0$.  
\end{remark}

To set up a variational problem on $V_{\div}$, we require access to the $\bK$-metric on any tangent space $T_{\bR}\SO(F)$ to the discrete group.  To that end, consider the set of matrices \(\bX\in\im(\bar\bA)\) 
right-translated by flow map elements \(\bR\in\SO(F)\), \ie,
\begin{equation}
    \cD_{\bR}\coloneqq \{\bX\bR \mid \bX \in \im(\bar\bA)\}\subset T_{\bR}\SO(F). 
\end{equation}
Since the right translation \(\bX\mapsto \bX\bR\) is smooth, this defines a smooth \emph{distribution} $\cD \subset T\SO(F)$ properly contained in the tangent bundle to $\SO(F)$ and illustrated in \Cref{fig:ImADistribution}. This enables an extension of the inner product \(\langle\cdot,\cdot\rangle_\bK\) on \(\im(\bar\bA)\) to each fiber \(\cD_{\bR}\) using right invariance:
for all \(\bX\bR,\bY\bR\in\cD_{\bR}\),
\begin{equation}
    \langle \bX\bR,\bY\bR\rangle_{\bK}\coloneqq\langle \bX,\bY\rangle_{\bK}.
\end{equation}

\begin{figure}%
    \centering
    \includegraphics[width=\linewidth]{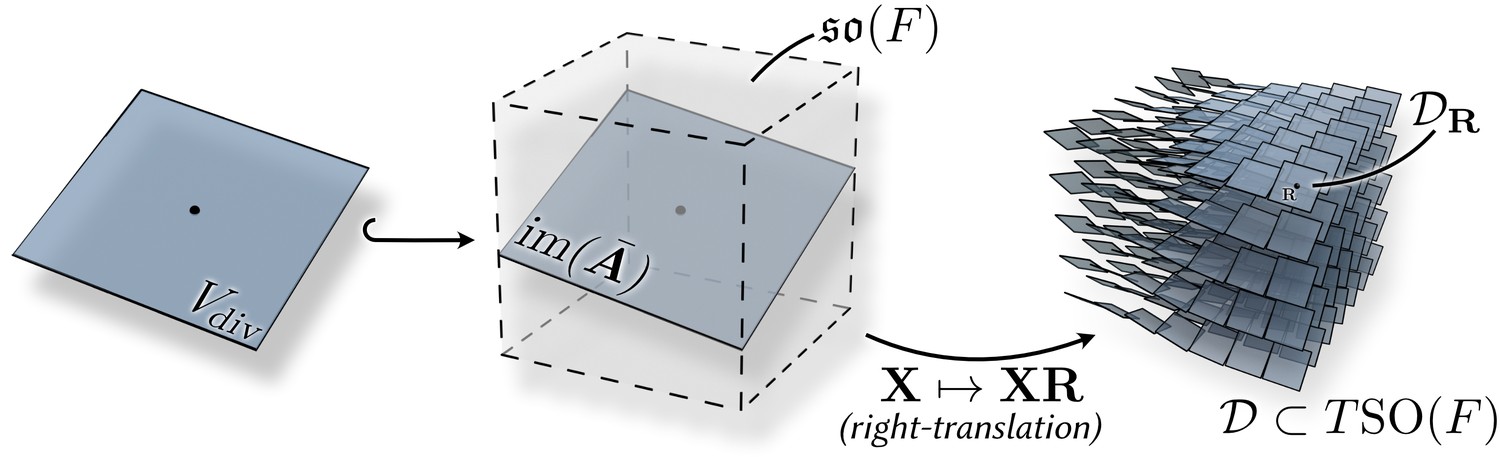}
    \caption{The nonintegrable distribution $\cD$ defined by $\im(\bar{\bA})\subset\so(F).$ Right-translation of $\bX\in\im(\bar\bA)$ yields elements $\bX\bR\in\cD_{\bR}$ of the partial tangent space $\cD_{\bR}\subset T_{\bR}SO(F)$ forming the distributional fiber at $\bR\in\SO(F)$. %
    }
    \label{fig:ImADistribution}
\end{figure}

Together, the right-invariant distribution \(\cD\) and the right-invariant inner product $\langle\cdot,\cdot\rangle_{\bK}$ defined above endow the Lie group \(\SO(F)\) with a \emph{sub-Riemannian} (or Carnot--Carath\'eodory) metric \cite{Gromov:1981:SMV} whose admissible directions are constrained to lie in \(\cD\). 
The triple \((\SO(F),\cD,\langle\cdot,\cdot\rangle_\bK)\) forms a so-called \emph{sub-Riemannian manifold} \cite{Montgomery:2002:TSG}. Accordingly, notions such as the distance and
energy of curves associated to this metric are defined only through admissible paths, \ie, curves whose tangent vectors lie in the distribution \(\cD\) at every point.

\subsubsection{Flat and Sharp}
\label{sec:MusicalIsomorphismsInSubRiemannianSetup}

It remains to specify how $\bK$ restricts 
to the Lie algebra $\sdiff(M)$ of divergence-free vector fields.
Consider the inclusion of $V_{\div}$ in $\mathbb{R}^V$ given by 
\(\iota\colon\rV_{\div}\hookrightarrow\RR^{V}\).
Its dual map is the restriction
\begin{equation*}
\iota^*\colon(\RR^{V})^*\to\rV_{\div}^*,\qquad \iota^*(\hat\bu)=\hat\bu|_{\rV_{\div}},
\end{equation*}
which restricts a covector $\hat{\bu}$ on the ambient coefficient space to act only on the coefficients corresponding to divergence-free vector fields.
The kernel of this restriction operation is equal to the annihilator
\begin{equation*}
\rV_{\div}^\circ\coloneqq \ker(\iota^*)=\{\hat\bu\in(\RR^{V})^*\mid \hat\bu(\bu)=0\ \forall \bu\in\rV_{\div}\},
\end{equation*}
containing those covectors which are zero on all of $V_{\div}$. It follows that the image covector space
\begin{equation*}
\rV_{\div}^* \cong (\RR^{V})^* \big/\rV_{\div}^\circ,
\end{equation*}
is isomorphic to the domain modulo the kernel of $i^*$, providing a discrete analogue of the statement 
\begin{equation*}
\sdiff(M)^\ast \cong \Omega^1(M)\big/d\Omega^0(M).
\end{equation*}

Consider now the restricted advection operator $\bar{\bA} \coloneqq \bA\circ\iota \colon V_{\div}\to\so(F)$.  An identical argument to the above yields the equivalence
\begin{equation*}
\im(\bar\bA)^* \cong \so(F)^* \big/\im(\bar\bA)^\circ,
\end{equation*}
so that covectors dual to elements in \(\im(\bar\bA)\) may be represented by cosets \(\bZ + \im(\bar{\bA})^\circ\subset\so(F)^*\) in the dual space, defined only up to the addition of elements in \(\im(\bar\bA)^\circ\).

The discrete analogue of pressure projection on the covector side is restriction to admissible directions via the adjoint map
\begin{equation*}
\bar\bA^*\colon \so(F)^*\to\rV_{\div}^*, \quad \bZ\mapsto \bar{\bA}^*\bZ.
\end{equation*}

\begin{figure}%
    \centering
    \includegraphics[width=\linewidth]{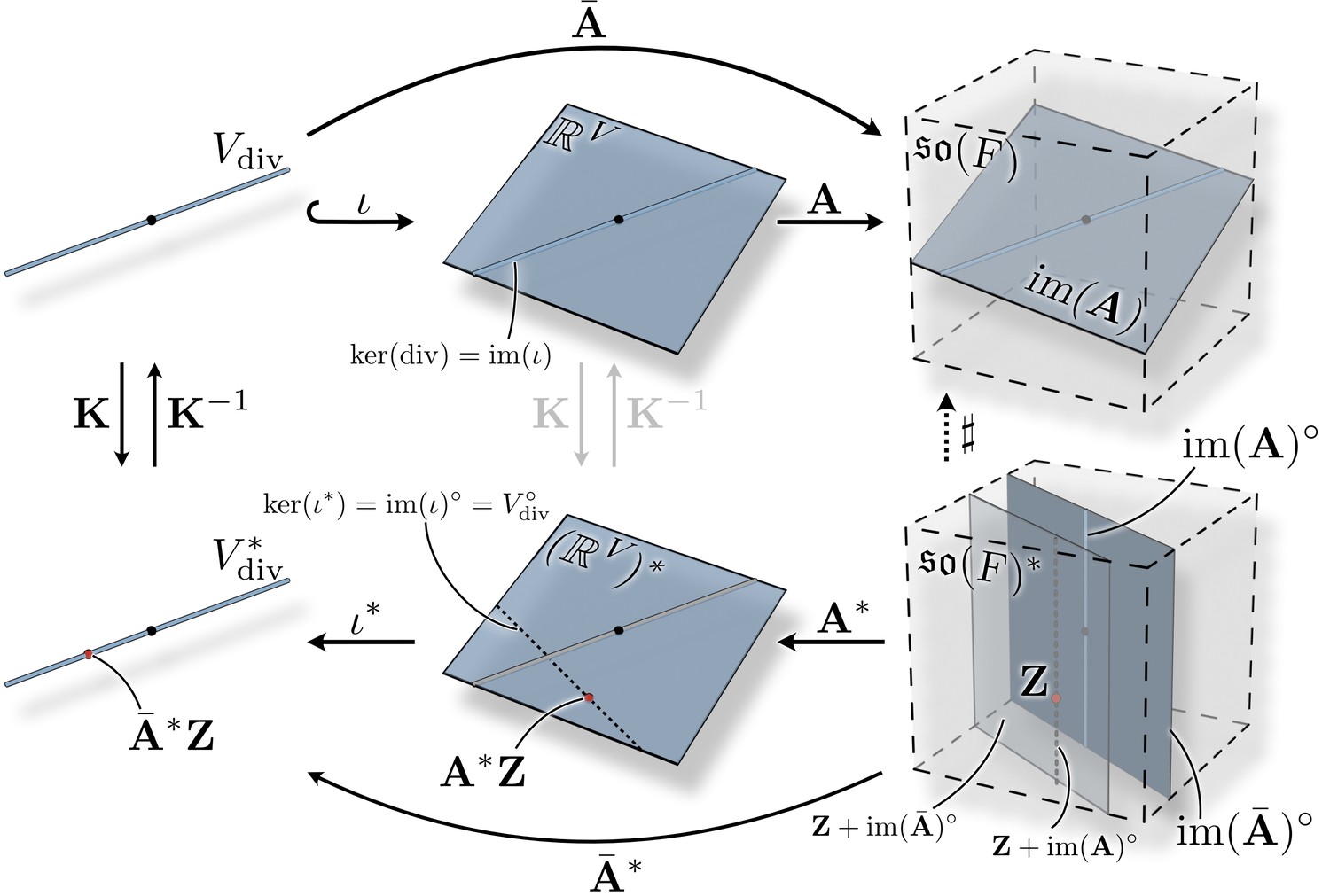}
    \caption{The relationship between the discrete advection $\bar{\bA}$ and its adjoint $\bar{\bA}^*$ as a commutative diagram.  Any coefficient vector $\bu\in V_{\div}$ maps to a unique skew-adjoint operator $\bar{\bA}_{\bu}\in\so(F)$. Similarly, any coset $\bY = \bZ+\im(\bar\bA)^\circ\subset\so(F)^*$ represented by $\bZ\in \so(F)^*/\im(\bar{\bA})^\circ$ maps to a unique element $\bar{\bA}^*\bZ\in V_{\div}^*$. These mappings are related by the sharp operator $\sharp:\so(F)^*\to\im\bar{\bA}$, which maps a coset $\bY$ to the unique $\bA_{\bv}$ associated to $\bv = \bK^{-1}\bar{\bA}^*\bZ$.}
    \label{fig:PressureProjection}
\end{figure}

This is most easily understood as a two step process using $\bar{\bA}^* = \iota^*\circ\bA^*$ (see \Cref{fig:PressureProjection}).  First, $\bA^*$ maps the coset representative $\bZ\in\so(F)^*$ to its image $\bA^*\bZ$ in the dual coefficient space $({\RR^V})^*$. This image is itself part of a coset, \ie, $\bA^*\bZ$ represents an equivalence class defined up to the addition of elements in $V_{\div}^\circ$.  The adjoint $\iota^*(\bA^*\bZ)$ then collapses this class to a covector $\bar{\bA}^*\bZ$ on the space $V^*_{\div}$ where the $\bK$-metric is defined.  Finally, this can be used to define a sharp map in the dual space $\so(F)^*$:
\begin{equation*}
\sharp \coloneqq \bar\bA\bK^{-1}\bar\bA^\ast\colon \so(F)^*\to\im(\bar{\bA})\subsetneq \so(F),
\end{equation*}
where \(\bZ^\sharp\in\im(\bar\bA)\) depends only on the class of \(\bZ\) modulo \(\im(\bar\bA)^\circ\).  Notice that there is no corresponding flat $\flat$ on $\so(F)$, since $\sharp$ is only injective on cosets $\bZ + \im(\bar{\bA})^\circ \subset \so(F)^*$.

Notice that this sharp map also contains the familiar pressure projection used in fluid simulation.  To see this, recall that the metric $\bK$ on $\mathbb{R}^V$ induces a Riesz isomorphism 
between the velocity coefficient space and its dual.  Then, for any $\bu\in V_{\div}$ and any $\bv\in\mathbb{R}^V$ it follows that
\[ \langle\bK\bv\,|\,i(\bu)\rangle = \bv^\intercal\bK\,i(\bu) = (i^\dagger\bv)^\intercal\bK\bu = \llangle \vec w, \vec u\rrangle, \]
where $i^\dagger:\mathbb{R}^V\to V_{\div}$ is the $\bK$-metric adjoint of the linear map $i$ and 
$\vec w\in \cV_{\div}$ is the div-free vector field with coefficients $\bw = i^\dagger\bv$. It follows from $L^2$-orthogonality that $i^\dagger\bv = i^\dagger(\bv + \ba)$ for any coefficient vector $\ba$ corresponding to a gradient field $\nabla f\in\cV$, and therefore $i^\dagger$ is precisely the $\bK$-orthogonal pressure projection on velocity coefficients.  Given a discrete divergence operator $\bD$ whose domain is $\mathbb{R}^V$ and considering its metric pseudoinverse $\bD^+$, this works out to 
\[i^\dagger\bv = (\bI-\bD^+ \bD)\bv = \bv-\bK^{-1}\bD^\intercal(\bD\bK^{-1}\bD^\intercal)^{-1}\bD\bv.\]

\subsection{Vakonomic variational principle}
\label{sec:VakonomicVariationalPrinciple}
The next goal is to construct a  variational problem on the sub-Riemannian manifold \((\SO(F),\cD,\langle\cdot,\cdot\rangle_\bK)\) whose solutions represent discrete incompressible Euler flows.  To determine the relevant equations of motion, we seek stationary paths \(\bR\colon [0,T]\to \SO(F)\) of the constrained action functional
\begin{equation}
    \label{eq:DiscreteAction}
    \cS(\bR) = \int_0^T \frac{1}{2}|\dot\bR|_{\bK}^2\,dt, \quad \text{ s.t.} \quad \dot\bR\in\cD_\bR.
\end{equation}
Observe that the velocity $\dot{\bR}$ is constrained to fiber $\cD_{\bR}$ of the distribution $\cD$ at each point in time, reflecting the sub-Riemannian character of the problem.

Since the action functional \teqref{eq:DiscreteAction} is right-invariant by construction, it has an equivalent expression on the Lie algebra $\so(F)$, 
\begin{equation}
\label{eq:ReducedDiscreteAction}
\cS(\bR) = \int_0^T \frac12|\bX|^2_{\bK}\,dt, \quad \bX\coloneqq\dot\bR\bR^{-1}\in \im(\bar{\bA}).
\end{equation}
Note once again that the image of the discrete infinitesimal advection $\im(\bar\bA) \subsetneq \im(\bA) \subsetneq \so(F)$ is not a Lie subalgebra, so that the reduced variational problem is nonholonomically constrained.

Now, suppose \(\bigmathring\bR\) is an arbitrary variation of \(\bR\colon [0,T]\to\SO(F)\) with
fixed endpoints.  It defines a corresponding reduced variation $\bY\coloneqq\bigmathring\bR\bR^{-1}\in\so(F)$ through right translation, which is not 
\begin{wrapfigure}[8]{r}{0.275\columnwidth}
    \centering
    \vspace{-0.\baselineskip}
{\hspace{-1.5em}\includegraphics[width=0.325\columnwidth]{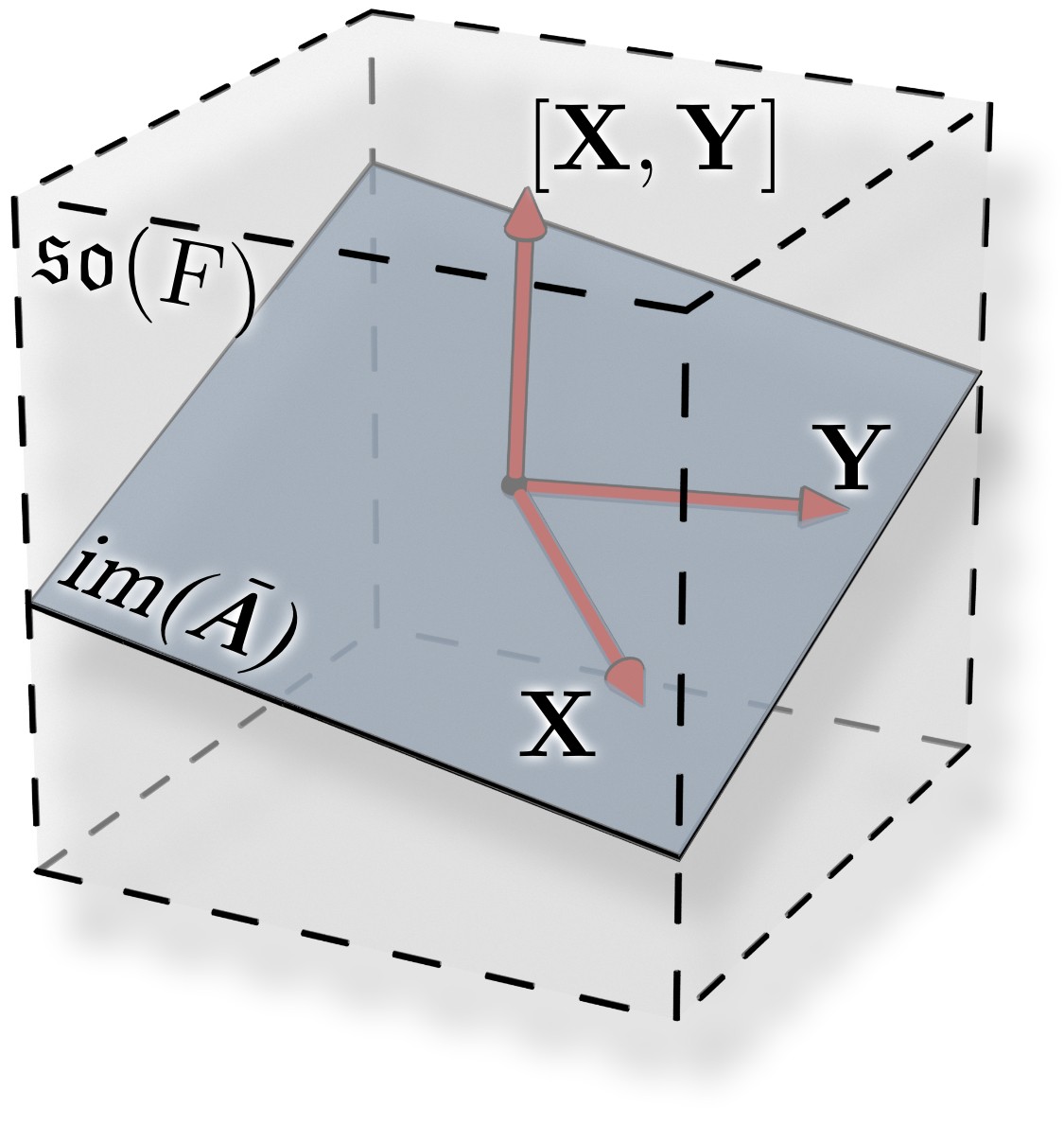}}
\end{wrapfigure}
necessarily constrained to $\im(\bar\bA)$. As before, the induced variation $\bigmathring{\bX}$ of the body velocity $\bX$ is given by the Lin
constraint 
\begin{equation}
\label{eq:DiscreteLinConstraint}
    \displaywidth=\parshapelength\numexpr\prevgraf+2\relax
\bigmathring\bX=\dot\bY-[\bX,\bY],
\end{equation}
which similarly takes values in the full Lie algebra \(\so(F)\). To make sense of expressions such as ``\(\langle\bX,\bigmathring{\bX}\rangle_{\bK}\)'' which will appear in the variation of \(\cS\), 
we extend the inner product \(\langle\cdot,\cdot\rangle_{\bK}\) on \(\im(\bar\bA)\) to an inner product \(\langle\cdot,\cdot\rangle_{\hat\bK}\) on \(\so(F)\) which agrees with \(\langle\cdot,\cdot\rangle_{\bK}\) on \(\im(\bar\bA)\). For instance, choose any subspace \(W\subset\so(F)\) transverse to \(\im(\bar\bA)\), so that \(\so(F)=\im(\bA)\oplus W\) is a direct sum. Then, equip \(W\) with an arbitrary inner product, and declare this decomposition to be orthogonal. It follows from Riesz representation that, for any $\bX\in\im(\bar\bA)$, there exists a covector 
\begin{equation*}
\bX^\flat \coloneqq \langle \bX,\cdot\rangle_{\hat\bK}\in \so(F)^*,
\end{equation*}
depending on the extended metric $\langle\cdot,\cdot\rangle_{\hat\bK}$.  However, observe that the action of $\bX^\flat$ on any $\bigmathring{\bX}\in\so(F)$ is given by 
\[\bX^\flat(\bigmathring{\bX}) = \langle\bX,\bigmathring{\bX}\rangle_{\hat\bK} = \langle\bX, \bigmathring{\bX}_{\parallel}\rangle_{\bK},\]
where $\bigmathring{\bX}_{\parallel}\in\im(\bar\bA)$ denotes the $\hat{\bK}$-orthogonal projection of $\bigmathring{\bX}$ into $\im(\bar\bA)$.  It follows that the map $\bigmathring{\bX} \mapsto \bX^\flat(\bigmathring{\bX})$ is always $\im(\bar\bA)$-valued, since $\bX^\flat$ automatically discards the components of \(\bigmathring\bX\) transverse to \(\im(\bar\bA)\).

While the introduction of an extended metric may seem inconsistent with the vakonomic point-of-view, it will soon be shown that the resulting equations of motion do not depend on any metric information from outside the distribution $\cD$.

\subsubsection{Vakonomic equations of motion}
\label{sec:VakonomicEquationsOfMotion}
\begin{wrapfigure}[9]{r}{0.325\columnwidth}
    \centering
    \includegraphics[width=0.325\columnwidth]{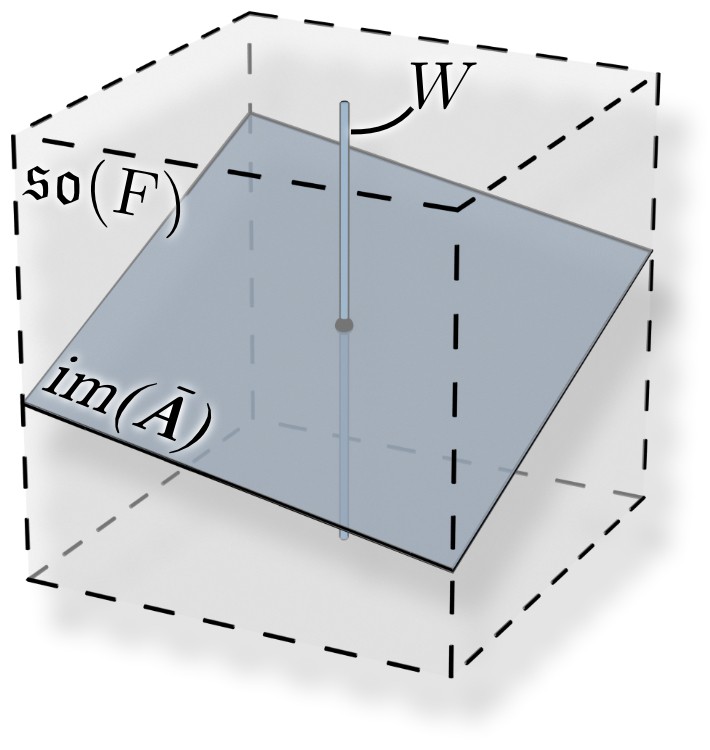}
\end{wrapfigure}
Deriving equations of motion from the constrained reduced action \teqref{eq:ReducedDiscreteAction} requires imposing the constraint \(\bX\in\im(\bar\bA)\).  To that end, consider the $W$-valued linear mapping
\begin{equation*}
\displaywidth=\parshapelength\numexpr\prevgraf+2\relax
	\theta\colon\so(F)\to W
\end{equation*}
which represents orthogonal projection (with respect to the extended metric $\hat\bK$) of elements in the ambient Lie algebra onto \(W\subset\so(F)\).
By definition, it follows that the kernel \(\ker\theta=\im(\bar\bA)\) of $\theta$ is the complementary subspace $\im(\bar\bA)$, and therefore the admissibility condition $\bX\in\im(\bar\bA)$ is equivalent to \(\theta(\bX)=0\). This enables the use of a Lagrange multiplier \(\alpha\colon[0,T]\to W^*\) to formulate an unconstrained equivalent to \teqref{eq:ReducedDiscreteAction}, 
\begin{equation}
\label{eq:DiscreteReducedAugmentedAction}
\tilde{\cS}(\bR,\alpha) = \int_0^T \left( \frac12\langle\bX,\bX\rangle_{\hat\bK} + \langle\alpha\,\vert\,\theta(\bX)\rangle \right)dt, \quad \bX=\dot\bR\bR^{-1},
\end{equation}
which can be directly differentiated to yield the vakonomic equations of motion given in the Theorem below.  

\begin{theorem}
\label{thm:DiscreteVakonomicEOM}
A path \((\bR,\alpha)\colon[0,T]\to\SO(F)\times W^\ast\) is stationary for the augmented action in~\teqref{eq:DiscreteReducedAugmentedAction} under variations with fixed endpoints for \(\bR\) if and only if the pair \((\bX\coloneqq \dot\bR\bR^{-1},\alpha)\),  satisfies the equations
\begin{equation}
\label{eq:vakonomic-eqs}
\begin{cases}
\displaystyle
(\tfrac{\partial}{\partial t}+\bad_{\bX}^*)
(\bX^\flat+\theta^*\alpha)=0,\\
\theta(\bX)=0.
\end{cases}
\end{equation}
Here, \(\theta^*\colon W^*\to\so(F)^*\) is the dual of \(\theta\), and
\(\bad_{\bX}^*\colon \so(F)^*\to \so(F)^*\) is the dual of the adjoint map
\(\bad_{\bX}\coloneqq [\bX,\cdot]\). %
\end{theorem}

\begin{proof}
Variations of \teqref{eq:DiscreteReducedAugmentedAction} %
with respect to \(\alpha\) immediately imply stationarity if and only if \(\theta(\bX)=0\), yielding the second equation in \teqref{eq:vakonomic-eqs}. For variations of \(\bR\) with fixed
endpoints, set \(\bY\coloneqq \bigmathring\bR\bR^{-1}\), so that \(\bY(0)=\bY(T)=0\), and
\(\bigmathring\bX=\dot\bY-[\bX,\bY]\) by \teqref{eq:DiscreteLinConstraint}. The first variation of the augmented action 
\eqref{eq:DiscreteReducedAugmentedAction} with respect to $\bR$ becomes
\begin{align*}
    &\quad\left.{d\over d\epsilon}\right|_{\epsilon=0}\int_0^T{\tfrac12}\langle\bX_{t,\epsilon},\bX_{t,\epsilon}\rangle_{\hat\bK} + \langle \alpha \,|\, \theta(\bX_{t,\epsilon}) \rangle \, dt\\
    &=\int_0^T\langle \bX,\bigmathring{\bX}\rangle_{\hat\bK}  + \langle \theta^*\alpha \,|\, \bigmathring\bX \rangle \, dt
    \overset{\eqref{eq:DiscreteLinConstraint}}{=}
    \int_0^T\langle \bX^{\flat} + \theta^*\alpha\,|\,\dot{\bY} - [\bX,\bY]\rangle\, dt\\
    &=\int_0^T \langle -(\tfrac{\partial}{\partial t} + \bad^*_\bX)(\bX^{\flat} + \theta^*\alpha) \,\vert\,\bY\rangle\,dt,
    \label{eq:VariationOfDiscreteAction}
\end{align*}
where the last equality follows from integrating the \(\dot\bY\) term by parts together with the fixed endpoint conditions \(\bY(0)=\bY(T)=0\). Since \(\bY\in\so(F)\) is arbitrary, we obtain stationarity if and only if 
\((\tfrac{\partial}{\partial t}+\bad_{\bX}^*)(\bX^\flat+\theta^*\alpha)=0\) in the dual space $\so(F)^*$.
\end{proof}
Crucially, after making the substitution 
\begin{equation*}
    \bZ = \bX^\flat + \theta^*\alpha,
\end{equation*}
\Cref{thm:DiscreteVakonomicEOM} establishes the \emph{coadjoint equation} on $\so(F)^*$, %
\begin{equation}
\label{eq:coadjoint-Z}
\dot{\bZ} + \bad_{\bX}^* \bZ = 0.
\end{equation}
Note that the difference \(\bZ-\bX^\flat\in \im(\bar\bA)^\circ\) lies in the annihilator of $\im (\bar\bA)$. Consequently, \(\bZ\) and \(\bX^\flat\) represent the same equivalence class in \(\so(F)^*/\im(\bar\bA)^\circ\).  Therefore, it follows from the discussion in \Cref{sec:MusicalIsomorphismsInSubRiemannianSetup} that $\bX$ and $\bZ^{\sharp}$ agree on the image of $\bar\bA$, \ie,
\begin{equation}\label{eq:sharp-operator}
    \bX = \bZ^{\sharp} \in \im(\bar\bA).
\end{equation}
Substituting~\teqref{eq:sharp-operator} into~\teqref{eq:coadjoint-Z} yields a closed vakonomic evolution on \(\so(F)^*\): 
\begin{equation}
\label{eq:Z-vakonomic-final}
\dot{\bZ} + \bad_{\bZ^\sharp}^*\bZ = 0.
\end{equation}

\begin{remark}
     While the definitions of $\bX^\flat$ and $\theta^*\alpha$ depend on the chosen metric extension \(\langle\cdot,\cdot\rangle_{\hat\bK}\), the evolution \teqref{eq:Z-vakonomic-final} for \(\bZ\in\so(F)\) is independent of it. 
     This follows from the definition of the sharp operator \(\sharp = \bar\bA\bK^{-1}\bar\bA^*\) (\Cref{sec:MusicalIsomorphismsInSubRiemannianSetup}) which does not require metric information from outside the subspace $\im(\bar\bA)$, consistent with the vakonomic perspective. 
\end{remark}

We now show that \teqref{eq:Z-vakonomic-final} can be expressed as a matrix Lax equation, yielding a well behaved isospectral evolution.
It is convenient to identify the dual space $\so(F)^*\cong\so(F)$ with the Lie algebra itself through the usual Frobenius pairing.  With this, the following Lemma provides an explicit representation for the coadjoint operator.
\begin{lemma}
\label{lem:CoadjointIsCommutator}
Let \(\so(F)^*\cong \so(F)\) be Riesz identification using the Frobenius inner product, so that \(\langle \bZ \,\vert\, \bY\rangle = \langle \bZ,\bY\rangle_{\rm Frob} \coloneqq \tr(\bZ^\intercal \bY)\) on \(\so(F)\).
Then the coadjoint action \(\bad_{\bX}^*\) associated with \(\bad_{\bX}\coloneqq [\bX,\cdot]\) 
satisfies
\begin{equation*}
\bad_{\bX}^* \bZ = -[\bX,\bZ],
\qquad
\bX,\bZ\in\so(F). 
\end{equation*}
\end{lemma}

\begin{proof}
Consider an arbitrary \(\bW\in\so(F)\). Using the definition of the dual map along with $\bX^\intercal=-\bX$, we compute
\begin{align*}
    \langle \bad_{\bX}^*\bZ \,\vert\, \bW\rangle
&=
\langle \bZ \,\vert\, \bad_{\bX}\bW\rangle
=
\langle \bZ, [\bX,\bW]\rangle_{\rm Frob} \\ 
&= 
\langle \bZ, \bX\bW-\bW\bX\rangle_{\rm Frob} = \langle \bX^\intercal\bZ-\bZ\bX^\intercal, \bW\rangle_{\rm Frob} \\
&= -\langle \bX\bZ-\bZ\bX, \bW\rangle_{\rm Frob} = \langle -[\bX,\bZ],\bW\rangle_{\rm Frob} \\ &= \langle -[\bX,\bZ]\mid\bW\rangle.
\end{align*}
It follows immediately that $\bad_{\bX}^*\bZ = -[\bX,\bZ]$, as desired.
\end{proof}

Substituting the result of \lemref{lem:CoadjointIsCommutator} into the evolution equation \teqref{eq:Z-vakonomic-final} yields 
a matrix \emph{Lax equation} on $\so(F)$:
\begin{equation}
	\label{eq:VakonomicLaxEquation}
    \dot{\bZ} = [\bZ^\sharp,\bZ].
\end{equation}

\subsection{Casimirs of motion and Noether charges}
\label{sec:CasimirsAndNoetherCharges}

The vakonomic equation of motion \eqref{eq:Z-vakonomic-final} or \eqref{eq:VakonomicLaxEquation} relates two equivalent sets of conservation laws.  The first describes conservation of the spectrum of \(\bZ\), while the second provides a discrete version of \emph{Kelvin's circulation theorem}, stating that the circulation along every loop transported by the flow is a constant of motion.

\subsubsection{Isospectral Property}
The Lax equation \eqref{eq:VakonomicLaxEquation} is isospectral, \ie, the spectrum of its solutions are constant along in time.  This property is satisfied by any system of the form
\begin{align}\label{eq:GeneralLaxSystem}
    \dot{\bZ}(t)=[\bX(t),\bZ(t)],\quad\overset{\text{(\autoref{lem:CoadjointIsCommutator})}}{\Longleftrightarrow}\quad 
    \dot\bZ(t) = -\bad^*_{\bX(t)}\bZ(t),
\end{align}
where \(\bX\colon [0,T]\to\so(F)\) is an arbitrary path on \(\so(F)\). The solution \(\bZ(t)\) to \eqref{eq:GeneralLaxSystem} is explicitly given by a similarity transformation of \(\bZ(0)\) for $\bX$ independent of $\bZ$:
\begin{align}\label{eq:LaxSolutionAsSimilarityTransform}
    \bZ(t) = \bR(t)\bZ(0)\bR(t)^\intercal,
\end{align}
where \(\bR\colon[0,T]\to \SO(F)\) is the solution to \(\dot\bR(t) = \bX(t)\bR(t)\) with initial condition \(\bR(0) = \id\).
Therefore, the spectrum of \(\bZ\) is preserved over time.

Observe that the conservation of moments of $\bZ(t)$ is both a necessary and sufficient condition for solving the Lax equation \eqref{eq:GeneralLaxSystem}.
This follows since any smooth path \(\bZ(t)\) with conserved moments must also have conserved spectrum, implying that \(\bZ(t)\) is similar to
\(\bZ(0)\) 
(\cf  \eqref{eq:LaxSolutionAsSimilarityTransform}). Thus, \(\bZ(t)\) satisfies \eqref{eq:GeneralLaxSystem} for some \(\bX(t)\).
This equivalence between \eqref{eq:GeneralLaxSystem} and moment conservation can be understood through \emph{Casimir invariants}, which are distinguished phase-space functions constant along coadjoint orbits.
The even moments \[\tr(\bZ^2),\tr(\bZ^4),\ldots, \tr(\bZ^{2\lfloor F/2\rfloor})\] represent a complete set of Casimir invariants on \(\so(F)^*\), with intersections of their level sets equal to the coadjoint orbits defined by \eqref{eq:GeneralLaxSystem}.  That is, the conservation of these Casimirs along a path \(\bZ(t)\) is equivalent to the 
evolution \(\dot\bZ(t) = -\bad_{\bX(t)}^*\bZ(t)\) by coadjoint actions.

\subsubsection{Discrete Kelvin's Circulation Law}

In the continuous setting, Kelvin's circulation law states that the circulation along every material loop is a constant of motion.  The vakonomic evolution equation \eqref{eq:Z-vakonomic-final} (or \eqref{eq:VakonomicLaxEquation}) satisfies a discrete analogue of the circulation law similar to (but stronger than) that in \cite{Pavlov:2011:SPD}.

To state this precisely, recall that 
the path \(\bX\colon[0,T]\to\so(F)\) in \eqref{eq:GeneralLaxSystem} represents the velocity field of a discrete fluid motion if it takes values in the space \(\im(\bar\bA)\subset\so(F)\) of discrete divergence-free fields. 
We define a \emph{discrete loop} to be a discrete divergence-free field \(\bY\in\im(\bar\bA)\subset \so(F)\).
This is analogous to the continuous setting, where a loop is a smooth closed curve, embedded in the fluid domain, that  
can be identified with a closed de Rham current via the limit of a divergence-free vector field.  
A time-dependent discrete loop \(\bY(t)\) is said to be transported by the flow of \(\bX(t)\) if it satisfies the Lax equation \(\dot\bY(t) = [\bX(t),\bY(t)]\).  However, note that \(\im(\bar\bA)\) is not closed under Lie brackets, meaning that 
most loops \(\bY\) flow out of the space \(\im(\bar\bA)\) under this evolution equation.
Therefore, the definition of a discrete loop is extended to include elements of \(\so(F)\), or the Lie-algebraic closure of \(\im(\bar\bA)\).

Given a covector $\bZ\in\so(F)^*$, the \emph{circulation} of $\bZ$ along the discrete loop $\bY\in\so(F)$ is just the evaluation pairing $\langle\bZ|\bY\rangle$. So, a path  \(\bZ\colon[0,T]\to\so(F)^*\)satisfies \emph{Kelvin's circulation law} 
with respect to the flow of \(\bX\colon[0,T]\to\im(\bar\bA)\) 
if the circulation \(\langle\bZ(t)|\bY(t)\rangle = \langle\bZ(0)|\bY(0)\rangle\) is independent of \(t\) for all loops \(\bY\colon[0,T]\to\so(F)\) transported by the flow, \ie, for all loops satisfying the Lax equation
\(\dot\bY(t) = [\bX(t),\bY(t)]\).  This leads to the following result describing how the vakonomic system \eqref{eq:Z-vakonomic-final} obeys this version of Kelvin's circulation theorem.

\begin{theorem}\label{thm:StrongKelvinCirculationTheorem}
    Consider a discrete velocity field \(\bX\colon [0,T]\to\im(\bar\bA)\).  A path \(\bZ\colon[0,T]\to\so(F)^*\) satisfies the Lax equation $\dot\bZ = [\bX,\bZ]$ if and only if it satisfies Kelvin's circulation law with respect to the flow of $\bX$. 
\end{theorem}

\begin{proof}
    Let $\dot{\bY} = [\bX,\bY]$ be a loop advected by the flow.   Observe that the derivative of the evaluation pairing is
    \begin{align*}
        \tfrac{d}{dt}\langle\bZ\,|\,\bY\rangle &= \langle\dot{\bZ}\,|\,\bY\rangle + \langle\bZ\,|\,\dot{\bY}\rangle \\
        &= \langle\dot{\bZ}\,|\,\bY\rangle + \langle \bZ\,|\,[\bX,\bY]\rangle = \langle \dot{\bZ} + \bad_{\bX}^*\bZ\,|\,\bY\rangle.
    \end{align*}
    This quantity vanishes for all $\bY$ if and only if $\bZ$ satisfies $(i)$, in which case $\bZ$ also satisfies $(ii)$.
\end{proof}

\begin{remark}
    Note that \Cref{thm:StrongKelvinCirculationTheorem} is strictly stronger than the corresponding result in \cite[Theorem 2]{Pavlov:2011:SPD}. There, it is proved that $\langle\bZ|\bY\rangle$ is constant in time for all loops $\bY$ that are weakly advected by the flow, \ie, such that $\langle\bB\,|\,\dot{\bY}\rangle = \langle \bB\,|\,[\bX,\bY]\rangle$ for all $\bB\in\im(\bar\bA)^*$.  This can be interpreted as a weak version of Kelvin's circulation law, with weakness stemming from the LdA formulation that solves only a projected Lax equation.  
\end{remark}

Notice that \Cref{thm:StrongKelvinCirculationTheorem} can also be understood as a Noether theorem associated with the discrete particle-relabeling symmetry, in direct analogy with the continuous Euler equations.  More precisely, the vakonomic action \eqref{eq:DiscreteAction} is invariant under the action of the relabeling group $\SO(F)$ on discrete fluid configurations.  By Noether's theorem, this symmetry yields a conserved quantity that can be identified with the momentum $\bR^{-1}\bZ\bR$.  In this case, it is straightforward to check that the evaluation pairing $\langle\bR^{-1}\bZ\bR|\bR^{-1}\bY\bR\rangle = \langle\bZ|\bY\rangle$ is conserved for any $\bY$ advected by the flow.

\begin{remark}
    The equivalency in \Cref{thm:StrongKelvinCirculationTheorem} implies that the even moments \(\tr(\bZ^2),\ldots,\tr(\bZ^{2\lfloor F/2\rfloor})\) are time-independent.  The converse is almost true: when the even moments are constant in time, the flow velocity \(\bX\) is guaranteed to take values in \(\so(F)\) but not necessarily in \(\im(\bar\bA)\).
\end{remark}

\subsection{Hamiltonian Formulation}
\label{subsec:hamiltonian}

Another noteworthy advantage of the vakonomic equation \eqref{eq:VakonomicLaxEquation} is its admission of a Hamiltonian formulation. 
To see this, recall the modern Hamiltonian formulation of the general Euler--Arnold system, describing geodesic flows 
on a Lie group \(G\) equipped with a right-invariant metric.  Making an appropriate variable transformation, the geodesic equations on $G$ 
can be written as a Hamiltonian flow on the dual space \(\mathfrak{g}^*\) of the Lie algebra \(\mathfrak{g} = T_eG\) with respect to the canonical Lie--Poisson structure: 
\begin{equation}\label{eq:LiePoisson}
    \{F,G\}_{\bZ} = \langle \bZ\,|\,[dF_{\bZ},dG_{\bZ}]\rangle, \quad \bZ\in\mathfrak{g}^*,\,\,F,G\in C^{\infty}(\mathfrak{g}^*).
\end{equation}
This flow, given with some abuse of notation by $\dot{\bZ}=\{\bZ,H\}$ or, equivalently, $\dot{\bZ} + \bad^*_{dH}\bZ = 0$ for a distinguished function $H\in C^{\infty}(\mathfrak{g}^*)$ called the \emph{Hamiltonian},
is understood as the result of Poisson reduction, \ie, the phase-space analogue of Euler-Poincar\'e reduction by symmetry (see, e.g., \cite{Marsden:1997:IMS}).  Moreover, the fact that Hamilton's equations of motion generate coadjoint orbits implies that solutions always preserve the distinguished phase-space Casimir invariants discussed previously.
Concretely, recall that Casimirs are those functions $C\in C^{\infty}(\mathfrak{g}^*)$ that Poisson-commute with all other functions, \ie, $\{C,F\}=\{F,C\}=0$ for all $F\in C^{\infty}(\mathfrak{g}^*)$.  It follows that under Hamiltonian evolution,
\begin{align*}
    \dot{C}_\bZ &= \langle \dot\bZ\,|\,dC_{\bZ}\rangle = -\langle \bad^*_{dH_{\bZ}}\bZ\,|\,dC_{\bZ}\rangle \\
    &= -\langle\bZ \,|\,[dH_{\bZ},dC_{\bZ}]\rangle = \{C,H\}_{\bZ} = 0,
\end{align*}
and any Casimir invariant is preserved along the flow. In particular, these Casimirs satisfy 
\[\bad_{dC_{\bZ}}^*\bZ = 0 \quad\overset{\text{(\autoref{lem:CoadjointIsCommutator})}}{\Longleftrightarrow}\quad [dC_{\bZ},\bZ] = 0,\]
so that their derivatives commute with $\bZ$.  Therefore, the Casimirs of the Lie-Poisson structure are precisely smooth functions of the spectral invariants discussed before.

Specializing this discussion to the current vakonomic system, observe that the coadjoint equation \eqref{eq:coadjoint-Z} coincides with Hamilton's equations for the Hamiltonian
\begin{equation}\label{eq:HamiltonianVakonomic}
    H(\bZ) = \frac{1}{2}\langle\bZ\,|\,\bZ^{\sharp}\rangle, \quad \bZ\in\so(F)^*.
\end{equation}
In particular,  \eqref{eq:HamiltonianVakonomic} is simply the kinetic energy \({1\over 2}\langle\dot\bR,\dot\bR\rangle_{\bK}\) expressed in the momentum variable \(\bZ\in\so(F)^*\).

\begin{remark}
    More geometric intuition for the Hamiltonian \eqref{eq:HamiltonianVakonomic} can be gained by noting its expression as the pullback \(H = h\circ\bar\bA^*\), where 
    \(h\in C^\infty(V_{\div}^*)\) is the energy functional \(h(\bv) = {1\over 2}\bv^\intercal\bK^{-1}\bv\).  In particular, \(H\) is constant along translations in \(\ker(\bar\bA^*) = \im(\bar\bA)^\circ\) and descends to a nondegenerate quadratic form \(h\) under the quotient by \(\ker(\bar\bA^*)\).
\end{remark}

Finally, we note that the equivalence between curves that lift to Hamiltonian flows, curves that are critical points of the action $\cS$ over the space of admissible curves, and curves computed using the vakonomic variational principle by Lagrange multipliers holds quite generally, for all regular curves on arbitrary sub-Riemannian manifolds \citep{Piccione:2001:VAG}.

\subsubsection{Vakonomic as limit of holonomic}
It is worth mentioning that, despite the degeneracy of the right-invariant metric $\bK$, the right-invariance of the distribution \(\cD\) can be used to admit an interpretation of the vakonomic Hamiltonian system as the limit of a holonomic Hamiltonian system (\cf, \cite{Abanov:2025:IDN}).

To see this, extend the metric \(\bK\) on \(\cD\) to a one-parameter family of right-invariant metrics \(\hat\bK^{\varepsilon}\) on \(\SO(F)\) as follows. For each \(\varepsilon>0\), define the metric \(\langle\cdot,\cdot\rangle_{\hat\bK^\varepsilon}\) to agree with \(\langle\cdot,\cdot\rangle_{\bK}\) on the distribution \(\cD\) and satisfy \(\langle\bY\bR,\bY\bR\rangle_{\hat\bK^\varepsilon}\to\infty\) as \(\varepsilon\to 0\) for any \(\bY\bR\notin \cD_\bR\).  Equivalently, it is enough that the $\hat{\bK}^\varepsilon$ and $\bK$ metrics agree on ${\rm im}\,\bar{\bA}$ and \(\langle\bY,\bY\rangle_{\hat\bK^\varepsilon}\to\infty\) for \(\bY\notin\im(\bar\bA)\).  
In the limit \(\varepsilon\to 0\), any path \(\bR\) with finite length is therefore forced to satisfy \(\dot\bR\in\cD_\bR\).

For \(\varepsilon\neq 0\), the (holonomic) Euler--Arnold equation for geodesics on \((\SO(F),\langle\cdot,\cdot\rangle_{\hat\bK^\varepsilon})\) takes the form 
\begin{align}
\label{eq:EulerArnoldWithEpsilon}
    \dot\bZ + \bad_{\bZ^{\sharp_\varepsilon}}^*\bZ = 0, \quad \bZ\in\so(F)^*,
\end{align}
with \(\sharp_{\varepsilon}\colon T_{\bR}^*\SO(F)\to T_{\bR}\SO(F)\) denoting the sharp operator for the extended metric \(\hat\bK^\varepsilon\).
The corresponding group trajectory \(\bR\) is obtained from \(\dot\bR = \bZ^{\sharp_{\varepsilon}}\bR\).
Notice that \eqref{eq:EulerArnoldWithEpsilon} is a Hamiltonian flow on \(\so(F)^*\) with respect to the standard Lie--Poisson bracket \eqref{eq:LiePoisson}
Specifically, it can be written as 
\(\dot\bZ = -{\rm sgrad}\,{H^\varepsilon}({\bZ})\) where the Hamiltonian vector field is given by the symplectic gradient 
\({\rm sgrad}\,{H^\varepsilon}\coloneqq \{H^\varepsilon,\cdot\}\) and the Hamiltonian is
\begin{align}
\label{eq:HamiltonianForEulerArnoldWithEpsilon}
    H^\varepsilon(\bZ) = {1\over 2}\langle \bZ\, |\, \bZ^{\sharp_\varepsilon}\rangle,
\end{align}
in analogy with the nonholonomic Hamiltonian formulation just derived.

An advantage of this Hamiltonian formulation is that the singular limit \(\varepsilon\to 0\) is well behaved.  
Although the metric \(\langle\cdot,\cdot\rangle_{\hat\bK^\varepsilon}\) and associated flat operator $\flat_{\varepsilon}$ diverge in this limit, the sharp operator \(\sharp_\varepsilon\) converges to a well-defined object
\begin{align}
    \sharp_0 = \bar\bA\bK^{-1}\bar\bA^* = \sharp,
\end{align}
which agrees with the vakonomic sharp.
Consequently, the Hamiltonian \eqref{eq:HamiltonianForEulerArnoldWithEpsilon} also converges to the vakonomic version \eqref{eq:HamiltonianVakonomic}, showing that the  vakonomic equation \eqref{eq:Z-vakonomic-final} arises as the singular limit of the (holonomic) Euler--Arnold equation \eqref{eq:EulerArnoldWithEpsilon}.

\section{Comparison to Other Structure Preserving Methods}\label{sec:comparison}

Before discussing how the vakonomic equations of motion \eqref{eq:Z-vakonomic-final} are simulated, it is worth discussing their relationship to other structure-preserving discrete fluid schemes in the current literature.  This section discusses the equations arising via the complementary Lagrange--d'Alembert (LdA) principle in some detail before briefly surveying other closely related work.

\subsection{Comparison to the Lagrange--d'Alembert principle}
\label{sec:ComparisonToLdA}

Recall that the LdA approach to fluid simulation enforces the non-holonomic constraint $\dot\bR\in\cD_{\bR}$ on the configuration variable
by restricting the admissible
virtual displacements \(\bigmathring\bR\in \cD\subset T\SO(F)\) to lie in the constraint distribution.  After reduction, this means that $\bY\coloneqq \bigmathring{\bR}\bR^{-1}\in\im(\bar\bA)$ must lie in the image of the discrete advection operator.

Concretely, let \(\bR:[0,1]\to\SO(F)\) have reduced body velocity \(\bX\coloneqq \dot\bR\bR^{-1}\in\im(\bar{\bA})\).
A variation \(\bigmathring\bR\) with fixed endpoints induces a reduced variation $\bigmathring{\bX}$ satisfying the Lin constraint \eqref{eq:DiscreteLinConstraint} in terms of $\bY= \bigmathring{\bR}\bR^{-1}$. Enforcing that $\bY \in \im(\bar\bA)$ in accordance with the LdA principle, 
the first variation of the reduced kinetic energy computes to 
\begin{align}
    &\quad\left.\frac{d}{d\epsilon}\right|_{\epsilon=0}\int_0^T \frac{1}{2}\lvert\bX_{t,\epsilon}\rvert_{\hat\bK}^2\,dt
    =\int_0^T \langle \bX,\bigmathring\bX\rangle_{\hat\bK}\,dt \nonumber\\
    &=\int_0^T \langle \bX^\flat \,\vert\, \dot{\bY}-[\bX,\bY]\rangle\,dt
    =\int_0^T \langle -(\tfrac{\partial}{\partial t}+\bad_\bX^*)\bX^\flat \,\vert\, \bY\rangle\,dt,
\label{eq:VariationOfDiscreteActionLdA}
\end{align}
where $\flat$ is understood with respect to the extended metric $\hat{\bK}$ from before.  
Since \(\bY\in\im(\bar\bA)\subset\so(F)\) is no longer arbitrary but constrained to a subspace,
\eqref{eq:VariationOfDiscreteActionLdA} implies only that the one-form 
\((\frac{\partial}{\partial t}+\bad_\bX^*)\bX^\flat\) vanishes on \(\im(\bar\bA)\).  Equivalently, there exists a reaction force \(\bF\in\im(\bA)^\circ\) in the annihilator subspace satisfying
\begin{equation}
\label{eq:DiscreteLdAlembertEOM}
\displaystyle
\dot\bX^\flat + \bad_\bX^*\bX^\flat = \bF.
\end{equation}
Notice that this reaction force is a distinguished element of $\im(\bar\bA)^\circ$, specifically chosen to project the coadjoint term $\bad_{\bX}^*\bX^\flat$ back into the dual space $\im(\bar\bA)^*$.  As such, it depends explicitly on the metric extension $\hat{\bK}$ of $\bK$, as does the covector $\bX^\flat$ that 
exists on the entirety of $\so(F)^*$. 
Collapsing this degree of freedom with the globally defined sharp operator $\sharp=\bar\bA\bK^{-1}\bar\bA^*:\so(F)^*/\im(\bar\bA)^\circ\to\im(\bar\bA)$ leads to 
the corresponding LdA evolution for $\bX\in\im(\bar\bA)$:
\begin{equation}
    \label{eq:DiscreteLdAlembertProjected}
    \dot{\bX} + \big(\bad^*_{\bX}\bX^\flat\big)^\sharp = 0.
\end{equation}
Clearly, \eqref{eq:DiscreteLdAlembertProjected} still relies on metric information away from $\im(\bar\bA)$ since $\bX^\flat$ appears explicitly.
Using the matrix representation of the coadjoint action in \Cref{lem:CoadjointIsCommutator}, induced by the Frobenius inner product, these equations also have the alternative expression \eqref{eq:LdA-overview} claimed in \Cref{sec:overview}.  To see this, let $\bZ \in \im(\bar{\bA})$ denote the (Frobenius) Riesz representation of $\bX^\flat$ and denote $\hat{\bK}$-orthogonal projection onto the constraint distribution $\cD$ by $\pi_{\cD}\colon\so(F)\to\im(\bar\bA)$.  Then, the LdA equations become
\begin{equation}
\label{eq:FinalMatrixLagrangeDAlembert}
\displaystyle
\dot{\bZ}=\pi_{\cD}([\bX,\bZ]).
\end{equation}

Note that the vakonomic equations \eqref{eq:Z-vakonomic-final} and the LdA equations \eqref{eq:DiscreteLdAlembertEOM} are genuinely different.  The vakonomic evolution 
evolves a dense dual variable \(\bZ\in\so(F)^*\) in the ambient skew-symmetric matrix space
whose horizontal sharp \(\bX=\bZ^\sharp\in\im(\bar\bA)\) lies in the sparse admissible subspace \(\im(\bar\bA)\). In contrast, the LdA equation evolves only the sparse variable \(\bX^\flat\in\im(\bar\bA)^*\) and repeatedly projects the right-hand side, discarding the components of the commutator that lie outside this admissible subspace.

\begin{remark}\label{rem:FEMtrick}
    Note that a choice of metric extension is required for computing quantities such as 
    $\langle\bA_k, [\bA_l, \bA_m]\rangle_{\hat{\bK}}$ appearing in the LdA equations, since the Lie bracket $[\bA_l,\bA_m]\notin \im(\bar{\bA})$ is not contained where the $\bK$-metric is defined.  Previous work \cite{Pavlov:2011:SPD} sidesteps this issue by arguing for the existence of an extension $\hat{\bK}$ satisfying
    \[ \langle\bA_k, [\bA_l, \bA_m]\rangle_{\hat{\bK}} \approx \int\langle\vec\psi_k, [\vec\psi_l,\vec\psi_m]\rangle\Vol,\]
    where approximation holds to second order in the grid spacing.  
    The right-hand side expression is then easily computable on the space $\cV_{\div}$ without explicitly constructing $\hat{\bK}$. 
\end{remark}

This difference between vakonomic and LdA is closely related to the sub-Riemannian picture discussed previously.  Vakonomic solutions \(\bZ_t\) satisfying \eqref{eq:Z-vakonomic-final} or equivalently \eqref{eq:VakonomicLaxEquation} integrate through the linear system \(\dot\bR = \bZ^\sharp\bR\) into critical paths \(\bR\colon [0,T]\to\SO(F)\) for the constrained action  \eqref{eq:DiscreteAction}.  Crucially, these paths are \emph{sub-Riemannian geodesics}: since the action \eqref{eq:DiscreteAction} integrates the kinetic energy in time, its extremizers are geodesics analogous to extrema of the time-integrated kinetic energy on a Riemannian manifold. 
Measured in the $\bK$-metric, the path \(\bR\) has constant speed, \({d\over dt}|\dot\bR(t)|_{\bK}=0\) and extremizes the total arc length
among all paths tangent to the distribution \(\cD\). 
Moreover, this property of being a sub-Riemannian geodesic depends only on the sub-Riemannian structure \((\SO(F),\cD,\langle\cdot,\cdot\rangle_{\bK})\) and is therefore independent of any particular choice of auxiliary metric extension \(\langle\cdot,\cdot\rangle_{\hat\bK}\).  In contrast, the Lagrange--d'Alembert equations describe geodesics of a projected connection which is not generally metrizable \cite{Fernandez:2008:EDN}.  Therefore, these trajectories cannot minimize path length in any meaningful sense.
\begin{remark}
    Interestingly, the LdA system can also be obtained as an appropriate singular limit of a holonomic system, see e.g., \cite{Abanov:2025:IDN}.  In this case, the kinetic energy is augmented with a Rayleigh dissipation potential $\frac{1}{\alpha}R(\bX)$, so that the resulting Euler--Arnold equations contain an infinitely strong frictional force as $\alpha\to 0$.
\end{remark}
To finish this discussion, it is worth noting more explicitly how the Hamiltonian formulation of the vakonomic equations \eqref{eq:Z-vakonomic-final} breaks down in the LdA case \eqref{eq:DiscreteLdAlembertEOM}.  This can easily be observed by considering the same canonical Lie-Poisson structure and computing the evolution of a Casimir $C\in C^{\infty}(\so(F)^*)$: 
\begin{align*}
    \dot{C}_\bZ &= \langle \dot\bZ\,|\,dC_{\bZ}\rangle = \langle (\bF-\bad^*_{dH_{\bZ}})\bZ\,|\,dC_{\bZ}\rangle \\
    &= \langle\bF\bZ \,|\,dC_{\bZ}\rangle - \{C,H\}_{\bZ} = \langle\bF\bZ \,|\,dC_{\bZ}\rangle \neq 0.
\end{align*}
Therefore, Casimirs are not preserved along the flow of the LdA system due to the reaction force $\bF$, meaning that solutions do not remain on coadjoint orbits and no Hamiltonian formulation is possible.  In fact, considering the natural ``Lie bracket'' $[\bU,\bV]_{\cD} = \pi_{\cD}([\bU,\bV])$ given by the $\hat\bK$-orthogonal projection $\pi_{\cD}:\so(F)\to\im(\bar\bA)$, the induced ``Lie--Poisson bracket'' $\{F,G\}_{\cD} = \langle\bZ\,|\,[dF_{\bZ},dG_{\bZ}]_{\cD}\rangle$ defined on $\im(\bar\bA)^*$ is almost Poisson.  That is, it is both bilinear and skew-symmetric but fails to satisfy the Jacobi identity captured by the vanishing of the Jacobiator,
\[{\rm Jac}(\bU,\bV,\bW) = [\bU,[\bV,\bW]] + [\bV,[\bW,\bU]] + [\bW,[\bU,\bV]], \]
where $[\cdot,\cdot]$ denotes whatever bracket is under consideration.  The key point is that the normal projection $\pi_{\cD}^\perp([\bV,\bW]) \neq 0 $ is nontrivial even along solutions to the LdA equations \eqref{eq:FinalMatrixLagrangeDAlembert}.  Conversely, ${\rm Jac}\equiv 0$ vanishes identically in the space where the vakonomic equations \eqref{eq:Z-vakonomic-final} are defined, retaining the Lie--Poisson character of the continuous system.  More details on the almost-Poisson brackets of nonholonomic mechanics can be found in 
\cite{vanderSchaft:1994:HFN, Koom:1997:HLA, Cantrijn:1999:APS}.

\subsection{Other Related Methods}

Here we record related approaches from the current literature along with their references.

\subsubsection{The vakonomic principle}
The vakonomic (\textbf{va}riational \textbf{k}ind of nonhol\textbf{onomic}) principle was coined by Vladimir Arnold \cite{Arnold:2006:MAC} and first appeared in the work of Valery Kozlov \cite{Kozlov:1983:DSN}. It represents a variational paradigm for treating nonholonomic systems 
through the imposition of constraints at the formulational level, \textit{before} taking variations. 
As mentioned, this is contrasted by the Lagrange--d'Alembert (LdA) model, where the nonholonomic constraints are applied \emph{after} taking unconstrained variations, and the resulting Euler--Lagrange equations are corrected through a virtual force that vanishes along the constraint directions.  This choice is somewhat contentious, and the textbook by Bloch et al.\ \cite{Bloch:2015:NMC} summarizes the long-standing debate over which of these models better describes physical systems observed in reality. 

The distinction between vakonomic and LdA is especially pronounced when the nonholonomic constraints arise from friction---as in a ball rolling without slipping, or the Chaplygin sleigh. In these cases, the equations of motion derived from the (non-variational) LdA principle agree better with physical measurements \cite{Bloch:2015:NMC}. Consequently, the vakonomic principle has largely been regarded as ``unphysical'' and confined to the setting of optimal control, where Hamiltonian structure is paramount.  Conversely, the present work deals with nonholonomic constraints that are inherently unphysical: they
do not originate from friction or any other physical source but
instead arise from the choice of Koopman discretization (\cf, \Cref{rem:nonholonomic}).  
Therefore, it is reasonable to believe that the vakonomic perspective could yield benefits over LdA in the case of simulating incompressible fluids.

\subsubsection{Discretization as a nonholonomic constraint: the Lagrange--d'Alembert tradition.}
Due to its widespread success in modeling physical systems, the LdA principle is now routinely applied to nonholonomic constraints arising numerically as well.
The observation that Koopman discretization imposes a nonholonomic constraint on fluid dynamics was first made in \cite{Pavlov:2011:SPD}. 
There, the LdA principle is used 
within a discrete exterior calculus (DEC) framework, effectively incorporating the lack of Lie-algebraic closure but also limiting the accuracy of the method to first 
order. Remarkably, the earlier work 
\cite{Mullen:2009:EPI} had already used these same equations of motion for a computer-graphics realization of this theory. This alignment is not a coincidence: a (Petrov-)Galerkin discretization of the standard weak-form equations for an incompressible fluid naturally obeys the LdA principle with a particular choice of ambient metric. 
More precisely, enforcing orthogonality of the trial-space residual against a finite-dimensional space of test functions satisfying the nonholonomic constraints enforces the principle of virtual work and generates a force orthogonal to the discretization space.
Building on this correspondence, \cite{Gawlik:2011:GVD} extended this framework to the more general semidirect-product Lie groups, enabling modeling of additional physics such as compressible fluids, while \cite{Natale:2018:VFD} carried the earlier work from DEC to general finite element settings. Many later works build on these ideas, 
yet none exploit the information available in the Lie algebra of skew-symmetric matrices. Because 
information is only required from ``one Lie bracket away,'' 
these works obtain it directly from the FEM construction (\cf, \Cref{rem:FEMtrick}), so that in the end no structure from the Lie algebra is used. 
Conversely, the present work formulates 
equations of motion directly in the Lie algebra
and the resulting Lax equations respect the structure of 
this matrix space and its nonholonomic constraint.

\subsubsection{Two-dimensional vorticity formulations.}
Variational principles are certainly not the only way to derive fluid simulators.  In two dimensions, incompressible fluid dynamics can be formulated in terms of a scalar vorticity function that is advected by the flow. 
Moreover, when finite elements are used, the vorticity is simply rotated by the divergence-free velocity at each time step. 
This program was carried out in \cite{Azencot:2014:FFS}, following the functional-operator 
formalism
of \cite{Azencot:2013:OAT}, which amounts to a mixed finite element method with tailored bases. The main drawback of this approach is that the resulting advection operator fails to be skew-symmetric even for divergence-free velocities, a consequence of choosing discrete function spaces that lack de Rham compatibility.  
Moreover, because the  time stepping employed is not exactly Cayley integration, it would not produce updates remaining in the space of rotation matrices even if the matrix were skew-symmetric. 
Conversely, the present formulation avoids these issues through the use of 
de Rham compatible Whitney bases on triangles and FEEC IGA B-splines on tensor grids, along with (Cayley) Lie group integration in time.

\subsubsection{Structure-preserving approaches.}
Beyond the purely Eulerian methods discussed so far, a number of fluid simulation works aim explicitly at some form of structure preservation. For instance, \cite{Elcott:2007:SCP} combines a vorticity formulation on closed loops with semi-Lagrangian backtracing to advance the fluid solution, unfortunately dissipating a good deal of information during the necessary transfer between the traced readings and the underlying grid.
Later, \cite{Nabizadeh:2022:CF} 
revisited this approach 
from the perspective of 
the velocity covector, making 
the semi-Lagrangian advection 
circulation-preserving. Together with 
back-and-forth error compensation and correction (BFECC)) \cite{Dupont:2003:BFECC}, this approach achieved enough error correction to keep this structure intact---though energy is still dissipated and circulation is ultimately lost. In \cite{Nabizadeh:2024:FIP}, the authors formulated discrete Euler equations that fit naturally into a hybrid grid--particle framework; combined with IGA FEEC B-splines, this yields a high-order structure-preserving scheme that preserves Casimir invariants along with the energy on the grid.
Finally, \cite{Nabizadeh:2025:thesis} showed that this setup is in fact a vakonomic formulation, consistent with the sub-Riemannian theory.

\subsubsection{Matrix hydrodynamics.}
An advantage of the hybrid methods in
\cite{Nabizadeh:2024:FIP,Nabizadeh:2025:thesis}
is that 
sub-Riemannian geodesics 
can be simulated directly on 
the infinite-dimensional Lie algebra 
$\sdiff(M)$
of divergence-free vector fields through a particle discretization. In contrast, Eulerian methods like the present work 
are restricted to the finite-dimensional discrete Lie algebra $\so(n)$ of skew-symmetric $n \times n$ matrices. 
While this setting is necessarily incomplete from the perspective of the continuous Euler equations,
it offers a very concrete and precise way of relating the discrete dynamics to geodesic flow on an appropriate (sub-)Riemannian manifold. 
In fact, The vakonomic discrete Euler equations in Lax form \eqref{eq:VakonomicLaxEquation} are also closely related to the \emph{Euler--Zeitlin} equations~\cite{Zeitlin:1991:FMA}, which have brought about a simulation paradigm known as matrix hydrodynamics (MHD).
Instead of a finite-dimensional approximation to the fluid representation space $\SDiff(M)$, these approaches utilize a finite-dimensional approximation to the Lie algebra $\sdiff(M)$ generated through geometric quantization.  The advantage of this is that no nonholonomic constraint is required: the distribution where discrete fluid velocities live is integrable by construction.  This leads to an Euler--Arnold system with Lie-Poisson structure, improved statistical behavior, and more correct energy cascading over time.  On the other hand, this quantized formulation is currently limited to spheres and tori~\cite{Cifani:2023:EGM,Modin:2024:TDF}, prone to excessive dispersion, and yields only $\mathcal{O}(n^{-1/2})$ convergence with increasing matrix size.  On the flat torus, the quantization breaks the translational symmetry and introduces artifacts; on the sphere, the rotational symmetry is preserved, enabling a more robust program. Consequently, the recent work of \cite{Cifani:2023:EGM,Modin:2024:TDF,Modin:2025:MHD} has focused on $\mathbb{S}^2$, using it to study long-time turbulent cascades with exact Casimir conservation; 3D extensions remain partial and limited to symmetric cases. 
A more detailed description of matrix hydrodynamics can be found in \Cref{app:MHD}.

\subsubsection{Isospectral flows.}
Besides the related approaches to fluid simulation already discussed, it is worth mentioning the extensive literature on the simulation of isospectral flows \cite{Calvo:1996:RKO,Calvo:1997:NSI,Modin:2020:LPM}. The MHD equations of motion, along with the vakonomic equations \eqref{eq:VakonomicLaxEquation} presented here,
are isospectral as a consequence of their Lax formulation, which implies motion on a single coadjoint orbit along with conservation of Casimirs as discussed in \Cref{subsec:hamiltonian}.
Generic time integrators destroy this property: \cite{Calvo:1996:RKO,Calvo:1997:NSI} therefore develop Runge--Kutta-type methods that preserve isospectrality, while \cite{Modin:2020:LPM} provides Lie--Poisson integrators for isospectral flows that, in addition to keeping the flow on a single coadjoint orbit, preserve the symplectic structure on that orbit. 
The Cayley-based update discussed in \Cref{sec:time-integration} plays an analogous role in the present setting, keeping the flow on the orthogonal group and preserving the Lax structure of the equations of motion.

\section{A Low-Rank Parameterization}
\label{sec:LowRankParameterization}

While the vakonomic equation \eqref{eq:Z-vakonomic-final} enjoys many benefits, simulating $\bZ$ requires evolving $O(F^2)$ degrees of freedom and can rapidly become untenable in realistic scenarios.  Thankfully, there is a useful low-rank ansatz that solves the Lax equation with a much smaller computational cost of $O(mF)$, where $m\ll F$ represents a number of dynamically evolving vector variables.
First, note the following result.

\begin{theorem}
\label{thm:lowrankansatz}
    Let $m\in\mathbb{Z}$ and let $\{\blambda_a,\bmu_a\}_{a=1}^m$ be $2m$ discrete fields in $\mathbb{R}^F$ satisfying the evolution equation
    \begin{equation}
    \label{eq:VakonomicLowRank}
        \dot{\blambda}_a = \bA_\bu \blambda_a, \quad
        \dot{\bmu}_a = \bA_\bu \bmu_a, \quad \bA_\bu = \left(\frac12 \sum_{a=1}^m\blambda_a\bmu_a^\intercal-\bmu_a\blambda_a^\intercal\right)^\sharp.
    \end{equation}
    Then, the matrix 
    \begin{equation}
        \bZ = \frac{1}{2} \sum_{a=1}^m \blambda_a\bmu_a^\intercal-\bmu_a\blambda_a^\intercal \eqqcolon \sum_{a=1}^m\bS_a,
        \label{eq:Z-Clebsch}
    \end{equation}
    satisfies the vakonomic Lax equation $\dot\bZ = [\bZ^\sharp,\bZ]$.

\end{theorem}

\begin{proof}
    Let $\bB_a = \blambda_a\bmu_a^\intercal,$ so that $\bS_a = \frac{1}{2} \left(\bB_a - \bB_a^\intercal\right)$. Observe that $\bB_a$ evolves according to
    \begin{align*}
        \dot{\bB}_a &= \dot{\blambda}_a\bmu_a^\intercal + \blambda_a\dot{\bmu}_a^\intercal = \bA_\bu\blambda_a\bmu_a^\intercal + \blambda_a\bmu_a^\intercal\bA_{\bu}^\intercal \\
        &= \bA_\bu\bB_a-\bB_a\bA_\bu = [\bA_\bu,\bB_a].
    \end{align*}
    It follows that the corresponding evolution of $\bS_a$ is 
    \begin{align*}
        \dot{\bS}_a &= \tfrac{1}{2} (\dot{\bB}_a-\dot{\bB}_a^\intercal) = \tfrac{1}{2} ([\bA_\bu,\bB_a]-[\bA_{\bu},\bB_a]^\intercal) \\
        &= [\bA_\bu,\frac{1}{2}(\bB_a-\bB_a^\intercal)] = [\bA_\bu,\bS_a].
    \end{align*}
    Therefore, the evolution of $\bZ = \sum_a\bS_a$ is 
    \[ \dot\bZ = [\bA_{\bu}, \bZ] = [\bZ^\sharp, \bZ], \]
    as desired. 
\end{proof}

\Cref{thm:lowrankansatz} enables rapid simulation of the isospectral flow \eqref{eq:VakonomicLaxEquation} by restricting the dynamical variable $\bZ$ to a space of low-rank matrices.  In fact, it will soon be shown that there is no need to construct the dense matrix $\bA_{\bu}$ at all (\cf\,\eqref{eq:cheap-momentum-map}), and therefore it is not necessary to perform any operations which scale with $O(F^2)$.  This is a consequence of interpreting $\{(\blambda_a,\bmu_a)\}$ as a momentum map representation of the dynamics, where these pairs are identified with (discrete) \emph{Clebsch variables} \cite{Clebsch:1859:IHG,Marsden:1983:COV}, \ie, coefficients in the half-density basis for $\tilde{\cF}$. 
In the usual continuous setting where functions are considered rather than half-densities, this Clebsch representation is $\vec u = \sum_{a=1}^m\vec\mu_a\nabla\lambda_a$ where each $\lambda_a$ can be identified with a component of the inverse flow map $\phi^{-1}\in\Diff(M)$ and each $\vec\mu_a$ is associated with a corresponding component of $\vec u_0\circ\phi^{-1}$. This is equivalent to the pullback of 1-forms $\eta_t = (\phi^{-1})^* \eta_0$, i.e., $\vec u_t = \nabla\phi^{-\top} (\vec u_0 \circ \phi^{-1})$ in three dimensional Euclidean space, where $\eta$ is the velocity covector \cite{Nabizadeh:2022:CF}.  Besides providing a scalable representation of the desired Euler dynamics, this also provides a simple initialization strategy that will be discussed in \Cref{subsec:resetting}.

\begin{remark}
    Since $\bA_\bu$ is a skew-symmetric operator, it preserves the norm of each $\blambda_i$ and $\bmu_i$.  Indeed, for any function
    $\bff\in\RR^{F}$ that evolves as $\dot\bff=\bA_\bu\bff$, we have
    $$\frac{d}{dt} \| \bff\|^2 = \llangle \dot{\bff},\bff\rrangle + \llangle \bff,\dot{\bff}\rrangle =  \bff^\intercal(\bA_\bu^\intercal + \bA_\bu)\bff = 0.$$
\end{remark}

The momentum map representation \eqref{eq:Z-Clebsch} in terms of Clebsch variables provides a natural choice of $m$, and 
also enables an efficient computation of its spectral invariants, \ie, Casimirs.  Since it is known that $d$-dimensional fluids require at most $(d-1)$ Clebsch components \cite{Yoshida:2009:CPB}, we choose $m=d$ where $d=2,3$ is the dimension of the ambient space where the simulation takes place, providing enough flexibility to generate the realistic-looking fluid trajectories seen in \cref{sec:experiments}. 
For the computation of Casimirs, consider the block matrices
\[
\bU :=
[\blambda_1\,\bmu_1\!\cdots\!\blambda_m\,\bmu_m]
\in\RR^{F\times 2m},
\quad
\bJ_{2m}:=\tfrac12\mathrm{blkdiag}(\bJ_2,\dotsc,\bJ_2),
\]
where $ \bJ_2=\begin{psmallmatrix}0&1\\-1&0\end{psmallmatrix}$ is the canonical symplectic matrix of dimension two. It follows that $\bZ$ has the alternative expression
\[
\bZ = \bU\,\bJ_{2m}\,\bU^\intercal,
\]
and, in particular, $\rank(\bZ)\le 2m$. Spectral invariants of $\bZ$ can be efficiently computed by observing that
 $\operatorname{tr}(\bZ^{2k}) = \operatorname{tr}\big((\bU\bJ_{2m}\bU^\intercal)^{2k}\big) = \operatorname{tr}\big((\bJ_{2m}\bU^\intercal\bU)^{2k}\big)$ where $\bJ_{2m}\bU^\intercal\bU \in \mathbb{R}^{2m\times 2m}$ is small.  In fact, observe that each of $\blambda_i,\bmu_i$ evolves according to the orthogonal operator $\bR = {\rm exp}(t\bA_{\bu})\in SO(F)$, and therefore all entries $[\bU^\intercal\bU]_{ij} = \blambda_i\cdot\bmu_j$ of the ``covariance matrix'' $\bU^\intercal\bU$ are constant along solutions to \eqref{eq:VakonomicLowRank}.  In particular, this leads to Casimirs of the form 
\begin{equation}\label{eq:CasimirForm}
     C_{2k}(\bZ) = \mathrm{tr}(\bZ^{2k}) = \tr(\bB^{2k}),
 \end{equation}
 where $\bB = \bJ_{2m}\bU^\intercal\bU\in\mathbb{R}^{2m\times 2m}$ is defined as the block matrix
 \begin{equation}
    \bB = [\bB_{ij}]_{i,j=1}^m, \quad \bB_{ij} = \tfrac 12 \begin{pmatrix}\blambda_i\cdot\bmu_j & \bmu_i\cdot\bmu_j \\ -\blambda_i\cdot\blambda_j & -\blambda_i\cdot\bmu_j \end{pmatrix}.
    \label{eq:blkdiagBwithLambdaAndMu}
\end{equation}

\subsection{Symplectic space properties}
As mentioned, it is useful to view the Clebsch parameterization and \Cref{thm:lowrankansatz} as a momentum map representation of the dynamics induced by canonical Hamiltonian evolution on a discrete symplectic vector space.   
To see this, 
consider the continuous setting and the infinite-dimensional vector space $\cP := \cH\cD(M) \times \cH\cD(M)$ whose elements are pairs $(\lambda,\mu)$ of smooth half density fields on $M$.  Each pair thus generates a single component of the representation $\vec u = \sum_a\nabla\lambda_a\mu_a$.  
The canonical Liouville $1$-form on $\cP$ is defined by
\begin{equation*}
    \big\llangle \vartheta({\lambda},{\mu}) \,\big|\, (\bigmathring{\lambda},\bigmathring{\mu})\big\rrangle = \int_M \mu\,\bigmathring{\lambda}
\end{equation*}
for any tangent vector $(\bigmathring{\lambda},\bigmathring{\mu})\in T_{(\lambda,\mu)}\cP\cong\cP$.
Using the standard identity for the exterior derivative of a $1$-form,
\begin{equation*}
    \langle d\alpha\,|\,X,Y\rangle = X\langle\alpha\,|\,Y\rangle - Y\langle\alpha\,|\,X\rangle - \langle \alpha\,|\,[X,Y]\rangle
\end{equation*}
together with the fact that $\cP$ is a vector space and hence there are no nontrivial Lie brackets, 
the symplectic $2$-form $\sigma = d\vartheta$ is given by:
\begin{equation*}
    \big\llangle \sigma\,\big|\,(\dot{\lambda},\dot{\mu}),(\bigmathring{\lambda},\bigmathring{\mu})\big\rrangle
    =
    \int_M \big(\dot{\mu}\,\bigmathring{\lambda} - \bigmathring{\mu}\,\dot{\lambda}\big),
\end{equation*}
for tangent vectors $(\dot\lambda,\dot\mu)$ and $(\bigmathring{\lambda},\bigmathring{\mu})$.
It is straightforward to verify that $\sigma$ is closed, non-degenerate, skew-symmetric, and bilinear. Therefore, $\sigma$ defines a canonical symplectic structure on $\cP$, hence a canonical symplectic structure on the product $\mathcal{P}^m$ through the direct sum $\sigma\oplus ... \oplus \sigma$.

To bring this to the discrete setting, consider the projection of $\sigma$ onto the subspace $\tilde{\cF}\subset\HD(M)$, which has the following representation in terms of coefficients in $P \cong \RR^F\times \RR^F$:
\begin{equation*}
    \big\llangle \sigma\,\big|\,(\dot{\lambda},\dot{\mu}),(\bigmathring{\lambda},\bigmathring{\mu})\big\rrangle
    =
    (\dot{\blambda},\dot{\bmu})^\intercal \bJ(\bigmathring{\blambda},\bigmathring{\bmu}), \quad \bJ = \begin{bmatrix} 0 & \bI \\ -\bI & 0\end{bmatrix}.
\end{equation*}
The symplectic matrix $\bJ$ is invariant under orthogonal transformations.
In particular, if $\bR\in\SO(F)$ satisfies $\bR^\intercal \bR = \bI,$ then the induced left action on the symplectic space 
$$
    (\blambda,\bmu)\mapsto(\bR\blambda,\bR\bmu),
$$
preserves the discrete symplectic structure on $P$.  Observe that the infinitesimal generator of this action is given by the vector field 
\begin{equation*}
    \xi_{\bA}(\blambda,\bmu) = (\bA\blambda,\bA\bmu), \quad \bA\in\so(F).
\end{equation*} 
Conversion between this auxiliary symplectic structure and the primary one on $\so(F)$ is facilitated by the \textit{momentum map} $J:P\to\so(F)^*\cong\so(F)$, which is defined (up to constants) through the equality
\[\mathrm{d}\langle J(\blambda,\bmu)\,|\bA\rangle = \sigma\big(\xi_\bA(\blambda,\bmu),\cdot\big) = (\blambda,\bmu)^\intercal\bJ(\bA,\bA),\]
where $(\bA,\bA)(\blambda,\bmu) = (\bA\blambda,\bA\bmu) = \xi_{\bA}(\blambda,\bmu)$ and the final equality uses that the block matrices $\bJ$ and $(\bA,\bA)^\intercal=-(\bA,\bA)$ commute.  Since $\bJ(\bA,\bA)$ is a symmetric matrix, it follows that the action of the momentum map has a simple representation in terms of a quadratic form.

\begin{proposition}
The momentum map $J\colon P\to\so(F)^*\cong\so(F)$ has the equivalent representation
\begin{equation}
    \label{eq:momentum-map-def}
    \langle J(\blambda,\bmu)\,|\, \bA \rangle
    =
    \tfrac12(\blambda,\bmu)^\intercal\bJ\, \xi_{\bA}(\blambda,\bmu),
    \quad
    \forall \bA\in\so(F).
\end{equation}
In particular, $J(\blambda,\bmu) = \tfrac12(\blambda\bmu^\intercal - \bmu\blambda^\intercal)$.
\end{proposition}
\begin{proof}
It is easy to check that differentiating \eqref{eq:momentum-map-def} satisfies the defining relationship for the momentum map.  To compute an explicit representation for $J$, observe that 
\begin{align*}
    2\langle J(\blambda,\bmu), \bA\rangle_{\rm Frob} &= \langle (\blambda,\bmu), \bJ(\bA\blambda,\bA\bmu)\rangle = -\langle\bJ(\blambda,\bmu),(\bA\blambda,\bA\bmu)\rangle \\
    &= \langle(-\bmu,\blambda),(\bA\blambda,\bA\bmu)\rangle = \langle \blambda\bmu^\intercal-\bmu\blambda^\intercal, \bA\rangle_{\rm Frob},
\end{align*}
which yields that $J(\blambda,\bmu) = \tfrac12(\blambda\bmu^\intercal - \bmu\blambda^\intercal)$.
\end{proof}

Recall that the discrete representation in \Cref{thm:lowrankansatz} has $m$ components.  Therefore, the momentum map with respect to the direct sum symplectic structure on $P^m$ is given (with some abuse of notation) by 
\[J(\blambda,\bmu) = \sum_{a=1}^m J(\blambda_a,\bmu_a) = \frac12\sum_{a=1}^m\blambda_a\bmu_a^\intercal-\bmu_a\blambda_a^\intercal.\]
It is also necessary to compute the adjoint of this momentum map with respect to the mapping $\bu \mapsto\bar\bA_{\bu}$ between the coefficient space $V_{\div}$ and its image $\im(\bar\bA)\subset \so(F)$ in the ambient Lie algebra,
\begin{equation*}
\begin{split}
    2\langle J(\blambda,\bmu)\,|\,\bar{\bA}_{\bu}\rangle &= \sum_{a=1}^m\sum_{k=1}^{V}u_k\langle\bA_k, \blambda_a\bmu_a^\intercal-\bmu_a\blambda_a^\intercal\rangle \\
    &= \sum_{a,k}u_k\big(\blambda_a^\intercal\bA_k\bmu_a - \bmu_a^\intercal\bA_k\blambda_a\big) \\
    &= \big\langle \sum_{a}\blambda_a^\intercal\bA(\cdot)\bmu_a - \bmu_a^\intercal\bA(\cdot)\blambda_a, \bu\big\rangle \\
    &= 2\langle \bar{\bA}^*J(\blambda,\bmu)\,|\,\bu\rangle.
\end{split}
\end{equation*}
It follows that the adjoint is given by 
\begin{equation}\label{eq:cheap-momentum-map}
    \bar{\bA}^*J(\blambda,\bmu) = \tfrac12\sum_{a=1}^m\blambda_a^\intercal\bA(\cdot)\bmu_a-\bmu_a^\intercal\bA(\cdot)\blambda_a = \sum_{a=1}^m\blambda_a^\intercal\bA(\cdot)\bmu_a,
\end{equation}
where the final equality used that $\bA_\bu$ is skew-symmetric for all $\bu\in V_{\div}$. Therefore, the action of the momentum map $J$ on the Lie algebra element $\bar{\bA}_\bu$, along with its adjoint operation, can be computed through combinations of matrix-vector products without requiring expensive matrix-matrix operations.  For notational convenience, the following sections will refer to the following reconstruction routine that computes the element of $V_{\div}$ associated to the Clebsch variables $\{(\blambda_a,\bmu_a)\}$: 
\begin{equation}\label{eq:VEL}
\begin{split}
    \textsc{Vel}(\blambda,\bmu) &= J(\blambda,\bmu)^\sharp = \bK^{-1}\bar\bA^*J(\blambda,\bmu) \\
    &= \sum_{a=1}^m \bK^{-1}(\blambda_a^\intercal\bA(\cdot)\bmu_a).
\end{split}
\end{equation}
Here, as before, $\bA(\cdot)$ is used to indicate the (linear) map $\bu \mapsto \bA_\bu$.

\section{Time integration}\label{sec:time-integration}

The vakonomic equations of motion \eqref{eq:VakonomicLaxEquation} describe the time-continuous motion of an ideal fluid approximately represented on a space of discrete half-densities. To form a tractable computational model useful for numerical simulation, it remains to discretize in time so that the delicate geometric structure of the vakonomic Lax equation is preserved.  In particular, we must ensure that the temporal discretization remains Hamiltonian.  Under the low-rank parameterization in \cref{sec:LowRankParameterization}, this will translate to rigid-body dynamics in $F$ dimensions, where a structure-preserving integrator is given by a mild extension of the method developed for $F=3$ in \cite{Engo:2001:NILP}.

Methods that preserve the structure of coadjoint orbits belong to the class of Lie group integrators (\cf, \cite{Iserles:2000:LGM, Celledoni:2014:LGI}) and produce discrete trajectories guaranteed to respect this structure.
This work considers a particular class known as Runge--Kutta--Munthe-Kaas (RKMK) \cite{Munthe:1999:HOR} methods that evolve Lax systems such as \eqref{eq:VakonomicLaxEquation} via coadjoint actions and admit arbitrarily high-order integration schemes.  Specifically, we apply the Lie-trapezoidal rule from Eng\o{} \& Faltinsen \cite{Engo:2000:CGI, Engo:2001:NILP}.

More precisely, the approximate solution \(\bZ^{(n)}\)  at \(t = n\Deltait t\) is updated according to the conjugation \begin{align}
\label{eq:TimeIntegration}
    \bZ^{(n+1)} = \bR^{(n)} \bZ^{(n)} \bR^{(n)\intercal},\quad \bR^{(n)} = \textstyle\widetilde{\exp}\left(\Deltait t (\frac{\bZ^{(n)} + \bZ^{(n+1)}}{2})^\sharp\right)
\end{align}
where \(\widetilde{\exp}\colon \so(F)\to \SO(F)\) is either the matrix exponential or a (computationally cheaper) analytic function that approximates it. Here, a useful approximation is provided by the Cayley map \(\Cay\colon\so(F)\to\SO(F)\):
\begin{align}
    \textstyle\Cay\left({1\over 2}\bX\right) = \textstyle \left(\bI - {1\over 2}\bX\right)^{-1}\left(\bI + {1\over 2}\bX\right),
\end{align}
which is a second-order approximation to $\exp(\bX)$.
The assumption that the algebra-to-group map \(\widetilde{\exp}\) is analytic ensures that the vector field \(\bX\) commutes with its pseudo-exponential \(\widetilde{\exp}(\bX)\).
In fact, the formula \eqref{eq:TimeIntegration} is enough to establish the following  important property of Lie-trapezoidal time integration. 
\begin{theorem}\label{thm:TimeIntegration}
    The map \eqref{eq:TimeIntegration} from \(\bZ^{(n)}\) to \(\bZ^{(n+1)}\) preserves the coadjoint orbit of $\bZ^{(n)}$ and the Casimirs of the Lie-Poisson bracket \eqref{eq:LiePoisson}. Moreover, the Hamiltonian 
    \[H^{(n)} = H(\bZ^{(n)}) = {1\over 2}\big\langle\bZ^{(n)}|\bZ^{(n)\sharp}\big\rangle,\] is conserved: \(H^{(n+1)} = H^{(n)}\).
\end{theorem}
\begin{proof}
    This map takes the form \(\bZ\mapsto \bR\bZ\bR^\intercal\) for some \(\bR\in\SO(F)\).  Therefore, it is a coadjoint action and also  isospectral, implying the first two properties.  To establish Hamiltonian preservation, consider the midpoint element \(\bX = {1\over 2}(\bZ^{(n)}+\bZ^{(n+1)})^\sharp\in\so(F)\).
    Then, \(\bR^{(n)}= \widetilde{\exp}(\Deltait t\bX) \in \SO(F)\) commutes with \(\bX\), and therefore 
    \begin{align*}
        H^{(n+1)} - H^{(n)} &= {1\over 2}\big\langle \bZ^{(n+1)}-\bZ^{(n)}\mid(\bZ^{(n+1)}+\bZ^{(n)})^\sharp\big\rangle \\
        &= \big\langle \bR^{(n)}\bZ^{(n)}\bR^{(n)\intercal} - \bZ^{(n)} \mid \bX\big\rangle \\
        &= \big\langle \bZ^{(n)}\mid\bR^{(n)\intercal}\bX\bR^{(n)} - \bX\big\rangle = 0,
    \end{align*}
    yielding the conclusion.
\end{proof}
Despite its geometric utility, directly implementing \eqref{eq:TimeIntegration} in the present setting is not practical, as it requires assembling the dense matrix $\bZ$ at every time step.  Thankfully, the momentum map representation from \cref{sec:LowRankParameterization} removes this requirement, enabling the equivalent Lie-trapezoidal integration scheme for the variables \(\blambda_a^{(n)},\bmu_a^{(n)} \in \mathbb{R}^F\) at time \(t = n\Deltait t\) in terms of the mapping \eqref{eq:VEL}, 
\begin{align}
\label{eq:ClebschTimeIntegration}
\begin{aligned}
     &\blambda_a^{(n+1)} = \bR^{(n)}\blambda_a^{(n)},\quad 
    \bmu_a^{(n+1)} = \bR^{(n)}\bmu_a^{(n)},\\
    &\bR^{(n)} = 
    \textstyle
    \widetilde{\exp}\left({\Deltait t\over 2}\bar{\bA}\left(\textsc{Vel}(\blambda^{(n)},\bmu^{(n)})+\textsc{Vel}(\blambda^{(n+1)},\bmu^{(n+1)})\right)\right).
\end{aligned}
\end{align}
Of course, there is an equivalent guarantee on the desired structural properties of \eqref{eq:ClebschTimeIntegration} as a Corollary to \Cref{thm:TimeIntegration}.
\begin{corollary}
    The map \eqref{eq:ClebschTimeIntegration} from $\bZ^{(n)}$ to $\bZ^{(n+1)}$ preserves the Hamiltonian $H^{(n)}$ along with the Casimirs \eqref{eq:CasimirForm}.
\end{corollary}
With this, \eqref{eq:ClebschTimeIntegration} provides an efficient Lie group 
integration scheme for the low-rank vakonomic equations \eqref{eq:VakonomicLowRank} implemented in the following section.

\begin{remark}\label{rmk:TimeIntegrationPoissonAutomorphism}
    In fact, the map  \(\bZ\mapsto \bR\bZ\bR^\intercal = \Ad_{\bR^\intercal}^*\bZ\) is a Poisson automorphism.  That is, besides preserving the coadjoint orbits of the Lie--Poisson structure \eqref{eq:LiePoisson}, its restriction to each of these orbits is a symplectomorphism. To see this, recall that Hamiltonian flows are Poisson automorphisms, and express  \(\bR = \exp(T\bX)\) for some \(\bX\in\so(F)\) and \(T\in\RR\).  Then, \(\bZ\mapsto\bR\bZ\bR^{\intercal}\) is the flow map from \(t=0\) to \(t=T\) corresponding to the linear Hamiltonian \(H_{\bX}(\bZ) = \langle\bX|\bZ\rangle\).
\end{remark}
\begin{remark}
    \Cref{thm:TimeIntegration} and \Cref{rmk:TimeIntegrationPoissonAutomorphism} appear to suggest that the time integrator \eqref{eq:TimeIntegration} is both energy-preserving and symplectic.  However, this is unlikely: the well known result of Zhong and Marsden \cite{Zhong:1988:LPI} states that an energy-preserving and symplectic integrator must generate exact solutions when the Hamiltonian system is maximally reduced (i.e., no other functions Poisson-commute with the energy). This is because a symplectic integrator is a one-parameter (\(\Deltait t\)) family of symplectic maps and therefore locally a Hamiltonian flow; energy-preservation and maximal reduction imply that the Hamiltonian of the integrator coincides with the energy of the system, forcing exactness in the integration.
    \Cref{thm:TimeIntegration} does not contradict this result
    since \(\bZ\mapsto \bR\bZ\bR^{\intercal}\) is only symplectic for a fixed \(\bR\in\SO(F)\) independent of $\bZ$:  the overall update map $\bZ^{(n)} \mapsto \bZ^{(n + 1)}$ in \eqref{eq:TimeIntegration} is \textit{not} a symplectomorphism  because the choice of $\bR^{(n)}$ also depends on $\bZ^{(n)}$.  Accordingly, the vakonomic integrator \eqref{eq:ClebschTimeIntegration} used here is energy-preserving but not symplectic.
\end{remark}

\begin{figure}[h]
    \centering
    \includegraphics[width=\linewidth]{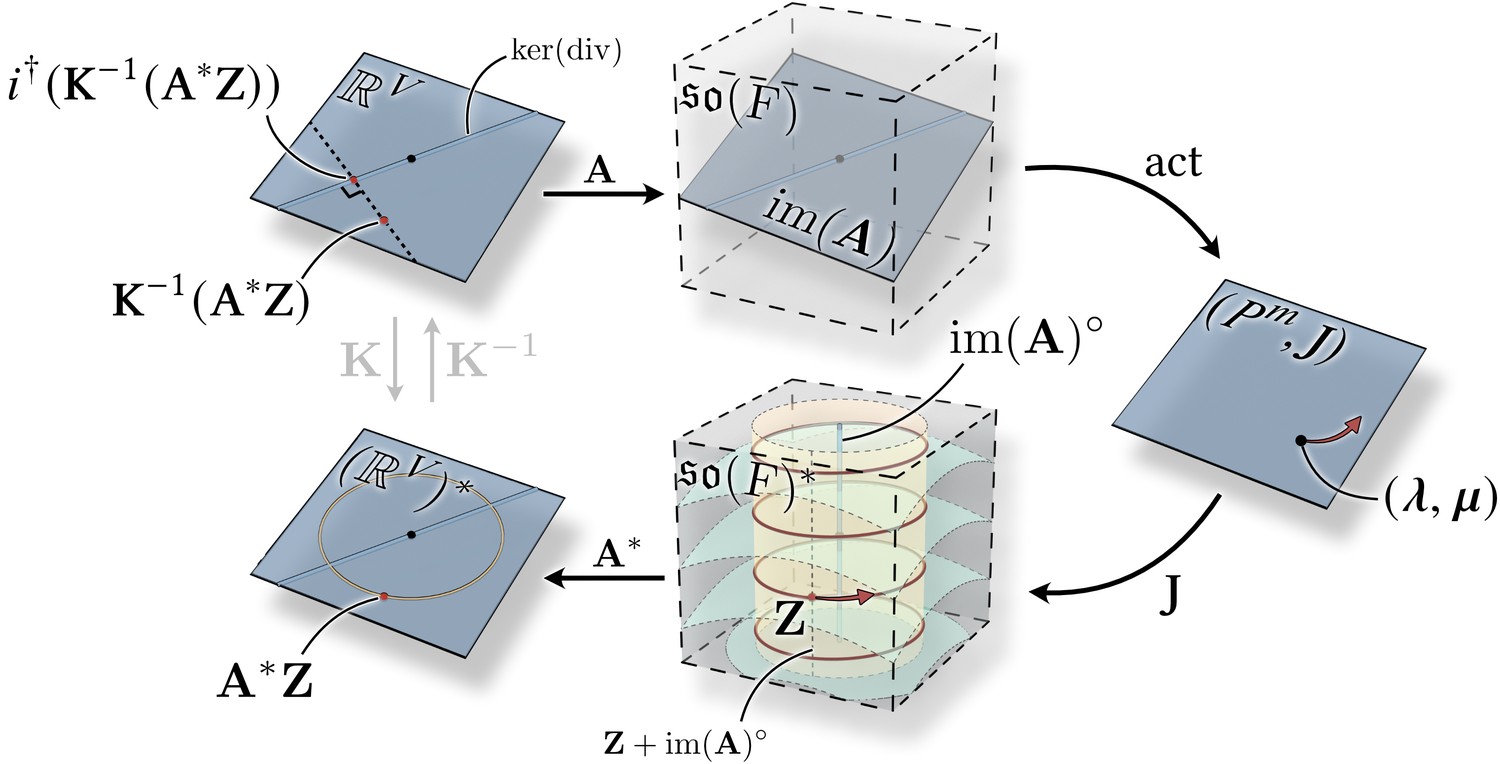}
    \caption{%
    The pipeline of our method, Vakonomic Fluids. The auxiliary symplectic space $(P^m,\mathbf{J})$ %
    is identified with the dual Lie algebra $\so(F)^*$, on which the equations of motion take an identical form. There, the dynamics is a Lie--Poisson system evolving along the intersection of a pulled-back energy level set by $\bA^*$ in the discrete velocity space $(\RR^V)^*$ with the coadjoint orbits induced by the matrix Lie group structure. Consequently, the symplectic variables evolve under the action ${\rm act}(\bA) = \xi_{\bA}$ of matrices in the subspace $\im(\bA)$. 
    }
    \label{fig:SimLoop}
\end{figure}

\section{Algorithm and Implementation}\label{sec:algorithm}

This section provides a detailed description of the vakonomic fluid simulation (\secref{sec:algorithm}), including a discussion of the end-to-end simulation pipeline and critical implementation details.  Precisely, the proposed algorithm evolves the Clebsch variables $(\blambda,\bmu)$ representing the discrete incompressible  velocity $\bu\in V_{\div}$ using Eqn.~\eqref{eq:ClebschTimeIntegration}.  This ensures that the discrete incompressible Euler equations are solved within the discrete Koopman framework outlined in \Cref{tab:Notation} and in a way that is consistent with the continuous theory in \Cref{sec:background}.  In particular, the solutions to \algref{alg:sim_loop_no_reset} retain their interpretation as geodesics on a (sub-)Riemannian manifold of discrete diffeomorphisms.

\subsection{Algorithm}

The general structure of each timestep is summarized in \Cref{fig:SimLoop}.  Each time update starts from a set of Clebsch variables \((\blambda^{(n)},\bmu^{(n)})\in\RR^{F}\times\mathbb{R}^F\), and solves a fixed point problem, that iterates the following two steps: first, the Clebsch variables are used to recover a velocity $\bu^{(n)} = \textsc{Vel}\big(\blambda^{(n)},\bmu^{(n)}\big)$, and then, that velocity is used to advect the Clebsch variables using the implicit Lie-trapezoidal integrator, resulting in $(\blambda^{(n+1)},\bmu^{(n+1)})$. Altogether, each time step is described in \algref{alg:step}.

Implementation separates the recovery of velocity into two steps,  \(\textsc{Vel} = \textsc{Project}\circ\textsc{Reconstruct}\), where $\textsc{Reconstruct}$ corresponds to the momentum map adjoint $\bA^* J : P \to (\RR^V)^*$, and $\textsc{Project}$ is the familiar pressure projection operation that extracts the divergence-free component of a vector field.  Because we choose an $H(\div)$-conforming finite element space for $\cV$, we are able to use the discrete Hodge decomposition machinery of Finite Element Exterior Calculus (FEEC) \citep{Arnold:2006:FEEC, Arnold:2010:FEEC}, and perform a streamfunction solve to recover the divergence-free component of $\bu$ (see \secref{sec:FEECandVdiv}). 

Altogether, our algorithm only requires implementations of \textsc{Advect}, \textsc{Reconstruct}, and \textsc{Project}. The precise implementation of these routines only requires a choice of two function spaces $\cF$ and $\cV$. In the next section, we describe our implementations with two such choices: Whitney forms on triangle meshes, and B-Splines from isogeometric analysis. 

\vspace{1pc}

\begin{algorithm}
    \caption{Vakonomic Fluids loop}
    \begin{algorithmic}[1]
        \State Initialize $\blambda^{(0)}$ and $\bmu^{(0)}$
        \For{$n \in 1, 2, \dots$}
            \State $(\blambda, \bmu)^{(n)} \gets \textsc{Step}\left((\blambda, \bmu)^{(n-1)}\right)$
        \EndFor
    \end{algorithmic}
    \label{alg:sim_loop_no_reset}
\end{algorithm}

\begin{algorithm}
    \caption{$\textsc{Step} (\blambda^{(n)}, \bmu^{(n)})$: Step of implicit trapezoidal integrator.
    }
    \begin{algorithmic}[1]
        \Require Discrete Clebsch variables $(\blambda, \bmu)^{(n)} \in P^m$
        \State $\bu^{(n)} = \bu^\star \gets \textsc{Project} (\textsc{Reconstruct} (\blambda, \bmu)^{(n)})$
        \Repeat \Comment{Fixed point iteration on $\bu^\star$}
            \State $(\blambda, \bmu)^{(n + 1)} \gets \textsc{Advect}\left(\bu^\star, (\blambda, \bmu)^{(n)}\right)$
            \State $ \hat{\bu} \gets \textsc{Reconstruct}\left((\blambda, \bmu)^{(n + 1)}\right) $
            \State $\bu^{(n + 1)} \gets \textsc{Project}(\hat{\bu}) $
            \State $\bu^\star \gets \frac12 (\bu^{(n)} + \bu^{(n + 1)}) $
        \Until $\bu^\star$ not converged 
        \Ensure $(\blambda, \bmu)^{(n + 1)}$
    \end{algorithmic}
    \label{alg:step}
\end{algorithm}

\begin{algorithm}
\caption{\textsc{Reconstruct}$(\blambda,\bmu)^{(n)}$}
\begin{algorithmic}[1]
    \Require Discrete Clebsch variables $(\blambda, \bmu)^{(n)} \in P^m$
    \For{$k \in 1, \dots, V$}
        \State $\hat{\bu}_k \gets \displaystyle \sum_{a=0}^m\blambda_a^\intercal\bA_k\bmu_a$
    \EndFor
    \Ensure $\hat{\bu}$
\end{algorithmic}
\label{alg:reconstruct}
\end{algorithm}
\begin{algorithm}
\caption{\textsc{Advect}$(\bu, \bs)$ }
\begin{algorithmic}[1]
\Require Velocity $\bu \in \RR^V$, quantities to advect $\bs \in (\RR^F)^s$
  \State  $\bA_\bu \gets \textrm{Eqn.}~\eqref{eq:discrete-adv-operator}$ %
  \For{$a \in 1, ..., s$}
  \State  $\overline{\bs_a} \gets \widetilde{\exp}(\Delta t \bA_\bu) \bs_a$ 
  \EndFor
  \Ensure $\overline{\bs}$
\end{algorithmic}
\label{alg:advect-exp}
\end{algorithm}

\begin{algorithm}
\caption{\textsc{Project}$(\hat{\bu})$ --- Leray projection via streamform--vorticity solve}
\label{alg:project-stream}
\begin{algorithmic}[1]
\Require Dual velocity $\hat{\bu} \in \RR^V$
  \State Solve $\ \mathbf{curl}^{\!\top}\bK\,\mathbf{curl}\,\bpsi
        = \mathbf{curl}^{\!\top}\hat\bu$
  \State $\bu \gets \mathbf{curl}\,\bpsi$
  \If{$b_{n-1}(M) > 0$}
    \Comment{if $M$ is not simply connected}
    \State $\bu \gets \bu + \Call{HarmonicProj}{\hat\bu}$
  \EndIf
  \Ensure $\bu$
\end{algorithmic}
\end{algorithm}

\subsection{Implementation Details}\label{subsec:implementation}
In this section we discuss two FEEC-based implementations based on the algorithms in \secref{sec:algorithm}. 
The first employs B-spline bases suitable for high-order simulations on grids, while the second uses first-order Whitney forms applicable to arbitrary triangle meshes. Both implementations provide a \emph{mimetic} discretization where the discrete differential operators reproduce the identities of the continuous de Rham complex exactly, so that the resulting velocity field is pointwise (\ie, strongly)
divergence free. We will see how this allows for an efficient computational algorithm based on function representations of the Clebsch variables driving our simulations. 

\subsubsection{FEEC and $V_{\div}$}
\label{sec:FEECandVdiv}

The FEEC construction, developed by Arnold, Falk, and Winther \cite{Arnold:2006:FEEC, Arnold:2010:FEEC, Arnold:2018:FEEC}, provides a structure-preserving framework for discretizing partial differential equations posed on differential forms. The central object is the discrete de Rham complex that is related to its continuous counterpart by a commuting diagram of bounded projection operators:
\[
\begin{tikzcd}[column sep=large, row sep=large]
\Lambda^0 \arrow[r, "\mathrm{grad}"] \arrow[d, "\Pi^0"'] 
    & \Lambda^1 \arrow[r, "\mathrm{curl}"] \arrow[d, "\Pi^1"'] 
    & \Lambda^2 \arrow[r, "\mathrm{div}"] \arrow[d, "\Pi^2"'] 
    & \Lambda^3 \arrow[d, "\Pi^3"'] \\
\Lambda^0_h \arrow[r, "\mathrm{grad}"'] 
    & \Lambda^1_h \arrow[r, "\mathrm{curl}"'] 
    & \Lambda^2_h \arrow[r, "\mathrm{div}"'] 
    & \Lambda^3_h
\end{tikzcd}
\]
That is, each square commutes: $\Pi^{k+1} \circ \mathrm{d}_k = \mathbf{d}_k \circ \Pi^k$, where $\mathrm{d}$ denotes the relevant exterior derivative ($\mathrm{grad}$, $\mathrm{curl}$, or $\mathrm{div}$) and $\mathbf{d}$ denotes its discrete version. This commuting-diagram property guarantees that key identities of vector calculus, such as $\mathrm{curl}\,\mathrm{grad} = 0$ and $\mathrm{div}\,\mathrm{curl} = 0$, hold exactly at the discrete level.  In the present case, the primary role of this compatibility is to ensure that the space $\cV = \Lambda_h^2$ has a discrete Hodge decomposition, and the pressure projection is guaranteed to produce a discrete vector $\bu\in V_{\div}$ that is pointwise divergence-free.

More precisely, denote by $\cI_{d-1}\colon \mathbb{R}^V \to \cV$ the interpolation operator defined by $\cI_{d-1}\bu = *\vec u^\flat$, which maps discrete $(d-1)$-forms to their continuous counterparts in terms of $*$, the continuous Hodge star operator.
It follows that the coefficient space $V_{\div}$ has the alternative expression
\[
    V_{\div} = \{\bu \in \mathbb{R}^V \mid \div(\cI_{d-1}\bu) = 0\}.
\]
Often, mixed finite element discretizations only enforce the divergence-free constraint on $\bu$ weakly, but not strongly (i.e. pointwise). The FEEC machinery guarantees divergence-freeness holds strongly; it can be restated in the language of \emph{mimetic interpolations}. Similarly, there exists a commuting property $\cI_{k+1}\mathbf{d}_k = \mathrm{d}_k \cI_k$ between the interpolation operators and the discrete/continuous exterior derivatives for each degree $k$. Letting $\mathbf{div} = \mathbf{d}_{d-1}$ denote the discrete divergence operator acting on the graph induced by the connectivity of the mesh or grid defining $M$ yields the fully discrete description of $V_{\div}$ given by
\[
    V_{\div} = \{\bu \in \mathbb{R}^V \mid \mathbf{div}(\bu) = 0\}.
\]
In practice, this constraint is efficiently realized through a Hodge decomposition which reconstructs the divergence-free part of the velocities. This is detailed in \secref{sec:discreteHodgeDecomposition}. Before that, we next detail the setup of our two instantiations of the FEEC machinery that allow us to concretely choose the space chosen computationally to represent spaces $\cF$ and $\cV$.

\begin{remark}
Numerical methods for ideal fluid simulation that do not employ compatible discretizations, such as those not based on FEEC, do not satisfy the above construction. In particular, the matrix
$\bC_{\bu}=\sum_{k=1}^{V}u_k\bC_k$,
which governs the advection of functions, is generally not skew-symmetric. Indeed, $\bC_{\bu}$ is skew-symmetric only when $\bu\in V_{\div}$, whereas the discrete incompressibility condition $\mathbf{div}(\bu)=0$ provided by the divergence operation in a non-compatible discretization does not generally imply $\bu\in V_{\div}$.
One example is Azencot et al.\ \cite{Azencot:2014:FFS}. Following the construction presented there yields infinitesimal operators that do not advect functions orthogonally with respect to the pulled-back $L^2$ inner product, i.e., the inner product induced by the mass matrix $\bM$.
It is worth emphasizing, however, that the advection operators $\bA_{\bu}$ acting on half-densities remain skew-symmetric regardless of the choice of finite element or interpolation basis. Consequently, the Lax equations in our formulation can, in principle, still be simulated using non-compatible bases. The drawback is computational rather than theoretical: doing so requires repeated evaluation of the matrix square root $\bM^{1/2}$, which is generally prohibitively expensive. We leave the development of efficient algorithms for this setting as future work.
\end{remark}

\subsubsection{B-spline and Whitney implementation details.}
We consider two instances of the FEEC framework: B-splines on tensor-product grids and Whitney forms on simplicial (triangle) meshes. We begin with the B-spline construction, following the isogeometric analysis (IGA) literature \cite{Hughes:2005:IGA,Cottrell:2009:IGA}, which employs the spline bases of computer-aided geometric design as the trial and test functions in a Galerkin discretization. This provides $C^{p-1}$ continuity across elements, rather than the $C^0$ continuity of standard nodal finite elements, leading to improved accuracy per degree of freedom. This construction is compatible with the FEEC framework \cite{Buffa:2010:IGA,Buffa:2011:IGD,Buffa:2016:IGA}; throughout this manuscript, we refer to it as \emph{IGA--FEEC}.
Following \cite{Nabizadeh:2024:FIP}, we adopt this discretization for both the representation space $\cF$ carrying $(\blambda,\bmu)$ and the velocity space $\cV$ carrying $\bu$. Writing $S_p$ for the univariate B-spline space of degree $p$ on a given knot vector and $S_{p-1} = \tfrac{d}{dx}S_p$, the two-dimensional spaces are
\begin{align*}
    \cF &= S_p \otimes S_p \\
    \cV &= (S_p \otimes S_{p-1}) \times (S_{p-1} \otimes S_p),
\end{align*}
\ie, $\cF$ is the $0$-form space $\Lambda^0_h$ and $\cV$ the flux space $\Lambda^{n-1}_h$, each component carrying degree $p$ in its own coordinate and $p-1$ in the others; the three-dimensional spaces follow the same pattern. Note that reducing the degree in exactly the differentiated direction makes $\mathbf{div}\circ\mathbf{curl} = 0$ hold exactly. We refer the reader to \cite{Nabizadeh:2024:FIP} for the remaining operators and further details. We additionally note that in this implementation, we choose to use the Cayley transform for advection, $\widetilde{\exp} = \mathrm{Cay}$ (see \secref{sec:time-integration}). 

Whitney forms provide an analogous construction on triangle meshes, giving a first-order FEEC discretization closely tied to discrete exterior calculus (DEC) \cite{Hirani:2003:DEC,Crane:2013:DEC}. Suppose the domain $M$ is a 2-dimensional surfaces, discretized by a triangle mesh $\sfM = (\sfP, \sfE, \sfF)$,
\begin{align*}
    \cF &= \spanset \{ \phi_\sfp   : \sfp \in \sfP \} \\ 
    \cV &= \spanset \{ \vec\psi_{\sfe} = (\phi_\sfp * d \phi_\sfq - \phi_\sfq * d \phi_\sfp)^{\sharp}: (\sfp, \sfq) = \sfe \in \sfE \}, 
\end{align*}
contain Whitney 0-forms and (the $L^2$-sharps of) Whitney flux $1$-forms, respectively, and $*$ the Hodge star operator on 1-forms (or equivalently, the 90-degree rotation operator on the corresponding vector fields). The components of $\bC_k$ can be readily computed in this basis by evaluating the functions $\phi, \psi$ on each triangle, which can be used to construct $\bA_k$ for $\textsc{Advect}$ and $\textsc{Reconstruct}$. For \textsc{Advect}, we opt to use a direct computation of matrix exponential as implemented in SciPy \citep{AlMohy:2011:Exp, Scipy}. The operator $\textsc{Project}$ is implemented using the streamfunction solve detailed previously.

Only in the Whitney form implementation, we also row-sum the mass matrix $\bM$, matching the standard diagonal barycentric Hodge star $\star_0$ from DEC. Using the diagonal operator does not change the conservation properties of the system (so long as Casimirs are also measured by the row-summed $\bM$), and in practice makes little difference to the simulation outputs, so we choose to adopt it for efficiency. This row-sum is not performed in the IGA--FEEC case.

Finally, these implementations store only function space coefficients $\bar{\blambda} := \bM^{-\tfrac12} \blambda$ and $\bar{\bmu} := \bM^{-\tfrac12} \bmu$ rather than half-density coefficients written in the algorithms above. This 
avoids costly computations involving the dense matrices $\bM^{\pm\tfrac12}$. Indeed, the routine \textsc{Reconstruct} requires the contraction $$\blambda^\intercal \bA_k \bmu_k = \big(\bM^{\frac12} \bar{\blambda}\big)^\intercal \Big(\bM^{-\tfrac12} \bC_k \bM^{-\tfrac12}\Big) \bM^{\tfrac12} \bar{\bmu}_k = \bar{\blambda}_k^\intercal \bC_k \bar{\bmu}_k,$$ which is algebraically equivalent. Similarly, \textsc{Advect} on function coefficients simplifies to  $\bM^{-\tfrac12} \bA_\bu \bM^{\tfrac12} \bar{\blambda} = \bM^{-1} \bC_\bu \blambda$. Thus, factorizing mass matrices is never required. 

\subsubsection{Discrete Hodge decomposition for $\mathbb{R}^V$}
\label{sec:discreteHodgeDecomposition}
To satisfy the discrete constraint $\mathbf{div}(\bu)=0$ in the space $V_{\div}$, we leverage the compatibility of the FEEC machinery together with a streamform--vorticity solver (see \algref{alg:project-stream}). Owing to the FEEC construction, the discrete identity $\mathbf{div}\circ\mathbf{curl}=0$ 
holds pointwise, so that $\mathbf{curl}\,\bpsi$ corresponds to a pointwise divergence-free 
velocity field. The problem of pressure projection therefore reduces to a single symmetric positive definite Poisson solve for $\bpsi$, which is amenable to standard CG or multigrid methods as in \cite{Nabizadeh:2024:FIP}.  This approach was seen empirically to be more efficient than alternatives, including a mixed finite element formulation requiring an iterative saddle-point solve \cite{Elman:1994:IPU}, as well as a pressure Poisson formulation with a mass-matrix inverse in the Schur complement that required an inner CG solve nested within an outer iteration.  

To describe this procedure in detail, note that $\mathbf{curl}\,\bpsi$ is already the Leray projection \cite{Leray:1934:SML} of $\hat\bu$ on a simply connected domain; otherwise the discrete Hodge decomposition carries a harmonic summand. The velocity is discretized as a flux $(d-1)$-form so the relevant harmonic space is
\begin{equation}
\cH = \ker \Delta_{d-1}, \qquad \dim\cH = b_{d-1}(M),
\end{equation}
where $\Delta_{d-1} = \mathbf{d}\,\bdelta + \bdelta\,\mathbf{d}$ is the Hodge Laplacian and $b_{d-1}$ is the $(d-1)$-st Betti number of the domain. Since the topology and mesh are fixed, we compute a $\bK$-orthonormal basis $\{\bh_1,\dots,\bh_{b}\}$ of $\cH$ once at initialization as the zero-frequency eigenspace of $\Delta_{d-1}$, and recover the harmonic component by one inner product per basis field,
\begin{equation}
\Call{HarmonicProj}{\hat\bu} = \sum_{i=1}^{b} \left(\bh_i^\top \hat\bu\right) \bh_i .
\end{equation}
Since $\cH$ is $\bK$-orthogonal to the image of $\mathbf{curl}$, this component is untouched by the Poisson solve. On the periodic box, the harmonic fields are the constant vector fields $b = d$ and this step reduces to restoring the mean flow.

\subsection{Resetting}\label{subsec:resetting}

\begin{figure}[tbp]
    \begin{center}
    \includegraphics[width=\columnwidth]{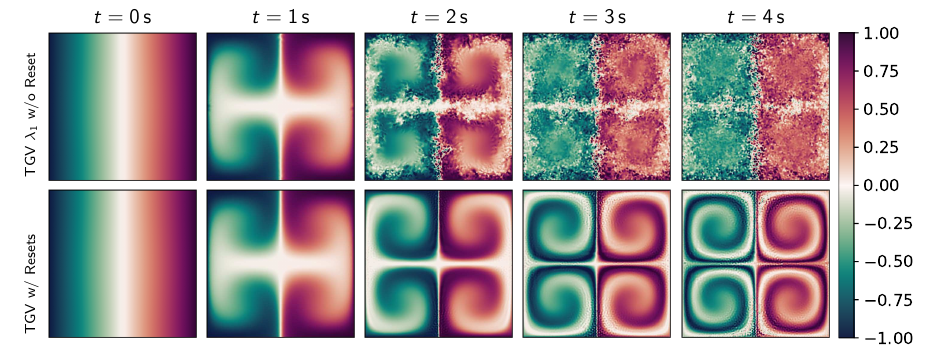}
    \end{center}
    \caption{Evolution of the $x$-component of $\blambda$ over time for the Taylor-Green Vortex initial condition. \textbf{Above:} Without resetting, $\blambda$ develops dispersive artifacts emerging from the inability to resolve the shearing present in the flow map. \textbf{Lower:} With resetting every 120 frames (0.3 s), 
    the dye advects with far fewer dispersive artifacts. 
    }
    \label{fig:tgv_reset_cmp}
\end{figure}

As discussed in \Cref{sec:LowRankParameterization}, the low-rank parameterization of $\bZ$ in terms of 
discrete Clebsch variables $(\blambda,\bmu)$ represents the inverse flow map $\phi^{-1}$ of the fluid together with the components of the velocity field $\vec u$, carried from a static earlier frame to the current frame through $\vec u_0(\phi^{-1})$. Consequently, the advection of these fields can cause large amounts of distortion,
including the stretching computed in the finite-time Lyapunov exponent (FTLE) fields (see \secref{sec:FTLE}) and visible in \Cref{fig:tgv_reset_cmp}.
Indeed, the vakonomic algorithm in \secref{sec:algorithm} is
built on a discrete analogue of Kelvin's circulation theorem and can therefore only rotate the finite-element coefficients (\ie, labels)   
associated to $\blambda$ and $\bmu$. While the advantage of this is exact preservation of the
Casimirs associated to powers of $\bZ$, its lack of mechanistic dissipation 
also produces visually dispersive artifacts that manifest as a blow-up in enstrophy over time.

To alleviate this and enable less dispersive simulations, it is occasionally useful to relax strict Casimir preservation 
and perform a so-called \emph{reset} of these labels.  This ensures that excess dispersion does not accumulate in the auxilliary symplectic space $P^m$ containing the discrete Clebsch variables.

To describe this more precisely, recall that the discrete Clebsch data consist of $m$ (or $m-1$) pairs of scalar labels $\{\blambda_a,\bmu_a\}_{a=1}^m$, where $m$ is the dimension of the ambient space.   
The label $\blambda_a$ holds the $a$-th coordinate of the inverse flow map $\phi^{-1}$, while $\mu_a$ holds the $a$-th Cartesian component of the pulled-back velocity $\vec u_0\circ\phi^{-1}$, stored as a scalar field. It is convenient to overload notation as before, so that 
the undecorated $\blambda$ and $\bmu$ without 
indices denote the full collection of these variables over $a=1,\dots,m$. 
Additionally, $\cI_\cV$ denotes the interpolation operator $\bu\mapsto\vec u$ sending the coefficients $\bu\in\mathbb{R}^{V}$ of the velocity field to its finite-element representation $\vec u \in \cV$, and $\cI_\cF:\mathbb{R}^F\to\cF$ denotes the corresponding interpolation operator on scalar functions. In particular, $\cI_\cV^+$ and $\cI_\cF^+$ denote the corresponding Moore--Penrose pseudoinverses (\ie, least-squares fits back to coefficients).

With this, resetting the labels $\blambda$ 
encoding the inverse flow map is straightforward: the components $\blambda_a$ are simply set to the identity,
\ie, to the current position of each quadrature point of the finite element space $\cF$.  Conversely, there are two reasonable ways to reset the labels $\bmu$ associated with the velocity field.
On one hand, resetting $\bmu$ (for fixed $\blambda$) can be 
accomplished via the adjoint $\bar{\bA}^*J$ of the momentum map $J$ 
computed in
\eqref{eq:cheap-momentum-map}. More precisely, consider the momentum vector
$$\hat\bu^\top \bv = \sum_{a=1}^m\blambda_a^\intercal\bA_\bv\bmu_a,
\quad\text{equivalently}\quad
\bu = \,\cI_\cV^+\,(\nabla\cI_\cF\blambda)^\top(\cI_\cF\bmu).$$

It follows that $\nabla\cI_\cF\blambda=\mathrm{id}$ since $\blambda$ 
was just 
re-initialized to the position field,
and therefore the least-squares problem that recovers the velocity in the $\cF$ basis
reduces to
$$\operatorname*{argmin}_{\bmu}\ \tfrac12\big\lVert \cI_\cV^+\cI_\cF\bmu - \bu\big\rVert_\bK^2 .$$
Writing $\bV=\cI_\cV^+\cI_\cF$, the corresponding normal equations give
$$\bmu = \left(\bV^\intercal\bK\bV\right)^{-1}\bV^\intercal\bK\bu,$$
although this requires inverting the matrix $\bV^\intercal\bK\bV$.
Alternatively, one could relax the requirement that $\bmu$ project orthogonally onto the finite
element basis spanned by $\cI_\cV$, instead seeking the
$\bmu$ whose interpolation is the (component-wise) orthogonal projection of $\bu$ onto the $\cI_\cF$ space.
This corresponds to a different least-squares problem
$$\operatorname*{argmin}_{\bmu}\ \tfrac12\big\lVert \bmu - \cI_\cF^+\cI_\cV\bu\big\rVert_\bM^2 ,$$
whose solution is given by 
\begin{equation}\label{eq:reset}
    \bmu = \cI_\cF^+\cI_\cV\bu = \bF\bu, 
\end{equation}
for $\bF=\cI_\cF^+\cI_\cV$.
This second formulation is considerably more computationally tractable, since it avoids assembling and
inverting $\bV^\top\bK\bV$. 
The trade-off is that the discrete velocity
$\bu$ and its representation $\bmu$ are no longer related through orthogonal projection onto the velocity subspace $\cI_\cV$.  Despite this mild drawback, we observed little difference between these two procedures in practice, and therefore all experiments in \Cref{sec:experiments} use the cheaper reset option \eqref{eq:reset}.

\vspace{1pc}

\begin{algorithm}
    \caption{Vakonomic Fluids loop with Resetting}
    \begin{algorithmic}[1]
        \State $(\blambda,\bmu)^{(0)} \gets \textsc{Reset}(\bu^{(0)})$ 
        \For{$n \in 1, 2, \dots$}
            \State $(\blambda, \bmu)^{(n)} \gets \textsc{Step}\left((\blambda, \bmu)^{(n-1)}\right)$
            \If{$n \mod p = 0$}
                \State $(\blambda, \bmu)^{(n)} \gets \textsc{Reset}(\bu^{(n)})$
            \EndIf
        \EndFor
    \end{algorithmic}
    \label{alg:sim_loop_w_reset}
\end{algorithm}

\vspace{1pc}

\begin{algorithm}
\caption{\textsc{Reset}$(\bu)$,  re-anchor flow map to identity}
\begin{algorithmic}[1]
\Procedure{Reset}{$\bu$}
  \State $\blambda_a \gets \bM^{\frac12}\cI_\cF^{+}\,x_a,\qquad a=1,\dots,d$
  \State $\bmu_a \gets \bM^{\frac12}\cI_\cF^{+}\cI_\cV\bu_a,\qquad a=1,\dots,d$
  \State \Return $(\blambda,\bmu)$ %
\EndProcedure
\end{algorithmic}
\label{alg:reset}
\end{algorithm}

\subsubsection{Automatic Resetting by Finite-Time Lyapunov Exponents} 
\label{sec:FTLE}

It is worth mentioning an automated heuristic for selecting the reset period, thereby eliminating this hyperparameter from the vakonomic fluid simulation pipeline.  Recall that the cumulative finite-time Lyapunov exponent (FTLE) is given by 
$$ S_\sff = \ln \parens*{ \sigma_{\max} (F_\sff) } = - \ln \parens*{ \sigma_{\min} (\Psi_\sff) },  $$
where $\Psi_\sff = [\grad \blambda]_\sff$ is the piecewise-constant gradient of the inverse flow map $\lambda$ over each face, $F_\sff = \Psi_\sff^{\dagger}$ is the deformation gradient (that uses a pseudo-inverse, to accomodate embeddings of 2D surfaces in ambient 3D space), and
$\sigma_{\min}$ computes the minimal nonzero singular value.  The quantity $S_{\sff}/T$ where $T$ is the time since the last reset is the FTLE, and $S_{\sff}$ measures the maximal stretching that has occurred since $\blambda$ was last reset. 
Following Nabizadeh et al.\ \cite{Nabizadeh:2024:FIP}, we find that selecting the reset frequency so that $S_{\sff}\leq 1$ consistently leads to good results.

\subsubsection{Impact of resetting}
We find that occasional resetting within the vakonomic fluid \algref{alg:sim_loop_w_reset} is a necessary ingredient for producing accurate long-time
simulations with minimal dispersion.  However, it should be noted that this operation is not Hamiltonian structure-preserving, meaning that Casimirs will change discontinuously upon reset.  More precisely, recall that the Casimirs are given by traces of powers of the dual matrix $\bZ$, as in \eqref{eq:blkdiagBwithLambdaAndMu}. Because the advection of the labels is orthogonal, the blocks associated with $|\blambda_i|^2$ are unchanged by advection, but the blocks associated with $\bmu$ are effectively perturbed. Since resetting changes these labels discontinuously, it follows that the Casimirs exhibit a jump as well.  
While the impact of resetting on
energy and Casimir preservation is examined numerically later in \secref{sec:verification}, we observe that \algref{alg:sim_loop_w_reset} is quite stable with the inclusion of resetting and capable of producing long-term simulations with minimal dispersion (see \Cref{fig:laplacian_eigenfunction_visual_comparison}).

\section{Numerical Experiments}

\label{sec:experiments}

In this section we assess our method in two stages. We first verify it through convergence tests against the Taylor--Green vortex and examine the behavior of the solver in preserving invariants of the flow (\secref{sec:verification}). We then validate it against established benchmarks: mollified point vortex systems (\secref{sec:MollifiedPointVortexSystems}), flows on surfaces (\secref{sec:FlowsOnSurfaces}), and a trefoil knot in three dimensions (\secref{sec:ThreeDimensionalTrefoilKnotExp}). Finally, we demonstrate the extension of the solver to stratified flow under the Boussinesq approximation (\secref{sec:Buoynacy}).
B-spline experiments on tensor grids were run on a workstation with an Intel Core Ultra~9 285K processor (24 cores) and 64\,GB of RAM, while the triangle-mesh experiments with the Whitney basis were run on a MacBook Pro with a 12-core Apple M4 Pro and 24 GB of RAM.  See \tabref{tab:experimentdetails} for a full breakdown of experiment configurations and timings. Our C++ implementation of the B-spline discretization builds on the codebase released with \cite{Nabizadeh:2024:FIP}, while the triangle-mesh implementation was developed from scratch. For baselines, we implement Functional Fluids \citep{Azencot:2014:FFS} as a fluid solver on surfaces, as well as a modified version that uses the IGA--FEEC basis and to preserve energy exactly via the same trapezoidal time integrator we use for our method (which we refer to as \emph{Vorticity FEEC} in the figures that follow). These baselines share a similar Koopman operator discretization as our method; the Koopman operators are used to update the vorticity directly, and have function spaces similar to the implementations of our method, making a fair comparison. In \Cref{fig:shieldedTaylorVortices} we additionally compare against Zeitlin's method, commonly known as matrix hydrodynamics after the original work \cite{Zeitlin:1991:FMA}, which we implement in MATLAB.

\begin{table}
    \centering 
    \scriptsize
    \caption{Details regarding each experiment.}
    \setlength{\tabcolsep}{3pt} %
    \begin{tabular}{lccc}\toprule
         \textbf{Experiment} & \textbf{Mesh Size} & \textbf{$\Delta t$ ($\mathrm{s}$)} & \textbf{Reset Period ($\mathrm{s}$)} \\
         \midrule
         \multicolumn{4}{l}{\textbf{IGA BSplines on Cartesian Grids}} \\ 
         \quad Taylor-Green Vortex (\figref{fig:convergence_studies}) & $128\times128$ & $1/48$ & $2.29$ \\
         \quad Eigenfunction Vortex (\figref{fig:laplacian_eigenfunction_visual_comparison}) & $128\times128$ & $1/4$ & $25.00$ \\
         \quad Shielded Taylor Vortex (\figref{fig:shieldedTaylorVortices}) & $128\times128$ & $1/48$ & $0.21$ \\
         \quad Leapfrogging Vortex (\figref{fig:leapfrogVortices}) & $128\times128$ & $1/48$ & $0.21$\\
         \quad Rayleigh-Taylor Instability (\figref{fig:buoyancy_demo}) & $256\times512$ & $1/96$ & $0.05$\\
         \quad Trefoil Knot (\figref{fig:trefoilknot}) & $64\times64\times64$ & $1/72$ & $0.21$ \\
         \midrule
         \multicolumn{4}{l}{\textbf{Whitney Forms on Triangle Meshes}} \\
         \quad Taylor-Green Vortex (\figref{fig:tgv_reset_cmp}) & 11.6k & $2.5\times10^{-3}$ & 0.3  \\
         \quad Shielded Taylor Vortex (\figref{fig:shielded-vorts}) & 89.3k & $1.5 \times 10^{-3}$ & 0.36 \\
         \quad Vortices on Oblate Spheroid (\figref{fig:six-point-vorts}) & 41.0k & $2.5 \times 10^{-2}$ & FTLE--Automated \\
         \quad Double Shear Layer (\figref{fig:shear-layer}) & 175k & $5 \times 10^{-3}$ & 0.18 \\
         \quad Sphere 1 Flow (\figref{fig:sphere-flow}) & 41k &   $2.5 \times 10^{-3}$ & 0.2 \\
         \quad Sphere 2 Flow (\figref{fig:sphere-flow}) & 164k & $1 \times 10^{-3}$ & 0.2 \\
         \quad Helicoid Flow (\figref{fig:helicoid-flow}) & 56.8k & $1.25 \times 10^{-3}$ &  0.2 \\
         \quad Bunny Flow (\figref{fig:bunny-flow}) & 85.7k & $3.33 \times 10^{-4}$ & 0.12  \\
          \bottomrule
    \end{tabular}
    \label{tab:experimentdetails}
\end{table}

\subsection{Verification}
\label{sec:verification}

The first set of numerical experiments deals with verification.  Precisely, we verify that \algref{alg:sim_loop_no_reset} preserves energy and Casimirs.

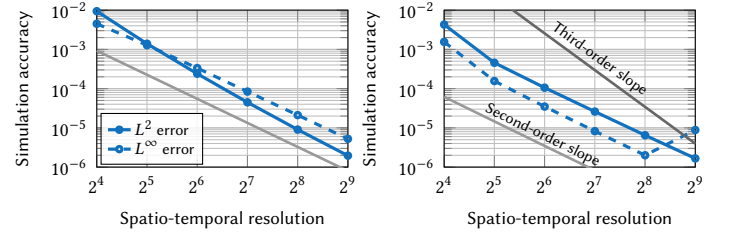
\begin{figure}
    \centering
    \begin{minipage}{0.47\columnwidth}
        \centering
        \definecolor{mycolor1}{rgb}{0.06600,0.44300,0.74500}%
\definecolor{mycolor2}{rgb}{0.12941,0.12941,0.12941}%
\begin{tikzpicture}

\begin{loglogaxis}[%
width=0.8\linewidth,
height=0.5\linewidth,
at={(0,0)},
scale only axis,
xmin=16,
xmax=512,
log basis x=2,
log basis y=10,
ymin=1e-06,
ymax=1e-2,
xminorticks=true,
yminorticks=true,
ymajorticks=true,
xminorgrids=true,
xmajorgrids=true,
yminorgrids=true,
ymajorgrids=true,
max space between ticks=1,
xlabel={\scriptsize \sffamily Spatio-temporal resolution},
ylabel={\scriptsize \sffamily Simulation accuracy},
xticklabel style={font=\scriptsize},
yticklabel style={font=\scriptsize},
legend pos = south west,
legend style={
  at={(0.015,0.015)}, anchor=south west,
  font=\tiny\sffamily,
  draw=black, fill=white, fill opacity=0.85, text opacity=1,
  inner xsep=2pt, inner ysep=1.5pt,
  nodes={inner sep=1pt},
  row sep=-1.5pt,
  /tikz/every even column/.append style={column sep=2pt},
},
legend cell align={left},
legend image code/.code={\draw[mark repeat=2,mark phase=2,#1]
    plot coordinates {(0cm,0cm)(0.15cm,0cm)(0.3cm,0cm)};},
]
\addplot [color=mycolor1, line width=1.2pt, mark size=1.0pt, mark=o, mark options={solid, mycolor1}]
  table[row sep=crcr]{%
16	0.00939104\\
32	0.00139492\\
64	0.000239827\\
128	4.47832e-05\\
256	9.04317e-06\\
512	1.95875e-06\\
};
\addlegendentry{\tiny \sffamily $L^2$ error}

\addplot [color=mycolor1, dashed, line width=1.2pt, mark size=1.0pt, mark=o, mark options={solid, mycolor1}]
  table[row sep=crcr]{%
16	0.00449973\\
32	0.00127648\\
64	0.000334061\\
128	8.40473e-05\\
256	2.11404e-05\\
512	5.24973e-06\\
};
\addlegendentry{\tiny \sffamily $L^\infty$ error}

\addplot [color=white!60!black, line width=1.0pt]
  table[row sep=crcr]{%
512	7.96247496736576e-07\\
16	0.000924069447019463\\
};

\end{loglogaxis}

\end{tikzpicture}%
    \end{minipage}\hfill
    \begin{minipage}{0.47\columnwidth}
        \centering
        \definecolor{mycolor1}{rgb}{0.06600,0.44300,0.74500}%
\definecolor{mycolor2}{rgb}{0.12941,0.12941,0.12941}%
\begin{tikzpicture}

\begin{loglogaxis}[%
width=0.8\linewidth,
height=0.5\linewidth,
at={(0,0)},
scale only axis,
xmin=16,
xmax=512,
log basis x=2,
log basis y=10,
ymin=1e-06,
ymax=1e-2,
xminorticks=true,
yminorticks=true,
ymajorticks=true,
xminorgrids=true,
xmajorgrids=true,
yminorgrids=true,
ymajorgrids=true,
max space between ticks=1,
xlabel={\scriptsize \sffamily Spatio-temporal resolution},
ylabel={\scriptsize \sffamily Simulation accuracy},
xticklabel style={font=\scriptsize},
yticklabel style={font=\scriptsize},
legend pos = south west,
]
\addplot [color=mycolor1, line width=1.2pt, mark size=1.0pt, mark=o, mark options={solid, mycolor1}]
  table[row sep=crcr]{%
16	0.00429524\\
32	0.000455707\\
64	0.000105681\\
128	2.60159e-05\\
256	6.48402e-06\\
512	1.66489e-06\\
};

\addplot [color=mycolor1, dashed, line width=1.2pt, mark size=1.0pt, mark=o, mark options={solid, mycolor1}]
  table[row sep=crcr]{%
16	0.00155474\\
32	0.000156159\\
64	3.49243e-05\\
128	8.24087e-06\\
256	2.00442e-06\\
512	8.82991e-06\\
};

\addplot [color=white!60!black, line width=1.0pt]
  table[row sep=crcr]{%
512	4.64275015329112e-08\\
16	0.0000616109149748735\\
};

\addplot [color=white!40!black, line width=1.0pt]
  table[row sep=crcr]{%
512	3.96152093813773e-06\\
16	0.191507361980362\\
};
\end{loglogaxis}

\end{tikzpicture}%
    \end{minipage}
    \begin{picture}(0,0)(0,0)
    \put(-47,35){\sffamily \tiny \rotatebox{-36.5}{Third-order slope}}
    \put(-72,4){\sffamily \tiny \rotatebox{-26}{Second-order slope}}
    \end{picture}
    \caption{Convergence of the method with FEEC--IGA B-spline discretization spaces, with quadratic 
    ($p=2$, left) and cubic ($p=3$, right) spatial accuracy. Observe that temporal discretization error dominates over longer time horizons, yielding second-order convergence in both cases.
}
    \label{fig:convergence_studies}
\end{figure}

\begin{figure}
    \centering
    \includegraphics[width=\linewidth]{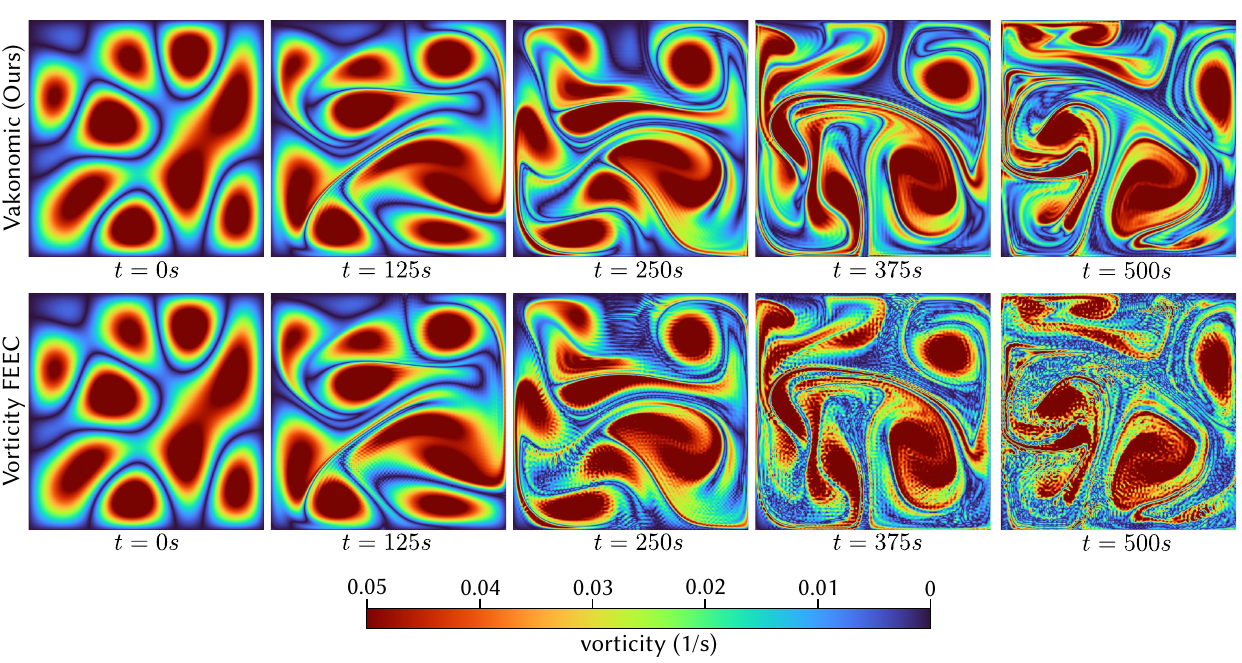}
    \caption{Vorticity initialized as a random superposition of the first 20 eigenfunctions of the scalar Laplacian, with i.i.d.\ Gaussian mode amplitudes. As the flow evolves, the finite-element functional-fluids-on-surfaces baseline (top) accumulates heavy dispersive artifacts, whereas our method (bottom) exhibits substantially less numerical dispersion. We attribute this to structure preservation: our scheme enforces a discrete Kelvin circulation theorem, while the baseline preserves only a weak notion of enstrophy.}
    \label{fig:laplacian_eigenfunction_visual_comparison}
\end{figure}

\begin{figure}
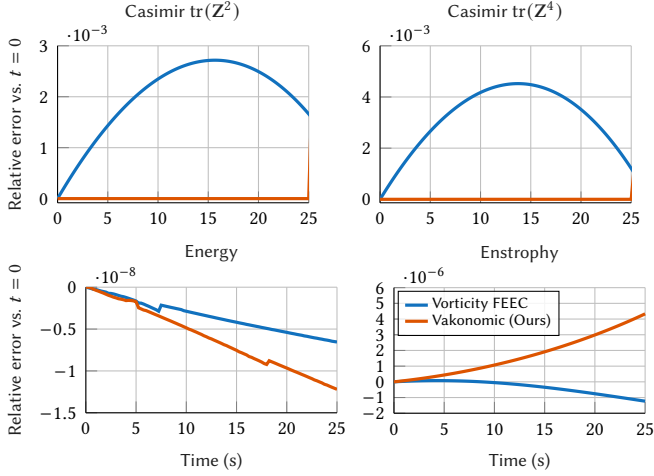

    \centering
    \begin{minipage}{0.47\columnwidth}
        \input{figs/bsplines_basis/plots/casimir2_comparison}
    \end{minipage}\hfill
    \begin{minipage}{0.47\columnwidth}
        \input{figs/bsplines_basis/plots/casimir4_comparison}
    \end{minipage}
    \begin{minipage}{0.47\columnwidth}
        \definecolor{mycolor1}{rgb}{0.06600,0.44300,0.74500}%
\definecolor{mycolor2}{rgb}{0.86600,0.32900,0.00000}%
\definecolor{mycolor3}{rgb}{0.12941,0.12941,0.12941}%
\begin{tikzpicture}

\begin{axis}[%
width=0.8\linewidth,
height=0.4\linewidth,
at={(0,0)},
scale only axis,
xmin=0,
xmax=25,
xtick={0,5,10,15,20,25},
xlabel style={font=\color{mycolor3}},
xlabel={\scriptsize \sffamily Time ($\mathrm{s}$)},
ymin=-1.5e-08,
ymax=0,
ytick={-1.5e-8,-1e-8,-5e-9,0},
ylabel style={font=\color{mycolor3}},
ylabel={\scriptsize \sffamily Relative error vs. $t=0$},
axis background/.style={fill=white},
title style={font=\scriptsize \sffamily \color{mycolor3}},
title={Energy},
axis x line*=bottom,
axis y line*=left,
xmajorgrids,
ymajorgrids,
legend style={legend cell align=left, align=left},
xticklabel style={font=\scriptsize},
yticklabel style={font=\scriptsize}
]
\addplot [color=mycolor1, line width=1.2pt]
  table[row sep=crcr]{%
0	0\\
0.25	-6.60548146850048e-11\\
0.5	-1.45378951481701e-10\\
0.75	-2.44995777071618e-10\\
1	-2.7244363590435e-10\\
1.25	-3.60503516419587e-10\\
1.5	-4.51906783386695e-10\\
1.75	-5.18611929553634e-10\\
2	-5.8537497777379e-10\\
2.25	-7.30960395787361e-10\\
2.5	-7.44284420101962e-10\\
2.75	-8.02055888952648e-10\\
3	-8.89271313733847e-10\\
3.25	-9.86188532279383e-10\\
3.5	-1.07787841440992e-09\\
3.75	-1.13643110385758e-09\\
4	-1.22565832024046e-09\\
4.25	-1.32552701976187e-09\\
4.5	-1.43149890044012e-09\\
4.75	-1.53324819770053e-09\\
5	-1.65259514792451e-09\\
5.25	-1.89280852014796e-09\\
5.5	-1.90293193228402e-09\\
5.75	-2.01113867984266e-09\\
6	-2.14816861237179e-09\\
6.25	-2.27533218740638e-09\\
6.5	-2.41577011288886e-09\\
6.75	-2.55397164775271e-09\\
7	-2.69178782921501e-09\\
7.25	-2.87003655731802e-09\\
7.5	-2.14753915610378e-09\\
7.75	-2.21448611952981e-09\\
8	-2.2825426691441e-09\\
8.25	-2.35004991112193e-09\\
8.5	-2.41788963278429e-09\\
8.75	-2.48054483507628e-09\\
9	-2.54858980427992e-09\\
9.25	-2.61030150785209e-09\\
9.5	-2.67664034734552e-09\\
9.75	-2.8051091565587e-09\\
10	-2.87166817622818e-09\\
10.25	-2.94110111859698e-09\\
10.5	-3.00842337875217e-09\\
10.75	-3.07575386709386e-09\\
11	-3.14200219653567e-09\\
11.25	-3.21085200427904e-09\\
11.5	-3.27758610147259e-09\\
11.75	-3.34383382141911e-09\\
12	-3.40989427393561e-09\\
12.25	-3.47855818561339e-09\\
12.5	-3.54214012550451e-09\\
12.75	-3.60872047750938e-09\\
13	-3.67720367383158e-09\\
13.25	-3.74299000583824e-09\\
13.5	-3.80624510387633e-09\\
13.75	-3.87060948335502e-09\\
14	-3.93789243106354e-09\\
14.25	-4.00115392880225e-09\\
14.5	-4.06463591146465e-09\\
14.75	-4.12548578868683e-09\\
15	-4.18597834929111e-09\\
15.25	-4.25337999620874e-09\\
15.5	-4.3115313330033e-09\\
15.75	-4.37277540130883e-09\\
16	-4.44189018713258e-09\\
16.25	-4.49725445430764e-09\\
16.5	-4.55972147492666e-09\\
16.75	-4.61920440657142e-09\\
17	-4.68005900738214e-09\\
17.25	-4.7371777067697e-09\\
17.5	-4.80468510112136e-09\\
17.75	-4.86116093534438e-09\\
18	-4.92126859966859e-09\\
18.25	-4.97866675264985e-09\\
18.5	-5.04830677104582e-09\\
18.75	-5.09700209305457e-09\\
19	-5.1562593590656e-09\\
19.25	-5.21668807503752e-09\\
19.5	-5.27902557790593e-09\\
19.75	-5.33414757070041e-09\\
20	-5.39282779703898e-09\\
20.25	-5.45750195250238e-09\\
20.5	-5.51658292199874e-09\\
20.75	-5.5723573795087e-09\\
21	-5.6361383196143e-09\\
21.25	-5.69033250819288e-09\\
21.5	-5.74266301252696e-09\\
21.75	-5.81789758825439e-09\\
22	-5.8619293569979e-09\\
22.25	-5.92655185773488e-09\\
22.5	-5.9879051780724e-09\\
22.75	-6.0475787293713e-09\\
23	-6.09410089589848e-09\\
23.25	-6.15753183332976e-09\\
23.5	-6.20868585936675e-09\\
23.75	-6.27433368356716e-09\\
24	-6.32000164243555e-09\\
24.25	-6.37743926023729e-09\\
24.5	-6.45307716947753e-09\\
24.75	-6.49727715892302e-09\\
25	-6.54429621649111e-09\\
};

\addplot [color=mycolor2, line width=1.2pt]
  table[row sep=crcr]{%
0	0\\
0.25	-5.83149964737082e-12\\
0.5	-8.52549980585754e-11\\
0.75	-2.3522393245894e-10\\
1	-3.64596493434646e-10\\
1.25	-4.68216197497761e-10\\
1.5	-5.89500906842133e-10\\
1.75	-7.02493126473426e-10\\
2	-8.35099212938847e-10\\
2.25	-9.6511235358369e-10\\
2.5	-1.06914803825973e-09\\
2.75	-1.18662842966896e-09\\
3	-1.27124360992765e-09\\
3.25	-1.37704317430158e-09\\
3.5	-1.42913994629436e-09\\
3.75	-1.49923862215272e-09\\
4	-1.59722536969617e-09\\
4.25	-1.57263116354574e-09\\
4.5	-1.58922208555208e-09\\
4.75	-1.63364348192788e-09\\
5	-1.73801423670203e-09\\
5.25	-2.46049363170115e-09\\
5.5	-2.58950918074423e-09\\
5.75	-2.72449428014358e-09\\
6	-2.79495956748717e-09\\
6.25	-2.92032652467071e-09\\
6.5	-3.04861174584647e-09\\
6.75	-3.17668623398645e-09\\
7	-3.30135912827919e-09\\
7.25	-3.4254277078793e-09\\
7.5	-3.55765849754999e-09\\
7.75	-3.6857378616532e-09\\
8	-3.81317542709646e-09\\
8.25	-3.93862359954253e-09\\
8.5	-4.06417279585173e-09\\
8.75	-4.1973278853019e-09\\
9	-4.32412349971137e-09\\
9.25	-4.45972952255309e-09\\
9.5	-4.58382674843717e-09\\
9.75	-4.72081800195733e-09\\
10	-4.85431608494585e-09\\
10.25	-5.00125765567676e-09\\
10.5	-5.10834127449213e-09\\
10.75	-5.2431755237791e-09\\
11	-5.3675686596816e-09\\
11.25	-5.50737654383544e-09\\
11.5	-5.63678186518909e-09\\
11.75	-5.79618914066121e-09\\
12	-5.91748192581968e-09\\
12.25	-6.0661652820401e-09\\
12.5	-6.17268709846724e-09\\
12.75	-6.31276270347709e-09\\
13	-6.44477392242864e-09\\
13.25	-6.57929534819974e-09\\
13.5	-6.7149859385287e-09\\
13.75	-6.86511303698629e-09\\
14	-6.97282535031043e-09\\
14.25	-7.12194342912743e-09\\
14.5	-7.25836983278234e-09\\
14.75	-7.40208534434222e-09\\
15	-7.53702899805413e-09\\
15.25	-7.71183928899107e-09\\
15.5	-7.82484049867956e-09\\
15.75	-7.95970080389503e-09\\
16	-8.08697700543021e-09\\
16.25	-8.28360357484202e-09\\
16.5	-8.37749771312445e-09\\
16.75	-8.52711146407839e-09\\
17	-8.6560294938569e-09\\
17.25	-8.79306493579381e-09\\
17.5	-8.9378527021424e-09\\
17.75	-9.11804636123124e-09\\
18	-9.22737002802853e-09\\
18.25	-8.78128597999854e-09\\
18.5	-8.90612039057322e-09\\
18.75	-9.02894422818542e-09\\
19	-9.15581084880941e-09\\
19.25	-9.28709890146464e-09\\
19.5	-9.4092908447172e-09\\
19.75	-9.53541083347183e-09\\
20	-9.6559496728161e-09\\
20.25	-9.78544794221879e-09\\
20.5	-9.91872203564651e-09\\
20.75	-1.00458626244544e-08\\
21	-1.01634071652548e-08\\
21.25	-1.02889019640993e-08\\
21.5	-1.04172680957899e-08\\
21.75	-1.05400835528403e-08\\
22	-1.06675621068494e-08\\
22.25	-1.07989036427263e-08\\
22.5	-1.09224287048003e-08\\
22.75	-1.10744571241693e-08\\
23	-1.11768083159238e-08\\
23.25	-1.13071917250263e-08\\
23.5	-1.1428132678963e-08\\
23.75	-1.15944496514514e-08\\
24	-1.16858712192555e-08\\
24.25	-1.18238533599279e-08\\
24.5	-1.19527160779986e-08\\
24.75	-1.20741623036249e-08\\
25	-1.21963959973126e-08\\
};

\end{axis}

\end{tikzpicture}%
    \end{minipage}\hfill
    \begin{minipage}{0.47\columnwidth}
        \definecolor{mycolor1}{rgb}{0.06600,0.44300,0.74500}%
\definecolor{mycolor2}{rgb}{0.86600,0.32900,0.00000}%
\definecolor{mycolor3}{rgb}{0.12941,0.12941,0.12941}%
\begin{tikzpicture}

\begin{axis}[%
width=0.8\linewidth,
height=0.4\linewidth,
at={(0,0)},
scale only axis,
xmin=0,
xmax=25,
xtick={0,5,10,15,20,25},
xlabel style={font=\color{mycolor3}},
xlabel={\scriptsize \sffamily Time ($\mathrm{s}$)},
ymin=-2e-6,
ymax=6e-6,
ytick={-2e-6,-1e-6,0,1e-6,2e-6,3e-6,4e-6,5e-6,6e-6},
ylabel style={font=\color{mycolor3}},
axis background/.style={fill=white},
title style={font=\scriptsize \sffamily\color{mycolor3}},
title={Enstrophy},
axis x line*=bottom,
axis y line*=left,
xmajorgrids,
ymajorgrids,
legend style={legend cell align=left, align=left},
legend pos= north west,
xticklabel style={font=\scriptsize},
yticklabel style={font=\scriptsize},
legend style={
  at={(0.015,0.615)}, anchor=south west,
  font=\tiny\sffamily,
  draw=black, fill=white, fill opacity=0.85, text opacity=1,
  inner xsep=2pt, inner ysep=1.5pt,
  nodes={inner sep=1pt},
  row sep=-1.5pt,
  /tikz/every even column/.append style={column sep=2pt},
},
legend cell align={left},
legend image code/.code={\draw[mark repeat=2,mark phase=2,#1]
    plot coordinates {(0cm,0cm)(0.15cm,0cm)(0.3cm,0cm)};},
]
\addplot [color=mycolor1, line width=1.2pt]
  table[row sep=crcr]{%
0	0\\
0.25	8.88547078843282e-09\\
0.5	1.72245681372174e-08\\
0.75	2.50221840897999e-08\\
1	3.2323585790722e-08\\
1.25	3.9058274198396e-08\\
1.5	4.52553463853142e-08\\
1.75	5.09310917537523e-08\\
2	5.60686121850586e-08\\
2.25	6.06256307319962e-08\\
2.5	6.47308698352228e-08\\
2.75	6.82775590226767e-08\\
3	7.1278416549543e-08\\
3.25	7.37529930702102e-08\\
3.5	7.56954376022114e-08\\
3.75	7.71033782507023e-08\\
4	7.79963752561561e-08\\
4.25	7.83586530634901e-08\\
4.5	7.82007502505703e-08\\
4.75	7.75285292039988e-08\\
5	7.63319703166999e-08\\
5.25	7.45851893420387e-08\\
5.5	7.24045542408729e-08\\
5.75	6.96917106276412e-08\\
6	6.64659861976249e-08\\
6.25	6.27376704557663e-08\\
6.5	5.85149077726485e-08\\
6.75	5.38035056628029e-08\\
7	4.86112014187693e-08\\
7.25	4.29069888497122e-08\\
7.5	3.67793842086166e-08\\
7.75	3.0141352457855e-08\\
8	2.30327307208406e-08\\
8.25	1.54579692547386e-08\\
8.5	7.4231991870761e-09\\
8.75	-1.06509995002534e-09\\
9	-1.0006915430212e-08\\
9.25	-1.93905889217631e-08\\
9.5	-2.92147825047592e-08\\
9.75	-3.95089438048916e-08\\
10	-5.01979322975466e-08\\
10.25	-6.13133060345369e-08\\
10.5	-7.28440945683546e-08\\
10.75	-8.478622992828e-08\\
11	-9.71331147394611e-08\\
11.25	-1.09880261061504e-07\\
11.5	-1.23020652045713e-07\\
11.75	-1.36547211178281e-07\\
12	-1.50453839230205e-07\\
12.25	-1.64737747807789e-07\\
12.5	-1.79384580592312e-07\\
12.75	-1.94394224534087e-07\\
13	-2.09761394560195e-07\\
13.25	-2.25474482087391e-07\\
13.5	-2.41527104105508e-07\\
13.75	-2.57914062973333e-07\\
14	-2.74632370146333e-07\\
14.25	-2.91668463793874e-07\\
14.5	-3.09018412794858e-07\\
14.75	-3.26766220835555e-07\\
15	-3.44816907501648e-07\\
15.25	-3.63161768132435e-07\\
15.5	-3.81789794762455e-07\\
15.75	-4.00700065640085e-07\\
16	-4.19885670685781e-07\\
16.25	-4.39327235265501e-07\\
16.5	-4.59028235657154e-07\\
16.75	-4.78981312226881e-07\\
17	-4.99175641458711e-07\\
17.25	-5.19611179058324e-07\\
17.5	-5.40277242326028e-07\\
17.75	-5.61161831797712e-07\\
18	-5.82263293543215e-07\\
18.25	-6.03571051631998e-07\\
18.5	-6.25093158540452e-07\\
18.75	-6.46787899098959e-07\\
19	-6.68682446386706e-07\\
19.25	-6.90756036283211e-07\\
19.5	-7.13004913852999e-07\\
19.75	-7.35417425557394e-07\\
20	-7.57985347031587e-07\\
20.25	-7.80714262831002e-07\\
20.5	-8.0358640713105e-07\\
20.75	-8.26592833635581e-07\\
21	-8.497538202192e-07\\
21.25	-8.73006703968698e-07\\
21.5	-8.96385877111399e-07\\
21.75	-9.19893968548384e-07\\
22	-9.4349012850209e-07\\
22.25	-9.67192364101997e-07\\
22.5	-9.909860835092e-07\\
22.75	-1.01487749437463e-06\\
23	-1.03882275325259e-06\\
23.25	-1.06285444886427e-06\\
23.5	-1.08694567604502e-06\\
23.75	-1.11110749697127e-06\\
24	-1.13530904114485e-06\\
24.25	-1.15956787159509e-06\\
24.5	-1.1838759219623e-06\\
24.75	-1.20820291096198e-06\\
25	-1.23255543216759e-06\\
};
\addlegendentry{\tiny \sffamily Vorticity FEEC}

\addplot [color=mycolor2, line width=1.2pt]
  table[row sep=crcr]{%
0	0\\
0.25	1.71221452279002e-08\\
0.5	3.47485501133082e-08\\
0.75	5.28191537783006e-08\\
1	7.13890197818919e-08\\
1.25	9.04861693717713e-08\\
1.5	1.10058256310004e-07\\
1.75	1.30134232023895e-07\\
2	1.50695412265923e-07\\
2.25	1.7175515770767e-07\\
2.5	1.93335713492142e-07\\
2.75	2.15480806490346e-07\\
3	2.37952420960847e-07\\
3.25	2.60983524420819e-07\\
3.5	2.8455132234612e-07\\
3.75	3.08551555452204e-07\\
4	3.33045914579923e-07\\
4.25	3.58137769327101e-07\\
4.5	3.83661574662556e-07\\
4.75	4.09632846802029e-07\\
5	4.36072885007585e-07\\
5.25	4.62921250247256e-07\\
5.5	4.90381918944406e-07\\
5.75	5.18349278690731e-07\\
6	5.46805508976925e-07\\
6.25	5.75758408073198e-07\\
6.5	6.05202123516839e-07\\
6.75	6.35148382215127e-07\\
7	6.65597960474682e-07\\
7.25	6.96551515602598e-07\\
7.5	7.2799139666389e-07\\
7.75	7.59938901102574e-07\\
8	7.92397798228581e-07\\
8.25	8.25359163573424e-07\\
8.5	8.58828673910241e-07\\
8.75	8.92794091460183e-07\\
9	9.27284458883521e-07\\
9.25	9.62267311622553e-07\\
9.5	9.97781705155853e-07\\
9.75	1.03379537358881e-06\\
10	1.07032679457255e-06\\
10.25	1.10734707320455e-06\\
10.5	1.14497195354565e-06\\
10.75	1.18305906935811e-06\\
11	1.22169957692192e-06\\
11.25	1.26082656842403e-06\\
11.5	1.30051990659375e-06\\
11.75	1.34068143616101e-06\\
12	1.38144258119203e-06\\
12.25	1.42269777076308e-06\\
12.5	1.46455500272706e-06\\
12.75	1.5069054599524e-06\\
13	1.54980868996628e-06\\
13.25	1.59326125384773e-06\\
13.5	1.63725868465212e-06\\
13.75	1.68177440565375e-06\\
14	1.72693146698843e-06\\
14.25	1.77259079589553e-06\\
14.5	1.81882710837493e-06\\
14.75	1.86561563864231e-06\\
15	1.91298976391563e-06\\
15.25	1.9608354441411e-06\\
15.5	2.00942117695356e-06\\
15.75	2.05852812702102e-06\\
16	2.10823335017149e-06\\
16.25	2.15841493155768e-06\\
16.5	2.2093870057931e-06\\
16.75	2.26079663863465e-06\\
17	2.3129049545397e-06\\
17.25	2.3655910243737e-06\\
17.5	2.41889549308526e-06\\
17.75	2.47272309461358e-06\\
18	2.52733467128801e-06\\
18.25	2.58299509555511e-06\\
18.5	2.63880261282022e-06\\
18.75	2.69525648757055e-06\\
19	2.75234885559031e-06\\
19.25	2.8100917578133e-06\\
19.5	2.8684942438228e-06\\
19.75	2.92755033223993e-06\\
20	2.98727649606967e-06\\
20.25	3.04766797169495e-06\\
20.5	3.10873727770044e-06\\
20.75	3.17049551393205e-06\\
21	3.23294719774393e-06\\
21.25	3.29609148374363e-06\\
21.5	3.35993863316096e-06\\
21.75	3.42450467010824e-06\\
22	3.48977668363292e-06\\
22.25	3.55578614927444e-06\\
22.5	3.62252090792585e-06\\
22.75	3.68998833491323e-06\\
23	3.75821927226882e-06\\
23.25	3.82719861779605e-06\\
23.5	3.89694118825791e-06\\
23.75	3.96743964088792e-06\\
24	4.03874393680993e-06\\
24.25	4.11081000273212e-06\\
24.5	4.18367563108532e-06\\
24.75	4.25736380332208e-06\\
25	4.33183891180979e-06\\
};
\addlegendentry{\tiny \sffamily Vakonomic (Ours)}

\end{axis}

\end{tikzpicture}%
    \end{minipage}
    \caption{Conserved quantities for the 2D Laplacian-eigenfunction vortices, each plotted as relative error against the initial value at $t = 0$. \emph{Top row:} the quadratic and quartic Casimir power traces. Our method conserves both exactly, up to the tolerance of the implicit Cayley timestep solver, whereas they drift under the functional-fluids-on-surfaces baseline. \emph{Bottom row:} energy and enstrophy. Both our method and the baseline conserve energy up to the tolerance of the fixed-point iteration in the implicit trapezoidal update. Although our method does not explicitly target the weak enstrophy, it is preserved well here: the vortices are well resolved by the available degrees of freedom, and little enstrophy is lost to discretization---more accurately than the baseline, which conserves it only up to the tolerance of its Cayley solver.}
    \label{fig:Laplacian_eigenfunctions_plot_comparisons}
\end{figure}

\begin{figure}
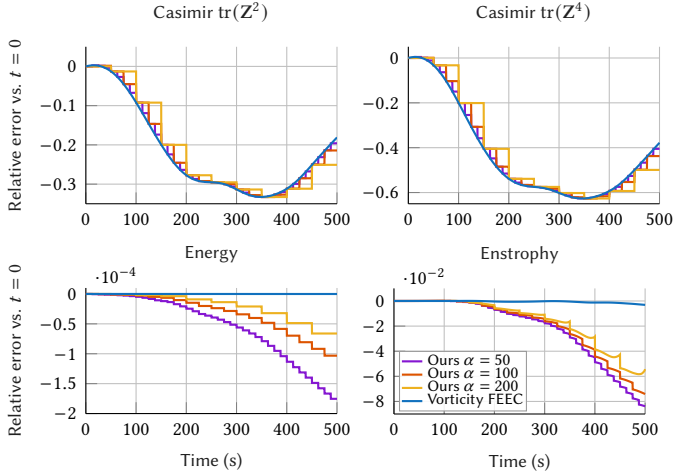

    \centering
    \begin{minipage}{0.47\columnwidth}
        \input{figs/bsplines_basis/plots/reset_sweep_casimir2}
    \end{minipage}\hfill
    \begin{minipage}{0.47\columnwidth}
        \input{figs/bsplines_basis/plots/reset_sweep_casimir4}
    \end{minipage}
    \begin{minipage}{0.47\columnwidth}
        \input{figs/bsplines_basis/plots/reset_sweep_energy}
    \end{minipage}\hfill
    \begin{minipage}{0.47\columnwidth}
        \input{figs/bsplines_basis/plots/reset_sweep_vort2}
    \end{minipage}
    \caption{ %
    Reset frequency, denoted as $\alpha$, swept from every 50 to every 200 steps, probing how our solver behaves under different reset rates. \emph{Top row:} the quadratic and quartic Casimir power traces. As resets become more frequent, the Casimir behavior converges toward the finite-element Lagrange--d'Alembert case, as expected---that limit corresponds to resetting infinitely often. \emph{Bottom row:} energy and enstrophy. There is a clear trade-off: resetting less often lets instability grow but dissipates energy and enstrophy less, while resetting too often dissipates these quantities more. 
    }
    \label{fig:Laplacian_eigenfunctions_reset_sweep}
\end{figure}

\subsubsection{2D Taylor--Green Vortex}
\begin{figure}[tbp]
    \includegraphics[width=\columnwidth]{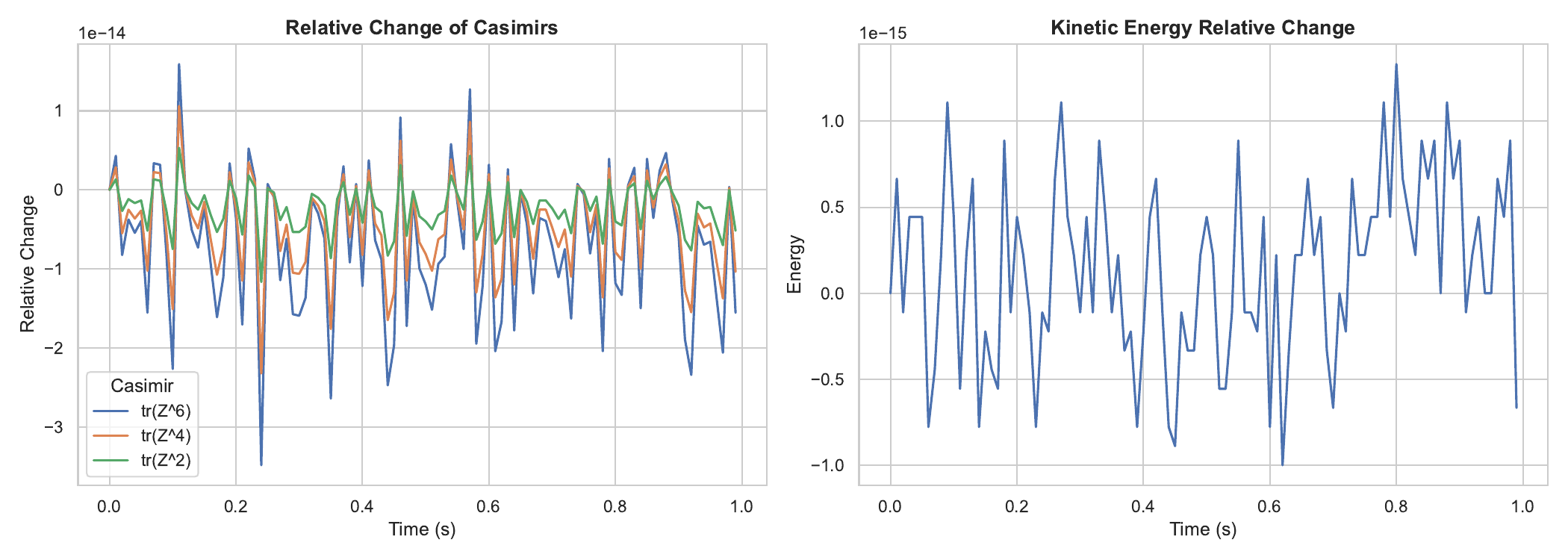}
    \includegraphics[width=\columnwidth]{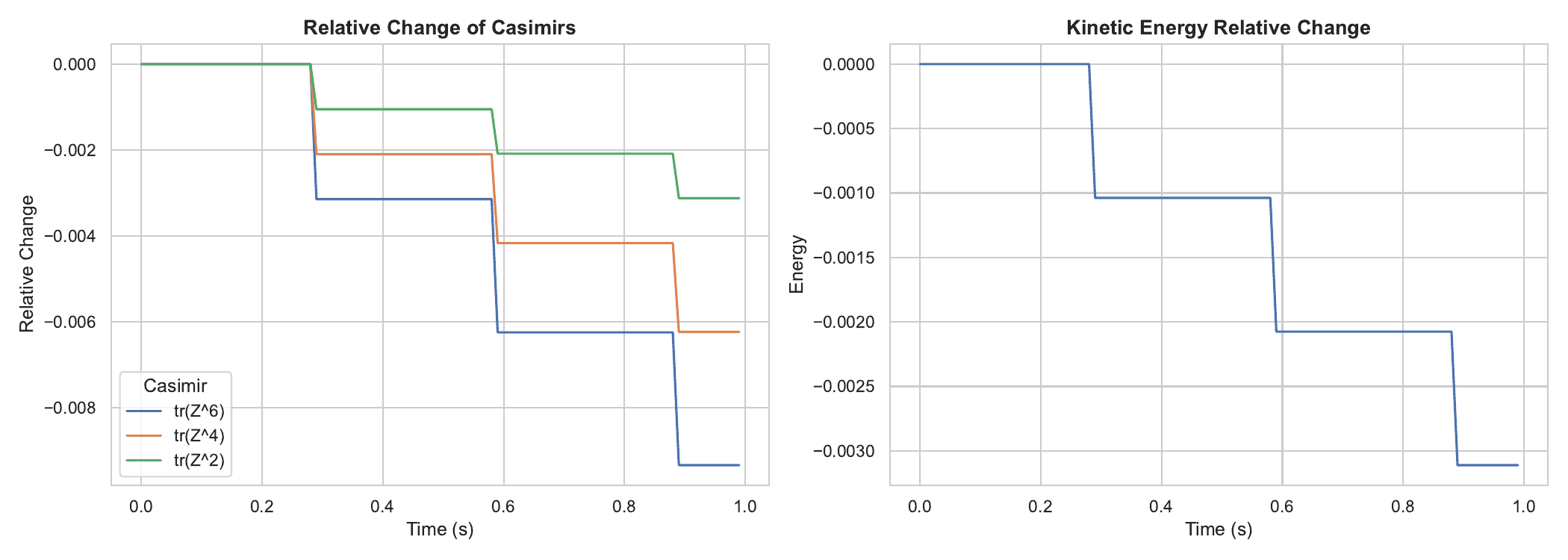}
    \caption{Casimirs $\tr(\bZ^k)$ and energy for the first 400 frames of Taylor-Green Vortex simulation. \textbf{Above:} No resetting. \textbf{Below:} Resetting every 120 frames.}
    \label{fig:tgv_noreset_energycasimirs}
\end{figure}

The Taylor--Green vortex is a steady-state solution to the incompressible Euler equations
\eqref{eq:inc_euler} on the domain $M = [-1, 1]^2$,
\begin{align}
    \bu_\mathrm{TGV} (x, y) = \begin{pmatrix}
        \sin \pi x \cos \pi y \\
        - \cos \pi x \sin \pi y
    \end{pmatrix}.
\end{align}
The field $\bu(x, y, t) = \bu_\mathrm{TGV}(x, y)$ solves the incompressible Euler equations for all $t$: the velocity, and hence the vorticity, remains constant in time. This provides a natural setting in which to study the convergence of our algorithm, which we report for our choice of IGA--FEEC basis (see \Cref{fig:convergence_studies}).

The same experiment serves to study the invariants of motion. Although the velocity is steady, the flow map is not, and because we track the inverse flow map explicitly in $\blambda$, this makes for a more demanding test than the stationarity of $\bu$ suggests. We run the simulation on a Delaunay triangulation of the domain with 11.6k vertices and a timestep of $\Delta t = 0.0025\,\mathrm{s}$, and contrast the evolution of an advected dye, initialized to the $x$-component of position, with and without resetting (\Cref{fig:tgv_reset_cmp}). Without resetting, the $x$-component of $\blambda$ develops dispersive discretization artifacts as the shearing accumulates, and the method loses accuracy. \Cref{fig:tgv_noreset_energycasimirs} plots the relative change in energy and Casimirs without resetting over the first $1\,\mathrm{s}$ (400 frames), before these artifacts appear; both the Casimirs $\tr(\bZ^k)$ and the energy are conserved to machine precision, on the order of $10^{-14}$. With resetting every $0.3\,\mathrm{s}$ (120 frames), each reset produces a small jump in the conserved quantities, visible in \Cref{fig:tgv_noreset_energycasimirs}. These losses accumulate slowly: over the full $5\,\mathrm{s}$ simulation the energy drops by only about $1.6\%$.

\subsubsection{Eigenfunction Vortices}
To probe how each method handles multi-scale vorticity fields, we initialize the flow with a random streamfunction drawn from the low-frequency spectrum of the scalar Laplacian: we take the first 20 eigenfunctions and assign each an independent, identically distributed Gaussian coefficient, yielding a smooth but spatially complex initial vorticity field. This construction's band-limited spectrum makes the onset of spurious high-frequency content easy to spot. As the flow evolves, the Vorticity FEEC baseline develops pronounced dispersive artifacts, with energy leaking into oscillatory modes that were not present initially. Our method remains markedly cleaner over the same horizon. See \Cref{fig:laplacian_eigenfunction_visual_comparison}. We attribute the difference to the structure each scheme preserves: our discretization admits a discrete analogue of Kelvin's circulation theorem, tying the evolution to a conserved geometric quantity along advected discrete loops, whereas the baseline enforces only a weak notion of enstrophy.

To evaluate the traces of powers of the dual variable $\bZ$ (i.e., its spectrum) for the LdA methods implemented in FEEC, we apply $\textsc{Reset}$ to obtain corresponding $(\blambda,\bmu)$ fields and then $\textsc{Reconstruct}$ to assemble a matrix $\bZ$ whose spectrum we can inspect. This lets us compare the Casimir-preservation behavior of the LdA methods against our solver (see \Cref{fig:Laplacian_eigenfunctions_plot_comparisons}).

To investigate how resetting affects the preservation of Casimir invariants, \Cref{fig:Laplacian_eigenfunctions_reset_sweep} presents a parameter sweep of the reset interval $\alpha$ from 50 to 200 steps. The results reveal a trade-off: while less frequent resets (larger $\alpha$) better preserve energy and enstrophy by reducing dissipation, they allow numerical instabilities to grow. Conversely, more frequent resets suppress these instabilities at the cost of higher dissipation.  This is consistent with the suggestion that enstrophy conservation is perhaps undesirable in discrete fluid models due to its generation of excessive background noise; see, \eg, \cite[Section 5]{Modin:2025:MHD}.

\subsubsection{Verification of Noether charges}

To verify the conservation of the discrete circulation discussed in \Cref{thm:StrongKelvinCirculationTheorem}, consider an initial divergence-free velocity field \(\bu\) and define the corresponding infinitesimal generator \(\bY = \bA_{\bu}\).
Verification of Kelvin's circulation theorem consists of monitoring the invariance of the  pairing \(\langle \bZ \mid \bY \rangle\) as \(\bZ\) evolves by the coadjoint dynamics and \(\bY\) evolves by the adjoint dynamics.
Importantly, recall that \(\bZ\) is represented in low-rank form via \eqref{eq:Z-Clebsch}, enabling an efficient evaluation of the pairing without explicitly forming full-rank matrices.  While a direct adjoint update of \(\bY = \bA_{\bu}\) would require evolving a sparse but full-rank matrix, the structure of the pairing allows for simple advection of the action of $\bY$ on the  vectors appearing in the low-rank representation of \(\bZ\).

Concretely, let \(\blambda_i(t)\) and \(\bmu_i(t)\) denote the vectors defining \(\bZ\),
so that the vectors
\[
\hat\blambda_i(t) = \bY(t)\,\blambda_i(t) %
\qquad
\hat\bmu_i(t) = \bY(t)\,\bmu_i(t)%
\]
capture the action of \(\bY\) on the evolving basis of \(\bZ\) with their values at time zero equal to \[
\hat\blambda_i(0) = \bA_{\bu}\,\blambda_i(0),
\qquad
\hat\bmu_i(0) = \bA_{\bu}\,\bmu_i(0).
\]
Since $\frac{d}{dt} \langle \bZ\mid \bY \rangle = 0$, it can be verified that $\{(\hat\blambda,\hat\bmu)\}_{i=1}^m$ must also evolve the same way as $\{(\blambda,\bmu)\}_{i=1}^m$; i.e., through a rotation.
The Noether charge can then be evaluated as the combination
\[
\langle \bZ \mid \bY \rangle(t)
=
\sum_i
\bmu_i(t)^\top \hat\blambda_i(t)
-
\blambda_i(t)^\top \hat\bmu_i(t).
\]
This formulation avoids explicitly advecting the full matrix \(\bY\), instead requiring only the evolution of the vectors \(\hat\blambda_i(t)\) and \(\hat\bmu_i(t)\). 
In practice, we observe in \Cref{fig:noethercharge} that a random charge initialized from a divergence-free velocity field is conserved up to the tolerance of the orthogonal advection function (in this case, the Cayley transform).

\begin{figure}
    \centering
    \definecolor{mycolor1}{rgb}{0.06600,0.44300,0.74500}%
\definecolor{mycolor2}{rgb}{0.12941,0.12941,0.12941}%
\begin{tikzpicture}

\begin{axis}[%
width=0.8\linewidth,
height=0.4\linewidth,
at={(0,0)},
scale only axis,
xmin=475.25,
xmax=500,
xtick={475.25, 480, 485, 490, 495, 500},
xticklabels={475, 480, 485, 490, 495, 500},
xlabel style={font=\color{mycolor2}},
xlabel={\small \sffamily Time ($\mathrm{s}$)},
ymin=-9e-8,
ymax=0,
ytick={0,-1e-8,-2e-8,-3e-8,-4e-8,-5e-8,-6e-8,-7e-8,-8e-8,-9e-8},
ylabel style={font=\color{mycolor2}},
ylabel={\small \sffamily Relative error vs. $t=475\,\mathrm{s}$},
axis background/.style={fill=white},
title style={font=\small \sffamily\color{mycolor2}},
title={Random Noether charge},
axis x line*=bottom,
axis y line*=left,
xmajorgrids,
ymajorgrids,
legend style={legend cell align=left, align=left},
legend pos= south west,
xticklabel style={font=\small},
yticklabel style={font=\small}
]
\addplot [color=mycolor1, line width=1.4pt, forget plot]
  table[row sep=crcr]{%
475.25	-0\\
475.5	-4.92373221820688e-10\\
475.75	-1.0792390946039e-09\\
476	-1.79466296992367e-09\\
476.25	-2.96107975650472e-09\\
476.5	-3.9761678563929e-09\\
476.75	-5.0208854101374e-09\\
477	-6.09427417886895e-09\\
477.25	-7.20551780805298e-09\\
477.5	-8.19152178935402e-09\\
477.75	-9.27211545235448e-09\\
478	-1.02786833524293e-08\\
478.25	-1.12835278707593e-08\\
478.5	-1.22430357025302e-08\\
478.75	-1.31863241144747e-08\\
479	-1.41480385507148e-08\\
479.25	-1.50324881996465e-08\\
479.5	-1.59289152712473e-08\\
479.75	-1.69588939686979e-08\\
480	-1.76641828789561e-08\\
480.25	-1.8359128602306e-08\\
480.5	-1.85011070982388e-08\\
480.75	-1.88064065075598e-08\\
481	-1.919144040195e-08\\
481.25	-1.95763108083988e-08\\
481.5	-2.00306402575197e-08\\
481.75	-2.05911734687142e-08\\
482	-2.13451621540904e-08\\
482.25	-2.21687884959329e-08\\
482.5	-2.31495138764166e-08\\
482.75	-2.43058029390358e-08\\
483	-2.57483215724544e-08\\
483.25	-2.68093771391482e-08\\
483.5	-2.79331779617228e-08\\
483.75	-2.92115767470787e-08\\
484	-3.06201761382252e-08\\
484.25	-3.21018107098491e-08\\
484.5	-3.35956598903879e-08\\
484.75	-3.50638351885035e-08\\
485	-3.64602965654793e-08\\
485.25	-3.78634300526945e-08\\
485.5	-3.92104762338953e-08\\
485.75	-4.05166027993952e-08\\
486	-4.19153570303841e-08\\
486.25	-4.33835455233926e-08\\
486.5	-4.49367866869665e-08\\
486.75	-4.65462736561829e-08\\
487	-4.80265533375836e-08\\
487.25	-4.93409994439767e-08\\
487.5	-5.05214221266418e-08\\
487.75	-5.14656956499814e-08\\
488	-5.23012012573629e-08\\
488.25	-5.30180067345184e-08\\
488.5	-5.36752886015364e-08\\
488.75	-5.43384769531782e-08\\
489	-5.47953193857452e-08\\
489.25	-5.5200133068261e-08\\
489.5	-5.55512604325321e-08\\
489.75	-5.60516649994752e-08\\
490	-5.66824531306629e-08\\
490.25	-5.7514742885181e-08\\
490.5	-5.80725093526108e-08\\
490.75	-5.88179134728424e-08\\
491	-5.90389490087327e-08\\
491.25	-5.93255872993856e-08\\
491.5	-5.96677818820699e-08\\
491.75	-6.00414203772989e-08\\
492	-6.04181596723618e-08\\
492.25	-6.12592927406562e-08\\
492.5	-6.20502618834596e-08\\
492.75	-6.27391487372639e-08\\
493	-6.32738561638534e-08\\
493.25	-6.36201469562196e-08\\
493.5	-6.38213724522787e-08\\
493.75	-6.404767999154e-08\\
494	-6.44010047882206e-08\\
494.25	-6.48814203797115e-08\\
494.5	-6.54721599190438e-08\\
494.75	-6.61458425552091e-08\\
495	-6.68432312647073e-08\\
495.25	-6.75967333481791e-08\\
495.5	-6.83623145515434e-08\\
495.75	-6.9211215067102e-08\\
496	-6.96328763750774e-08\\
496.25	-7.02429073305842e-08\\
496.5	-7.10706281120621e-08\\
496.75	-7.21007981354599e-08\\
497	-7.28051938158334e-08\\
497.25	-7.36196418223504e-08\\
497.5	-7.47393786179086e-08\\
497.75	-7.57251030927586e-08\\
498	-7.69502947597612e-08\\
498.25	-7.83916358294442e-08\\
498.5	-7.94756239591024e-08\\
498.75	-8.02001918952089e-08\\
499	-8.08950890227337e-08\\
499.25	-8.13969516253434e-08\\
499.5	-8.19435191687858e-08\\
499.75	-8.25376565539558e-08\\
500	-8.29952271193328e-08\\
};
\end{axis}
\end{tikzpicture}%
    \caption{A random divergence-free velocity field $\bu_0$ is selected to initialize matrix vectors $\{(\hat\blambda, \hat\bmu)\}_{i=1}^m$ through the matrix $\bY(0)= \bA_{\bu_0}$. Then the variables $\{(\hat\blambda, \hat\bmu)\}_{i=1}^m$ are advected alongside with the Clebsch variables used during the simulation $\{(\blambda, \bmu)\}_{i=1}^m$. The Noether charge, here $\langle \bZ | \bY\rangle(t)$ remains preserved up to the tolerance of the Cayley solver used to compute the rotations used for each timestep update of these variables.}
    \label{fig:noethercharge}
\end{figure}
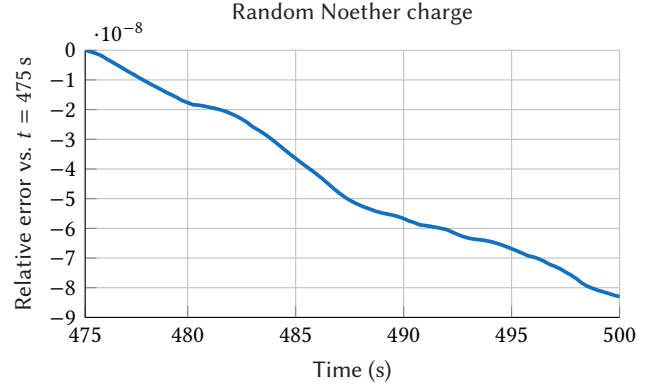

\subsection{Mollified Point Vortex Systems}
\label{sec:MollifiedPointVortexSystems}
Another class of interesting solutions to the Euler equations is provided by the evolution of mollified point vortices, where $\omega$ is a sum of Gaussian functions.
It will now be shown that the vakonomic method can approximate the motion of such vortices. 

\subsubsection{2D Shielded Taylor vortices}

Following \cite{Nabizadeh:2022:CF}, we initialize a pair of identical vortices on the square domain $[-\pi,\pi]^2$, centered at $(\pm\tfrac12 d, 0)$ for a separation distance $d$. For our method and the finite-element functional-fluids-on-surfaces baseline, the domain is closed, with homogeneous Dirichlet conditions imposed on the vorticity at the boundary; the domain is chosen large enough relative to the vortex cores that the boundary has negligible influence on the interaction.
Each vortex is prescribed through its vorticity using the radial profile
\[
\omega(r)=\frac{1}{a}\Big(2-\frac{r^2}{a^2}\Big)\,
\exp\!\left(\tfrac12\Big(1-\frac{r^2}{a^2}\Big)\right),\qquad a=0.3,
\]
where $r$ is the radial distance to the vortex center. This profile has a positive core and, beyond $r=\sqrt{2}\,a$, a surrounding annulus of opposite-sign vorticity. The two contributions cancel exactly, so each vortex carries zero net circulation, $\int\omega\,\mathrm{d}A=0$. This ``shielding'' makes the far field decay faster than a bare monopole and localizes the interaction between the two vortices. As the cores share the same sign, the like-signed pair co-rotates and merges, probing the scheme's handling of filamentation and small-scale vorticity transport.

For comparison, we also run the matrix-hydrodynamics scheme of Zeitlin on this benchmark. Unlike the two methods above, this baseline is formulated on the periodic flat torus, so its instance of the experiment uses the doubly periodic domain $[-\pi,\pi]^2$ rather than the closed domain with Dirichlet boundaries.

\begin{figure}
    \centering
    \includegraphics[width=\columnwidth]{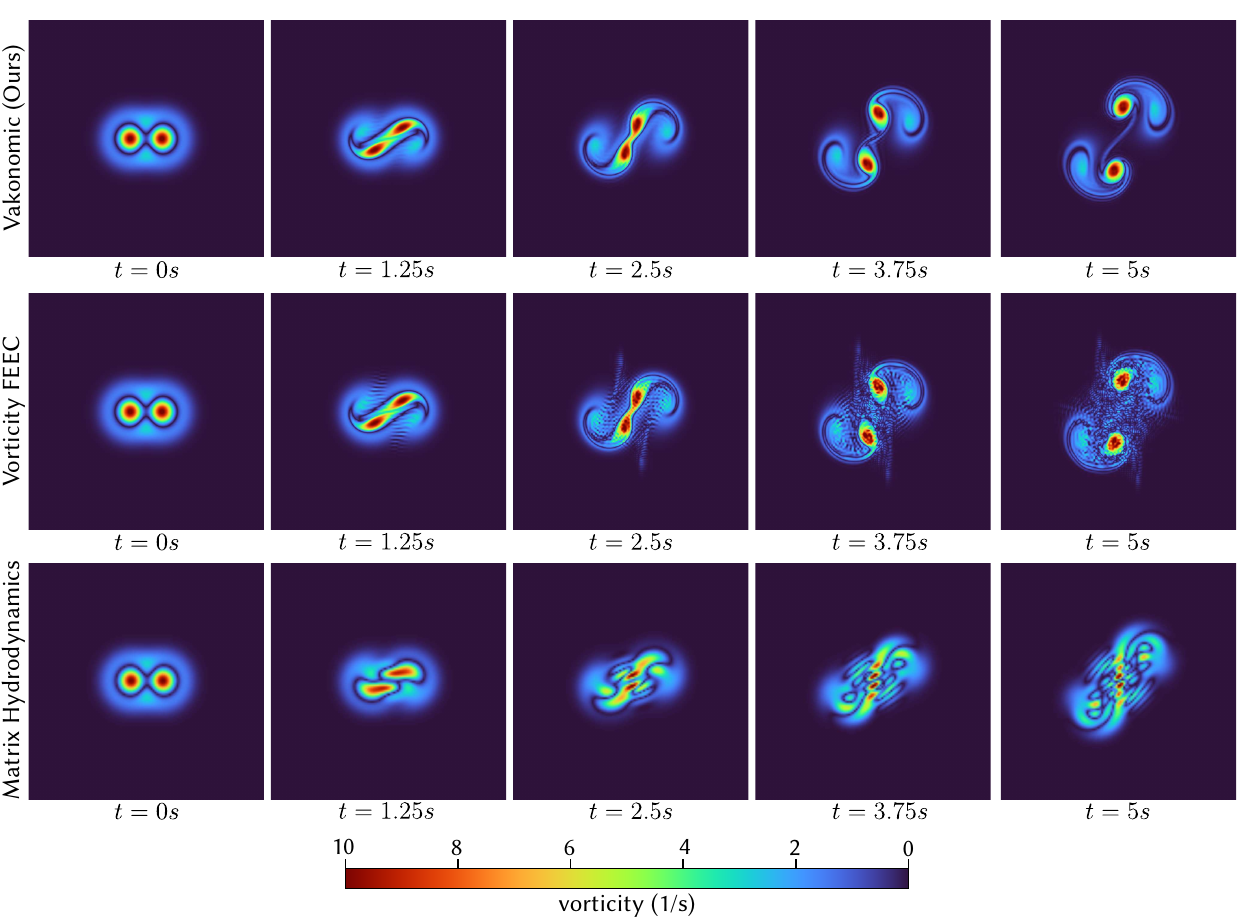}
    \caption{2D shielded Taylor-vortex experiment, comparing our vakonomic formulation against the functional-fluids-on-surfaces FEEC baseline. Both methods are run at $128\times128$ resolution with timestep $\Delta t = 1/48\,\mathrm{s}$ and quadratic spatial accuracy; for our method, the flow map is reset every 10 frames (roughly every $10/48\,\mathrm{s}$). Note the lack of dispersive effects in the vakonomic method.}
    \label{fig:shieldedTaylorVortices}
\end{figure}

\begin{figure}
    \includegraphics[width=\columnwidth]{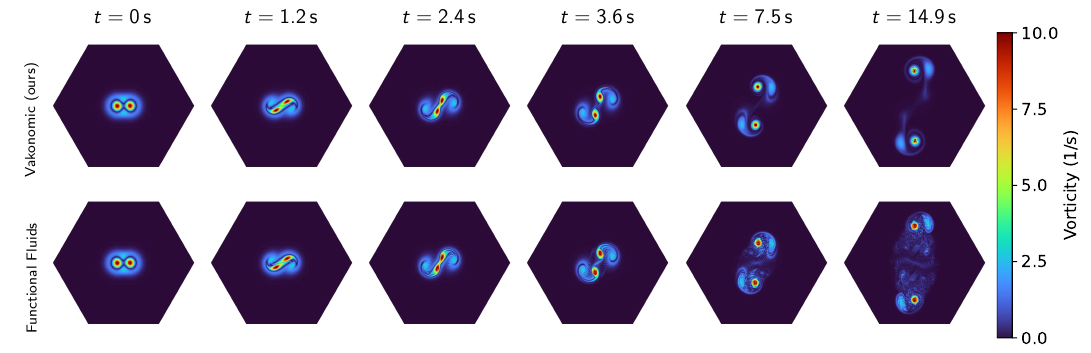}
    \caption{Evolution of the Shielded Taylor Vortices over time.} %
   \label{fig:shielded-vorts}
\end{figure}

\Cref{fig:shieldedTaylorVortices} compares the performance of the proposed vakonomic \algref{alg:sim_loop_w_reset} against the vorticity FEEC and matrix hydrodynamics (MHD) baselines, using the B-spline discretization at resolution $128^2$.  Observe that the vakonomic approach produces more realistic behavior as the vortices evolve, with limited dispersion despite the use of a relatively low resolution grid.  This is a consequence of the discrete relabeling symmetry satisfied by the vakonomic model, leading to correct geometric behavior independent of resolution. Interestingly, MHD produces the least physical evolution despite being the only model that is Lie-algebraically closed, perhaps due to its substantially slower convergence with grid resolution. \Cref{fig:shielded-vorts} additionally displays the same experiment using the Whitney form construction 
with a longer integration in time.

\subsubsection{2D Leapfrogging vortices}
Long-term behavior of our method can be studied 
following \cite{Nabizadeh:2024:FIP}, by superposing four regularized (Lamb--Oseen--type) \cite{Lamb:1895:HD, Batchelor:2000:IFD} point vortices on the horizontal centerline
of $[-\pi,\pi]^2$.  Each point vortex is modeled with a Gaussian vorticity blob
\[\omega_i(\bx)=\pm\frac{\Gamma}{2\pi}\exp\left(-\frac{|\bx-\bx_i|^2}{2\sigma^2}\right),\]
of core width $\sigma=9\pi/256\,\mathrm{m}$ and strength $\Gamma=100\,\mathrm{s}^{-1}$.
The four centers form two coaxial, oppositely signed dipoles of separations $d_a$ and $d_b$:
$(+,-)$ at $(\mp\tfrac12 d_a,0)$ and $(+,-)$ at $(\mp\tfrac12 d_b,0)$. 
Each $(+,-)$ dipole self-propels perpendicular to its connecting axis. Because both dipoles translate in the same direction, the narrower pair accelerates, contracts, and slips through the wider pair, which in turn expands and slows. The roles then alternate, producing the characteristic leapfrogging motion.
As the total circulation produced is zero, this test stresses invariant preservation and the long-term accuracy
of mutual induction between coherent structures.
\begin{figure}
    \centering
    \includegraphics[width=\linewidth]{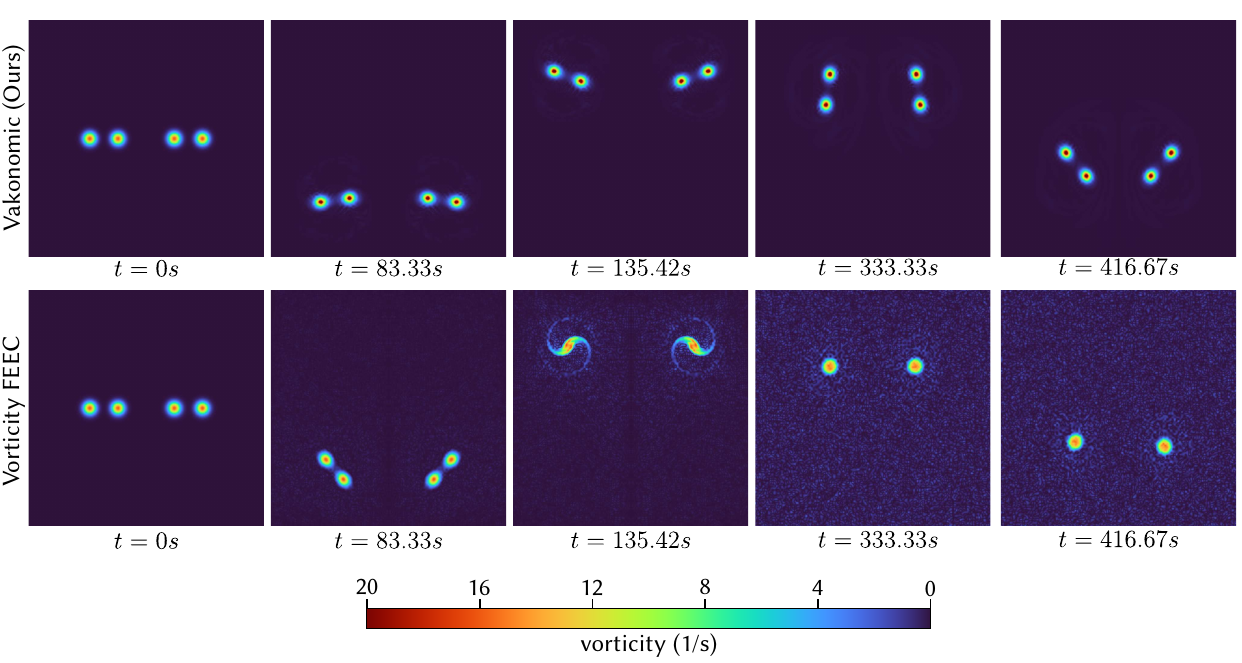}
    \caption{%
    2D vortex leapfrogging marathon experiment, comparing our vakonomic formulation against the functional-fluids-on-surfaces FEEC baseline. Both methods are run at $128\times128$ resolution with timestep $\Delta t = 1/48\,\mathrm{s}$ and quadratic spatial accuracy; for our method, the flow map is reset every 10 frames (roughly every $10/48\,\mathrm{s}$). Note the un-physical collision of vortices in the baseline that is absent from the vakonomic approach.}
    \label{fig:leapfrogVortices}
\end{figure}

\begin{figure}
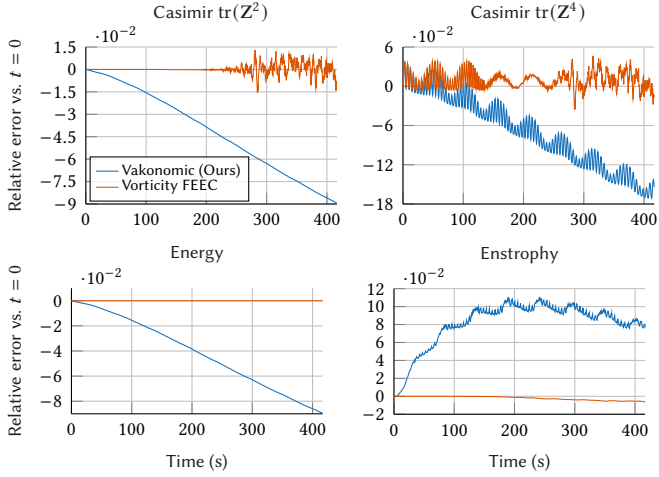

    \centering
    \begin{minipage}{0.47\columnwidth}
        \input{figs/bsplines_basis/plots/leapfrog_casimir2}
    \end{minipage}\hfill
    \begin{minipage}{0.47\columnwidth}
        \input{figs/bsplines_basis/plots/leapfrog_casimir4}
    \end{minipage}
    \begin{minipage}{0.47\columnwidth}
        \input{figs/bsplines_basis/plots/leapfrog_energy}
    \end{minipage}\hfill
    \begin{minipage}{0.47\columnwidth}
        \input{figs/bsplines_basis/plots/leapfrog_vort2}
    \end{minipage}
    \caption{ %
    Conserved quantities for the leapfrogging vortices, each plotted as relative error against the initial value at $t = 0$. \emph{Top row:} the quadratic and quartic Casimir power traces. \emph{Bottom row:} energy and enstrophy. Here the Casimirs follow the same trend as the kinetic energy. This is expected: in the trace of powers of $\bZ$, the covariance-like structure lets the inverse-flow-map rows and columns be normalized, leaving the remaining rows and columns associated with the kinetic energy measured in $\cF$. With resets applied often (every 10 frames), the Casimirs therefore dissipate in step with the energy. Although the plot may suggest that structure is lost, our algorithm still follows a motion informed by a discrete particle-relabeling symmetry and the underlying Lie--Poisson structure, yielding visually more accurate results.  
    }
    \label{fig:leapfrog_vortices_plot_comparisons}
\end{figure}

\Cref{fig:leapfrogVortices} displays the results of an extended vortex leapfrogging simulation, comparing the vakonomic approach to a vorticity FEEC baseline.  Observe that the proposed method maintains appropriate vortex separation for all time with no change in the circulation (recall that this quantity is a discrete invariant). Conversely, the baseline method merges vortices un-physically after about 135s, leading to incorrect results in the long-term.

\Cref{fig:leapfrog_vortices_plot_comparisons} plots the relative error in the conserved quantities for the leapfrogging vortices. Notably, the dissipation in Casimirs (due to resetting) seen in the vakonomic case closely mimics the 
dissipation trend of the kinetic energy. This coupling arises because the covariance-like structure seen in \eqref{eq:blkdiagBwithLambdaAndMu} normalizes the components representing the inverse-flow map, leaving the Casimirs dominated by the kinetic energy metric in $\mathcal{F}$. Thus, frequent resetting (every 10 frames) drives down both energy and Casimirs simultaneously while minimizing overall dispersion. Crucially, this decay is solely a consequence of the resetting procedure employed to make this instance of the vakonomic formulation numerically tractable. Without resets, the system preserves its geometric nature, evolving along a path governed by the discrete particle-relabeling symmetry and Lie–Poisson structure (cf. \Cref{fig:tgv_noreset_energycasimirs}.
Designing alternative resetting schemes that avoid these dissipative drawbacks, akin to the hybrid grid-particle methods explored in \cite{Nabizadeh:2024:FIP},  remains an exciting direction for future research.\subsubsection{Six Vortices on an Oblate Spheroid}
Another interesting test monitors the behavior of multiple point vortices on a surface with positive sectional curvature \citep{Vankerschaver:2014:PVD}.  Following Azencot et al.\ \cite{Azencot:2014:FFS}, 
six Gaussian vortices are instantiated on an oblate spheroid (with aspect ratio $2 : 1$). The vorticity field is initialized as 
\begin{align*}
    \omega(\bx) &= \sum_{k=1}^{6} A \exp\!\left(-\frac{\lVert \mathbf{x}-\mathbf{m}_k\rVert^2}{2r^2}\right), \\ 
    \mathbf{m}_k &= \left( 1.1,\, \rho\cos\frac{2\pi k}{6},\, \rho\sin\frac{2\pi k}{6} \right), \quad 1\leq k\leq 6,
\end{align*}
where the parameters $A=2, r=0.1$, and $\rho = 0.835$
are chosen so that the points $\bm_k$ lie on the oblate spheroid mesh of 41k vertices within the slice where $x = 1.1$. 

Running the vakonomic \algref{alg:sim_loop_w_reset} with $\Delta t = 0.025$ and resetting automatically using the FTLE-based approach described above (roughly every $1.5s$)
yields the results shown in \Cref{fig:six-point-vorts}. As expected, the vortices remain separated and equally spaced while rotating for a substantial amount of time before eventually pairing together and merging due to the curvature of the spheroid. Observe that minimal dispersion is produced.  It is worth noting that while the Functional Fluids baseline only tracks the vorticity data, our approach tracks an approximation of the inverse flow map, which is a much more numerically challenging requirement.

We additionally remark that this oblate spheroid mesh is identical to the lower resolution icosphere mesh used before, stretched in $x$-axis by a factor of 2 -- as a result, the triangles are stretched out and the mesh fails to be Delaunay. This is not an issue when using the full Galerkin Hodge stars, but the mass lumped Hodge stars can fail to be invertible on this mesh, so we avoid using them.

\begin{figure}[tbp]
    \includegraphics[width=\columnwidth]{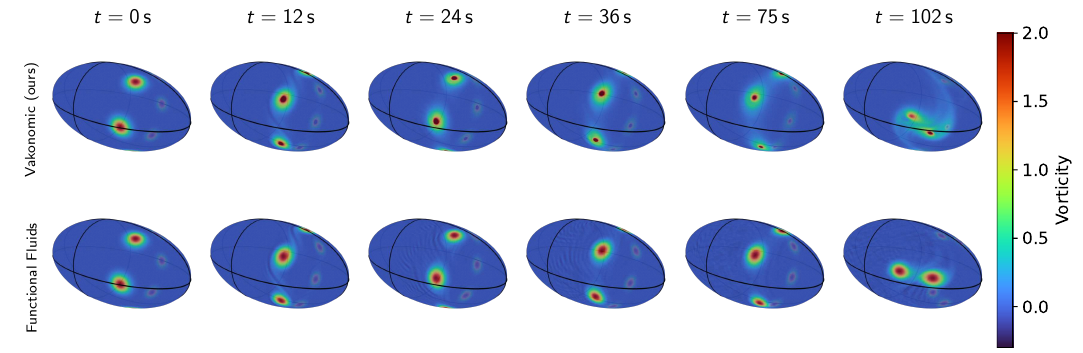}
    \caption{Point Vortices on an Oblate Spheroid. The sphere is made $50\%$ transparent to illustrate vortices on the reverse side.}
    \label{fig:six-point-vorts}
\end{figure}

\subsection{3D Trefoil Knot}
\label{sec:ThreeDimensionalTrefoilKnotExp}
As a further validation, we consider a three-dimensional trefoil knot experiment: a $(2,3)$-torus knot is embedded in the periodic box and assigned a constant circulation along its core, which is then mollified into a vortex tube of finite thickness \cite{Nabizadeh:2024:FIP}. Under minimal dissipation, the knotted filament is expected to deform, undergo vortex reconnection, and separate into distinct rings \cite{Kleckner:2013:CDK}. Our method reproduces this reconnection without loss of vorticity strength, as indicated by the coloring of the filaments in \Cref{fig:trefoilknot}, even at the coarse resolution of $64^3$ with $\Delta t = 1/72\,\mathrm{s}$ and resets every 15 frames (approximately $5/24\,\mathrm{s}$).

\subsection{Flows on Surfaces}
\label{sec:FlowsOnSurfaces}
Finally, to illustrate our method's ability to perform larger scale fluid simulation experiments, we run a number of examples of more complex flows on surfaces. 

\subsubsection{Double Shear Layer}
We begin with the double shear layer on a hyperboloid-like surface from Azencot et al.\ \cite{Azencot:2014:FFS} and San and Maulik \cite{San:2013:CGP}. This experiment initializes the vorticity to be the sum of two layers of opposite sign, as well as a small perturbation as one rotates around the hyperboloid.
\begin{align*}
\omega(x, y, z)
&=
\sum_{j=1}^2 s_j\,\sigma\,\operatorname{cosh}^{-2}\left(\sigma(z-z_j)\right)
\;+\;
\delta \cos\theta, \\
\theta&=\operatorname{atan2}(y,x),
\end{align*}
with parameters
\begin{align*}
(z_1,z_2)&=(0.5,-0.5), \quad (s_1,s_2)=(1,-1), \\
\sigma&=30, \quad \delta=\frac{0.05}{\pi}.
\end{align*}

We run this experiment with a timestep size of $\Delta t = 0.005\,\mathrm{s}$, and with a reset performed every $0.18\,\mathrm{s}$ for the vakonomic experiment, on a mesh with 175k vertices. We compare our results to our implementation of Azencot et al.\ \cite{Azencot:2014:FFS}, and show the results in \Cref{fig:shear-layer}. Our results have notably fewer noisy artifacts, and cleanly illustrate the ``roll-up'' phenomenon theoretically expected in this experiment.

\begin{figure}[tbp]
    \includegraphics[width=\columnwidth]{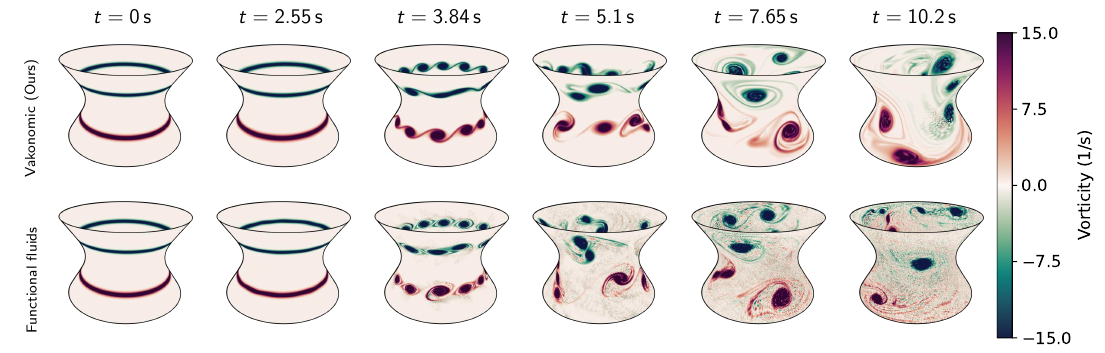}
    \caption{Vorticity evolution over time for the double shear layer experiment. Our method (above) is compared to Azencot et al.\ \cite{Azencot:2014:FFS} (below). }
    \label{fig:shear-layer}
\end{figure}

\subsubsection{Flow on a Sphere}
\label{sec:sphere-flow}
Next, we illustrate some examples of flows on various surfaces. The initial is formed from the Arnold-Beltrami-Childress (ABC) flow, sampled at the barycenter of each face and orthogonally projected to the tangent space of that face, before being averaged to vertices. The ABC flow is initialized by the velocity field
\begin{align*}
    \bu_{A B C} (x, y, z) = \begin{pmatrix}
    \sin a z + \cos a y \\
    \sin a x + \cos a z \\
    \sin a y + \cos a x
    \end{pmatrix},
\end{align*}
where $a = 10$ is a scale parameter. 

We initialize this projected ABC flow on two unit icosphere meshes (i.e., a subdivided icosahedron projected to a sphere), with 6 and 7 subdivisions respectively (i.e., with 41k and 164k vertices, respectively), and run our method with timestep sizes of $\Delta t = 10^{-3}\,\mathrm{s}$ for the higher resolution mesh and $\Delta t = 2.5 \times 10^{-3}\,\mathrm{s}$ for the lower resolution mesh, performing a reset every $0.2$ s. The evolution of the vorticity over time is shown in \Cref{fig:sphere-flow}, along with . %
the finite-time Lyapunov exponent computed using the accumulated $\blambda$ from the previous reset.

\begin{figure}[tbp]
    \includegraphics[width=\columnwidth]{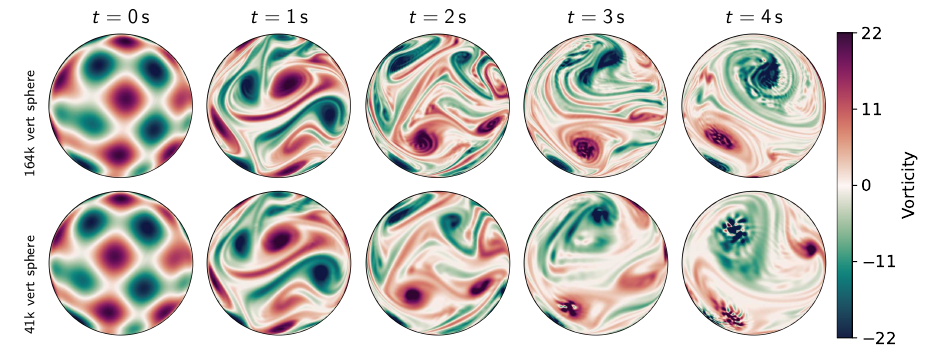}
    \includegraphics[width=0.99\columnwidth]{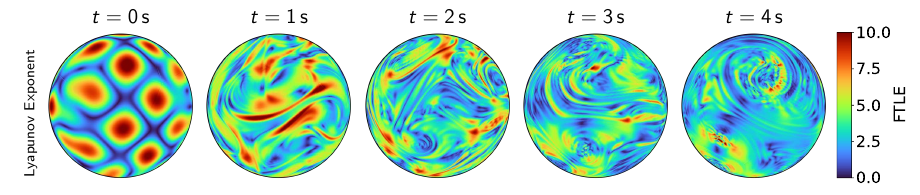}
    \caption{\textbf{Above:} Vorticity evolution on two spheres of different resolution, initialized by a projected ABC flow. \textbf{Below:} The finite-time Lyapunov exponent of the flow visualized on the higher resolution sphere.}
    \label{fig:sphere-flow}
\end{figure}

\subsubsection{Flow on a Helicoid}
Similarly, we run a simulation on a helicoid, a minimal surface, once again initialized with an ABC-flow projected to the tangent space of the surface, and compare our results to Azencot et al.\ \cite{Azencot:2014:FFS}. As before, \Cref{fig:helicoid-flow} confirms that that fewer noisy artifacts are produced by the vakonomic approach compared to the functional fluids method. 

\begin{figure}[tbp]
    \includegraphics[width=\columnwidth]{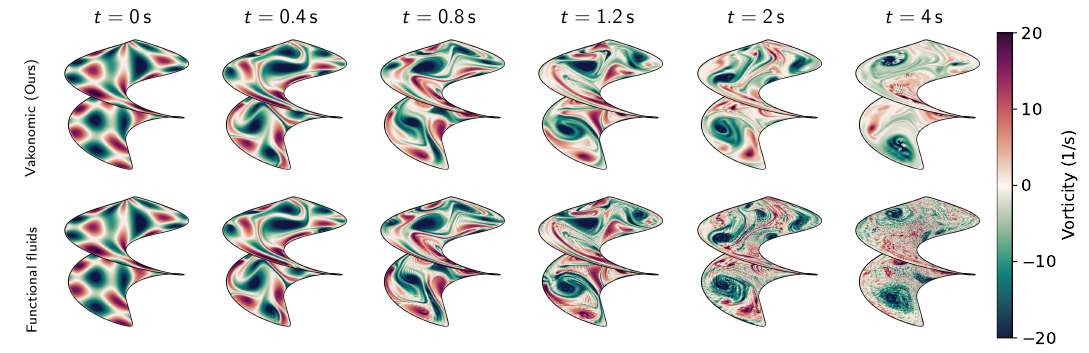}
    \caption{A flow on a helicoid; initialized by a projected ABC flow.}
    \label{fig:helicoid-flow}
\end{figure}

\subsubsection{Flow on a Bunny}
\begin{figure}[htbp]
    \includegraphics[width=\columnwidth]{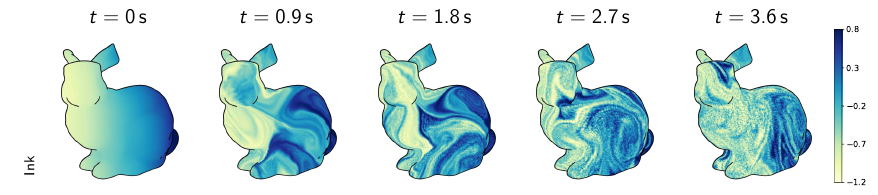}
    \caption{A flow on a Stanford bunny; initialized once again by a projected ABC flow.
    }
    \label{fig:bunny-flow}
\end{figure}

Finally, to conclude, we run a fluid simulation on a Stanford Bunny mesh with 85.7k vertices; this mesh has more intricate geometry than our other examples. We initialize a flow using the same projected ABC flow as in \secref{sec:sphere-flow}, running the simulation with $\Delta t = \frac{1}{3} \times 10^{-3}\,\mathrm{s}$ and performing a reset every $0.12\,\mathrm{s}$. The simulation runs stably for a significant amount of time. In \Cref{fig:bunny-flow}, we plot the evolution of a passively advected dye initialized as the $x$-coordinate function. The results demonstrate mixing behavior characteristic of turbulent phenomena.

\subsection{Buoyancy: stratified flow in the Boussinesq approximation}
\label{sec:Buoynacy}
\begin{figure}[t]
    \centering
    \includegraphics[width=\linewidth]{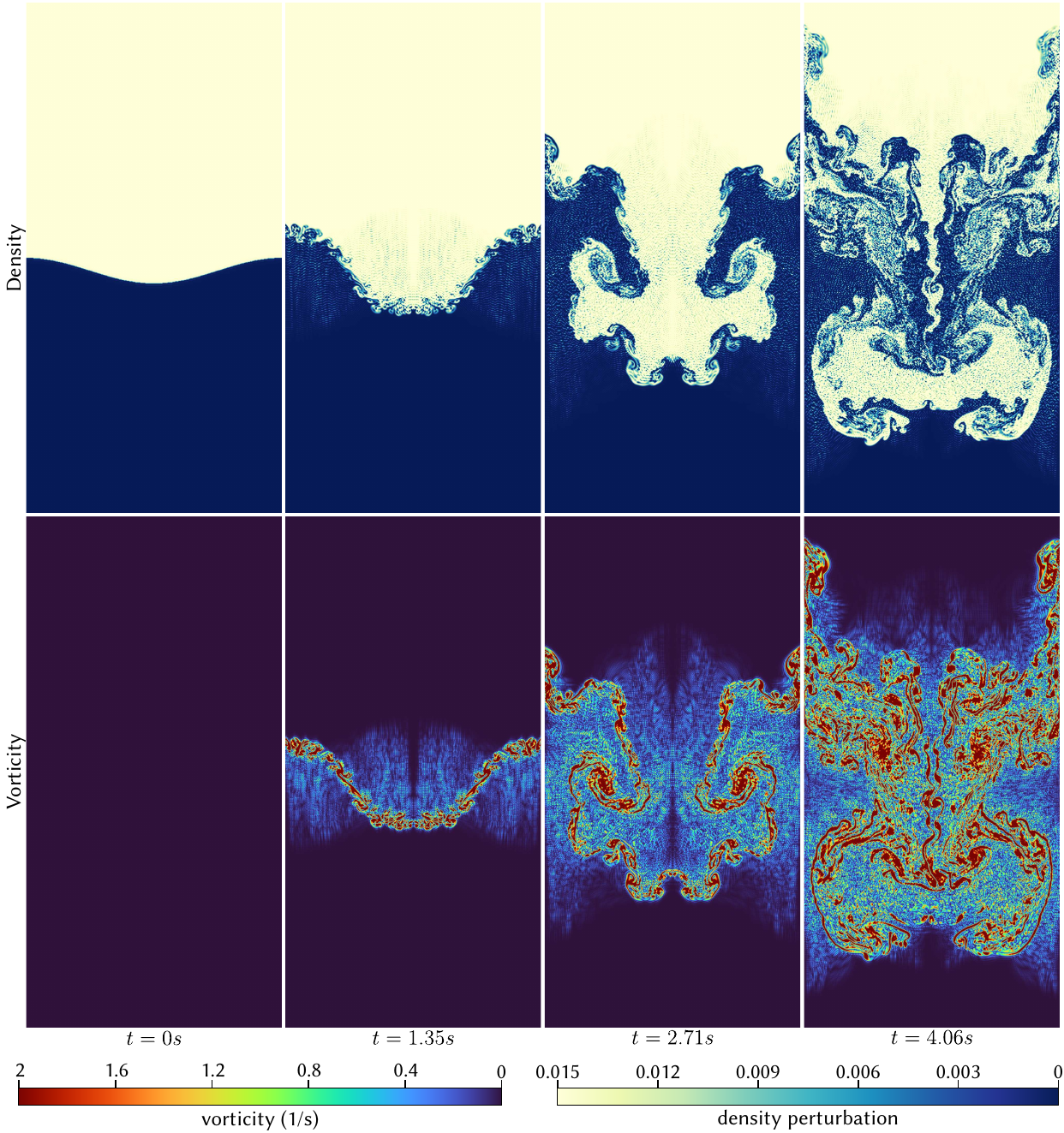}
    \caption{Rayleigh--Taylor instability of a stratified fluid under the Boussinesq approximation. A heavier fluid (white ink) sits atop a lighter fluid (blue) and falls into it, developing the characteristic plumes and intricate mixing patterns as the flow evolves. Resolution $256\times512$, base timestep $\Delta t = 1/96\,\mathrm{s}$ with adaptive timestepping capped at CFL $0.45$.
    }
    \label{fig:buoyancy_demo}
\end{figure}

Besides the purely incompressible case, we also simulate density-stratified flow in the Boussinesq approximation \cite{Tritton:2012:PFD}, where a density perturbation $\delta\brho$ about a constant reference state is advected passively by the incompressible velocity, entering the momentum balance only through gravity.  For a gravitational force $g$ along $-\hat{\mathbf y}$, these equations are
\begin{equation}
\partial_t\vec u + (\vec u\!\cdot\!\nabla)\vec u = -\nabla p - g\,\delta\rho\,\hat{\mathbf y},
\quad \nabla\!\cdot\!\vec u = 0,
\quad \partial_t\,\delta\rho + (\vec u\!\cdot\!\nabla)\delta\rho = 0 .
\end{equation}
Adjoining a passively advected scalar to the fluid configuration places the system in the semidirect-product Euler--Poincar\'e setting \cite{Marsden:1984:SPR,Marsden:1997:IMS}: the volume-preserving diffeomorphism group $\mathrm{SDiff}(M)$ acts on the space of advected scalars $\mathcal F$, and the reduced Lagrangian depends on $\delta\rho$ as an advected parameter.
Structure-preserving discretizations of such systems, in which the auxiliary field is advected consistently with the fluid bracket, are treated by Desbrun et al. \cite{Desbrun:2014:VDR}. The buoyancy coupling is generated by the potential energy given as the density perturbation weighted by height along gravity,
\begin{equation}
\mathcal V[\delta\rho] = \int_M g\,\delta\rho\,y \,\Vol ,
\end{equation}
whose variation gives the force $-g\,\delta\rho\,\hat{\mathbf y}$. 

\vspace{0.6em}
\begin{algorithm}
\caption{$\textsc{StepBoussinesq}(\blambda,\bmu,\delta\brho,\bg)^{(n)}$: Implicit Trapezoidal Step with Buoyancy}
\label{alg:step-boussinesq}
\begin{algorithmic}[1]
    \Require Discrete Clebsch variables $(\blambda,\bmu)^{(n)} \in P^m$, Density Perturbation $\delta\brho^{(n)}$,
             Accumulated gravity field $\bg^{(n)}$
    \State $\bs_g \gets -\tfrac12 g\,\by$ \Comment{constant source; $\by$ the fixed height field}
    \State $\bu^{(n)} = \bu^\star \gets \textsc{Project}\big(\textsc{Reconstruct}\big(\left\{(\blambda,\bmu)^{(n)}, (\delta\brho,\bg)^{(n)}\right\}\big)\big)$
    \Repeat \Comment{Fixed point iteration on $\bu^\star$}
        \State $(\blambda,\bmu)^{(n+1)} \gets \textsc{Advect}\big(\bu^\star,(\blambda,\bmu)^{(n)}\big)$
        \State $\delta\brho^{(n+1)} \gets \textsc{Advect}\big(\bu^\star,\delta\brho^{(n)}\big)$
        \State $\bg^\star \gets \textsc{Advect}\big(\bu^\star,\bg^{(n)}\big)$
        \State $\bg^{(n+1)} \gets \bg^\star + \Delta t\,\bs_g$
            \Comment{accumulate $\int g\,y$ along paths}
        \State $\hat\bu \gets \textsc{Reconstruct}\big(\left\{(\blambda,\bmu)^{(n+1)}, (\delta\brho,\bg)^{(n+1)}\right\}\big)$
        \State $\bu^{(n+1)} \gets \textsc{Project}(\hat\bu)$
        \State $\bu^\star \gets \tfrac12\big(\bu^{(n)} + \bu^{(n+1)}\big)$
    \Until{$\bu^\star$ not converged}
    \Ensure $(\blambda,\bmu,\delta\brho,\bg)^{(n+1)}$
\end{algorithmic}
\end{algorithm}
\vspace{0.6em}

This potential is accommodated within the machinery of \secref{sec:algorithm} by introducing one additional advected half-density $\bg$, the \emph{accumulated gravity half-density field}, initialized to zero and advected alongside the Clebsch labels. These equations are
\begin{equation}
\bg^\star = \textsc{Advect}\,\!\big(\bu^\star, \bg^{n}\big), \qquad
\bg^{n+1} = \bg^\star + \Delta t\,\bs_g ,
\end{equation}
with the constant source $\bs_g = -\tfrac12 g\,\by$ accumulated after each advection and where $\by$ is the fixed $y$-direction half-density field.  In this way, $\bg$ accumulates the time integral of the gravitational potential along particle paths. Contracting $(\delta\brho,\bg)$ against the tensor $\bA$ --- with $\delta\brho$ in a $\blambda$-slot and $\bg$ in a $\bmu$-slot --- then yields the buoyancy force on the grid, leaving $\textsc{Advect}$ and $\bA$ unchanged. The resulting update step is given in \Cref{alg:step-boussinesq}, and the primary loop of \algref{alg:sim_loop_w_reset} applies verbatim with $\textsc{Step}$ replaced by $\textsc{StepBoussinesq}$, initializing $\bg \gets \mathbf 0$ and $\delta\brho \gets \delta\brho_0$. At each reset, $\textsc{Reset}$ re-anchors $(\blambda,\bmu)$ and restarts $\bg \gets \mathbf 0$, with $\delta\brho$ carried forward as a material invariant and the accumulated impulse already folded into the reconstructed velocity. \Cref{fig:buoyancy_demo} demonstrates how this procedure recovers the characteristic chaotic, fractal-like structure of the Rayleigh--Taylor instability. Notably, despite the reduced dispersion of the vakonomic method relative to LdA approaches (cf. \cite[Figure 2]{Gawlik:2021:VFD}, some dispersive effects remain, appearing in the vorticity and producing graininess in the advected ink.

\section{Conclusion}
\label{sec:conclusion}

A vakonomic variational approach to discretizing the incompressible Euler equations has been considered and implemented on various benchmark examples.  Using a spatial discretization based on the Koopman representation of diffeomorphisms, the resulting nonholonomic constraint has been treated in a vakonomic way that differs substantially from the standard choice of Lagrange--d'Alembert.  Remarkably, the resulting equations of motion are fully Hamiltonian, exhibiting a discrete Lie--Poisson structure that guarantees an analogue of Kelvin's circulation theorem along with the preservation of Casimir invariants key to physical fidelity.
To facilitate computational efficiency, a momentum map representation based on discrete Clebsch variables has been proposed and used to solve the Lax equation governing the vakonomic dynamics, yielding a simulation algorithm that does not require the manipulation of matrix variables. Ultimately, it was shown that the proposed vakonomic fluids behave stably and consistently even at low grid resolutions, leading to increased robustness and physical realism in the long term.

Many avenues for future work in this area are available and worth pursuing.  On the
theoretical side, it is interesting to consider whether the Lie algebra $\so(F)$ underlying the vakonomic discretization actually converges to $\sdiff(M)$ in the continuum limit, and if so at what rate.  Similarly, it is unclear at present how this notion of convergence compares with the analogous convergence considered by  \cite{Zeitlin:1991:FMA}, along with \cite{Modin:2025:MHD}, in matrix hydrodynamics.
Closely tied to this is the role played by the rank of the discrete system, which appears to shape the dynamics in ways not fully understood; while only $d-1$ Clebsch pairs are needed in theory, there is some variance in practice, and a sharper understanding here would clarify both the fidelity and cost of the method. At a more foundational level, it is natural to ask whether the Clebsch representation proposed here is most appropriate for discretizing these flows, or whether alternative symplectic (or Poisson) spaces might furnish more faithful (or more economical) descriptions.

On the numerical side, the main drawback of the proposed approach is its dispersive nature, controlled here through periodic resetting of the discrete Clebsch fields.  Better methods for dispersion-limiting are possible and certainly desired.  A particularly appealing option involves the inclusion of double-bracket dissipation \cite{Bloch:2026:Arxiv}, which injects diffusion only along coadjoint orbits and without affecting Casimir-preservation.  
Ideally, any alternative to the present approach will  limit dispersion without harming the Lie--Poisson structure of the discrete vakonomic equations or reintroducing other undesirable artifacts. It is also worth mentioning that the half-density parameterization introduced here, and hence the directional derivative tensor $\bA$, is not restricted to representing divergence-free velocity fields.  Therefore, the same strategy can be used to transport general velocity fields, and it is interesting to consider extensions of the present approach beyond incompressible flow.  While the vakonomic equations in this case will still take a Lax form, the Casimirs will no longer coincide with those of the Euler equations, and it is unclear how this will affect the dynamics.  
Finally, the semidirect-product perspective used to incorporate buoyancy is extendable to a broad family of physics. In fact, the same construction should accommodate shallow-water flow, compressible fluids, magnetohydrodynamics, and complex fluids, each recovered by adjoining appropriate advected quantities to the fluid configuration variable.  Regardless, we find the vakonomic approach to fluid simulation particularly exciting.

\section*{Acknowledgements}
The work is supported by NSF CAREER Award 2239062.
A.G. acknowledges support received through the U.S. Department of Energy, Office of Science, Office of Advanced Scientific Computing Research's Applied Mathematics Competitive Portfolios program as well as their Mathematical Multifaceted Integrated Capability Centers (MMICCS) program, under Field Work Proposal 22025291 and the Multifaceted Mathematics for Predictive Digital Twins (M2dt) project.  M. Nabizadeh acknowledges the generous support of the Wistron Corporation, from ARO grant W911NF-26-1-A154, and from the MIT Generative AI Impact Consortium. Additional support was provided by SideFX Software.
Sandia National Laboratories is a multimission laboratory managed and operated by National Technology \& Engineering Solutions of Sandia, LLC, a wholly owned subsidiary of Honeywell International Inc., for the U.S.~Department of Energy's National Nuclear Security Administration under contract DE-NA0003525.
This paper describes objective technical results and analysis. Any subjective views or opinions that might be expressed in the paper do not necessarily represent the views of the U.S.~Department of Energy or the United States Government.

\appendix

\section{Quantum Formulation}

Here, we note how the vakonomic Euler--Arnold system \eqref{eq:Z-vakonomic-final} mirrors a von Neumann quantum mechanical formulation.  

\subsection{Observables and States}
In 
algebraic quantum mechanics, a quantum system is described by a unital $C^*$-algebra (or von Neumann algebra) $\mathcal{A}$, and the collection of all observables is its self-adjoint part, 
\[A = \{O\in\mathcal{A}\mid O^*=O\}. \]
Moreover, the space 
\[\ii A = \{\ii O\mid O\in A\},\]
formed by multiplying all observables in $A$ with the imaginary element $\ii$ consists of skew self-adjoint operators and forms a Lie algebra under the commutator bracket:
\[ [\ii f, \ii g] = -[f,g] = \ii\big(\ii[f,g]\big), \]
where $\ii[f,g]\in A$ is self-adjoint when $f,g\in A$ are self-adjoint.

Similar to how the dual space of a Lie algebra is the state space for the standard Euler--Arnold coadjoint equation (\Cref{sec:background}), the dual space of a \(C^*\)-algebra is dynamically interesting. Note that 
the dual $\mathcal{A}^*$ contains the space of quantum states as a convex subset:
\[S = \{\omega\in\mathcal{A}^*\mid \omega\geq 0,\,\, \omega(\id) = 1\}.\]
In finite dimensions, states may be represented by density operators (matrices) $\omega$, in which case their pairing with observables is
given by the Frobenius inner product
\(\langle \omega\,|f\rangle = \langle\omega,f\rangle = \tr(f\omega), \) 
and the normalization condition becomes $\langle\omega|\,\id\rangle = \tr(\omega) = 1.$
Pure states are those density operators with rank 1, i.e. $\omega = qq^*$ for some unit vector $q$.  Any other states are mixed states, lying in the convex hull of the set of pure states.

A canonical example of this structure is given by classical probability theory, where the space of observables contains measurable functions defined over some sample space.  This forms a commutative von Neumann algebra, where each observable can be considered an ``infinite diagonal matrix'' over the 
points of the sample space.  Each state is a similar object, but carries the additional conditions of entry-wise non-negativity and diagonal values which sum to one.  Clearly, this is just the space of probability measures defined over the sample space.  The dual pairing in this case is given by the expectation of the observable over the probabilistic state.

\subsection{Quantum Hydrodynamics}
In the smooth theory of hydrodynamics, the 
space of ``quantum observables'' is given by the reverse of the construction from before:
\[-\ii\sdiff(M) = \{-\ii\LD_X\,|\, X\in\sdiff(M)\}.\]
Each \(-\ii\LD_X\) is a Hermitian operator on $L^2(M)$ strongly analogous to the usual momentum operator \(\hat p = -\ii\hbar {d\over dx}\).
This ``velocity'' observable generates advection: given a scalar function \(\psi\in C^\infty(M;\mathbb{C})\), the Schr\"odinger-like evolution 
\[\ii{\partial\over\partial t}\psi = -\ii\LD_X\psi,\]
simply advects $\psi$ along the vector field $X\in\sdiff(M)$.
Note that the corresponding space of dynamical states is (up to multiplication by \(\ii\)) the usual dual space \(\sdiff(M)^*\) containing the coadjoint orbits of the Euler--Arnold equation. 

It is interesting to observe the contrast between this quantum system and the classical probability theory.  The observable space $-\ii\sdiff(M)$ is not a commutative algebra, and the associated density operators give rise to linear functionals on $-\ii\sdiff(M)$ evolving by the coadjoint equation, von Neumann equation, Heisenberg equation, Lax equation, etc.  To describe this more precisely, recall that $\sdiff(M)^* \cong \Omega^1(M) \mathbin{/}d\Omega^0(M)$ is the space of one-forms modulo exact differentials (\Cref{sec:EquationsOfMotion}).
A pure state, represented by the wavefunction \(\psi\in C^\infty(M;\CC)\), gives rise to a density matrix as the ``outer product'' of the vector \(\psi\), i.e., $\rho_{\psi} = \psi\psi^*$, considered as an operator on $L^2(M;\mathbb{C})$ with the inner product $\langle\psi, \phi\rangle = \int_M \bar\psi\phi\,dV$.  
Under the Schr\"{o}dinger-like evolution above, the density evolves by the Lax equation 
\[\ii\frac{\partial}{\partial t}{\rho}_{\psi} = -[\ii\LD_{X},\rho_{\psi}].\]
Indeed, for any test function $\phi\in C^{\infty}(M;\CC)$ it follows from integration by parts that 
\begin{align*}
    \dot{\rho}_{\psi}\phi &= \dot\psi\langle\psi,\phi\rangle + \psi\langle\dot\psi,\phi\rangle = (-\LD_X\psi)\langle\psi,\phi\rangle-\psi\langle\LD_X\psi,\phi\rangle \\
    &= -\LD_{X}(\rho_{\psi}\phi) + \psi\langle\psi,\LD_X\phi\rangle = [\rho_{\psi},\LD_X]\phi,
\end{align*}
and therefore the claim follows from multiplication by $\ii$.
Moreover, the dual pairing between a density operator $\rho_{\psi}$ and an observable $-\ii\LD_{X}$ is given as before:
\[\langle\rho_{\psi}\,|\,-\ii\LD_{X}\rangle = \int_M \bar\psi(-\ii\LD_X)\psi\,dV.\]

Of course, there is also the usual dual pairing between an arbitrary ``state'' $[\eta]\in\sdiff(M)^*$ and an ``observable'' $X\in\sdiff(M)$, \[\langle [\eta]\, |\, {X}\rangle = \int_M \eta(X)\, dV.\]
It is instructive to compute the element of $\sdiff(M)^*$ associated to the pure state density operator $\rho_{\psi}$. 
Observing that $\LD_X\psi = d\psi(X)$ is the ordinary directional derivative, it follows that
\[\langle\rho_{\psi}\,|\,-\ii\LD_X \rangle = \langle [-\ii\bar\psi d\psi]\,|\, X\rangle,\]
where the pairings on the left and right are defined as above. 
Now, recall that the differential of the modulus satisfies
\[d|\psi|^2 = d\bar{\psi}\psi + \bar\psi d\psi = 2\Re(\bar\psi d\psi) = 2\Im(\ii\bar\psi d\psi).\]
Therefore, the expression whose equivalence class is desired becomes 
\[-\ii\bar\psi d\psi = \Re(-\ii\bar\psi d\psi) + \Im(-\ii\bar\psi d\psi) = \Re(-\ii\bar\psi d\psi) - \tfrac{1}{2}d|\psi|^2, \]
showing that its representative in $\sdiff(M)^*$ is given by 
\[[\rho_{\psi}] = \Re(-\ii\bar\psi d\psi).\]
Remarkably, this is just the Clebsch representation in classical hydrodynamics (\Cref{sec:LowRankParameterization}). Choosing $\psi = \lambda + \ii\mu$, it follows that 
\[-\ii\bar\psi d\psi = -\ii(\lambda-\ii\mu) d(\lambda+\ii\mu) = (\lambda d\mu - \mu d\lambda) - \ii (\lambda d\lambda + \mu d\mu),\]
and so $[\rho_{\psi}] = \lambda d\mu - \mu d\lambda$.
Therefore, the Euler equations can be considered as a ``nonlinear Schr{\"o}dinger'' system 
where the Schr{\"o}dinger operator $-\ii\LD_X$ depends on the density matrix $X = \sum_i [\rho_{\psi^i}] = \sum_i\lambda_id\mu_i-\mu_id\lambda_i$ represented by Clebsch variables.

\section{Matrix Hydrodynamics}\label{app:MHD}
In contrast to the present approach that discretizes a space of velocities to generate directional derivatives, matrix hydrodynamics discretizes the Lie algebra $\sdiff(M)$ directly.  The advantage of this is that no constraint on admissible velocities is necessary: the approximation space is automatically closed under the Lie bracket.  In particular, each element \(\bv\in V_{\div}\) can be represented by a matrix \(\bL_\bv\), and the image \(\bL(V_{\div})\) is closed under matrix commutation.  This means there is no distinction between vakonomic and Lagrange--d'Alembert (since there is no nonholonomic constraint), and the equation of motion becomes a simple Lax equation for the involved matrices.

\subsection{Continuous System}
To carry out this program in 2D, it is natural to design the Lie algebra structure on \(V_{\div}\) through a frequency-space approximation of the Lie algebra structure of \(\sdiff(M)\).  Then, $V_{\div}$ can be parameterized by streamfunctions, and the 2D Euler equation can be written as 
\(\dot w = \{\Delta^{-1}w,w\}\), where \(w\in C^\infty(M)\) is the vorticity function, \(\Delta\) is the Laplacian, the function \(p = \Delta^{-1}w\) is  the streamfunction, and \(\{\cdot,\cdot\}\) is the Poisson bracket induced by the symplectic structure (area form) on the 2D domain \(M\).  Explicitly, \(\{f,\cdot\} = \adv_{J\nabla f} = \nabla f \cdot J\nabla \) where \(J\) denotes clockwise 
\(90^\circ\) rotation and \(\nabla f\) is the gradient of the smooth function \(f\in C^{\infty}(M)\).

The Poisson structure \(\{\cdot,\cdot\}\) on smooth functions \(C^\infty(M)\) describes a Lie algebra that can be conveniently expressed in terms of the frequency components of its elements.  
For this purpose, \(M\) in matrix hydrodynamics is usually taken as a flat torus or the sphere, on which Fourier or spherical harmonic expansions are readily available.  Considering the flat torus \(M = \TT^2=\RR^2/\ZZ^2\), for example, each function \(f\colon\TT^2\to\RR\) is expressed by its Fourier series \[f(x,y) = \sum_{(j_1,j_2)\in\ZZ^2}\hat f_{j_1,j_2}e^{2\pi\ii (j_1x + j_2y)},\] where the Fourier coefficients \(\hat f_{\bj}\in\CC\) satisfy the reality condition \(\hat f_{\bj} = \bar{\hat{f}}_{-\bj}\). 
The normalized advection operator 
\(\cL_f = (-4\pi^2)^{-1}\{f,\cdot\}\)
can then be expressed in terms of these Fourier coefficients as 
\[\cL_{f} = \sum_{\bj\in\ZZ^2}\hat f_{\bj}\cL_{\bj}, \quad \cL_{\bj}= \frac{1}{2\pi\ii} e^{2\pi\ii\bj\cdot\bx}\bj\cdot J\nabla.\] 
Observe that the operators \(\cL_{\bj}\) for $\bj\neq 0$ form a basis for \(\sdiff(\TT^2)\) (modulo the constant harmonic fields) with the structural relation 
\begin{align}
\label{eq:FrequencyStructuralEquation}
    [\cL_{\bj},\cL_{\bk}] = (\bj\times\bk)\cL_{\bj+\bk},
\end{align}
where \(\bj\times\bk = j_1k_2 - j_2k_1\).  

In this representation, the vorticity equation \(\dot w = \{\Delta^{-1}w,w\}\) can be expressed as an evolution of the advection operator $\cL_w\in\sdiff(\TT^2)$:
\begin{align}
\label{eq:VorticityEquationOperatorForm}
    \dot\cL_{w} = -4\pi^2[\cL_{p},\cL_{w}],\quad \hat p_{\bj} = -{1\over |2\pi\bj|^2}\hat w_{\bj},
\end{align}
where $\hat{p}_{\bj}$ are the coefficients of the inverse Laplacian in frequency space.  

\subsection{Discrete System}
In 1991, Zeitlin employed a finite-dimensional analogue of \eqref{eq:FrequencyStructuralEquation} and thereby introduced a finite-dimensional analogue of the vorticity equation \eqref{eq:VorticityEquationOperatorForm}.
First, the frequency space \(\ZZ^2\) is truncated regularly into \(\ZZ_N^2\), where \(\ZZ_N = \ZZ/(N\ZZ) = \{0,1,\ldots, N-1\}\). 
The space of real stream functions with these frequencies, given by 
\begin{align*}
&\bigg\{f(\bx) = {1\over N^2}\sum_{\bj\in\ZZ_N^2}\hat f_\bj\, e^{\frac{2\pi \ii\bj\cdot\bx}{N}}\bigg\} \\
&\qquad \cong\left\{ \hat f\colon \ZZ_N^2\to\CC\,\middle|\,\hat f_\bj = \bar{\hat{f}}_{-\bj}\right\},
\end{align*}
then forms the \(V_{\div}\) space (modulo the two-dimensional constant harmonic vector fields), where the \(1/N^2\) factor follows the discrete Fourier transform convention.
Remarkably, it turns out that
there are \(N\times N\) complex matrices \(\bL_{\bj}\) for each frequency \(\bj\in\ZZ_N^2\) obeying the reality relation \(\bar\bL_{\bj}^\intercal = \bL_{-\bj}\) as well as the \emph{sine bracket relation}
\begin{align}
\label{eq:SineBracket}
    [\bL_{\bj},\bL_{\bk}] = \tfrac{N}{2\pi}\sin\left(\tfrac{2\pi}{N}(\bj\times\bk)\right)\bL_{\bj+\bk},
\end{align}
where expressions such as \(\bj+\bk\) are understood modulo \(N\) and the relation \eqref{eq:SineBracket} approximates \eqref{eq:FrequencyStructuralEquation} as \(N\to\infty\).  Provided these matrices can be explicitly constructed, this provides a Lie algebra homomorphism $f\mapsto\bL_{f}$ between $V_{\div}$ and the space $\su(N)$ given by
\[f(\bx) = {1\over N^2}\sum_{\bj\in\ZZ_N^2}\hat f_\bj\, e^{\frac{2\pi \ii\bj\cdot\bx}{N}}\mapsto {1\over N^2}\sum_{\bj\in\ZZ_N^2}\hat f_{\bj}\,\bL_\bj\eqqcolon\bL_{f}.\]

Letting \(\rho = e^{4\pi\ii/N}\) be the squared \(N\)-th root of unity, an explicit construction for \(\bL_\bj\) that satisfies \eqref{eq:SineBracket} proceeds as follows.  Find any two matrices \(\bS,\bT\in\CC^{N\times N}\) satisfying the commuting relation \(\bT\bS = \rho\bS\bT\).  Then 
\[\bL_\bj = \ii{N\over 2\pi} \rho^{j_1j_2\over 2}\bS^{j_1}\bT^{j_2},\]
satisfies the sine bracket relation \eqref{eq:SineBracket}.  
Writing \(\bW = \bL_w\) and \(\bP = \bL_p\), this enables a statement of Zeitlin's discrete Euler equation that forms the basis for matrix hydrodynamics:
\begin{align}
\label{eq:ZeitlinEquation}
    \dot\bW = -[\bP,\bW],\quad \hat p_{\bj} = -{1\over |2\pi\bj|^2}\hat w_{\bj}.
\end{align}
In practice, it is more convenient to express this in terms of the spectral coefficient \(\hat w\),
\begin{align}
    \dot{\hat w}_{\bn} = \sum_{\bj\in\ZZ_N^2}{1\over 2\pi}\sin\left({2\pi\over N}(\bj\times(\bn - \bj))\right){1\over |2\pi\bj|^2}\hat w_{\bj}\hat w_{\bn - \bj},
\end{align}
yielding a solvable system that does not require explicit construction of the operators $\bL_{\bj}$.

\begin{remark}
    An useful example of matrices $\bS,\bT$ which parameterize the operators $\bL_{\bj}$ is given by the following. Let \(\bS = \operatorname{diag}(1,\rho,\rho^2,\ldots,\rho^{N-1})\), and \((\bT\bu)_i = u_{i+1}\).
    Using this explicit construction, \(\bL_{f}\) is given by \[(\bL_f)_{i,j} = {\ii\over 2\pi}(\tilde f)_{i+j,j-i}, \quad \tilde f_{m,n} = \sum_{\ell\in\ZZ_N}f_{m,\ell}e^{-\frac{2\pi \ii n\ell}{N}}.\] 
    Notice that $\tilde f$ is the discrete Fourier transform of \(f\) applied only in the \(y\) direction.
\end{remark}

\bibliographystyle{amsplain-nodash}
\bibliography{vakonomic}

@article{Abanov:2025:IDN,
	title = {Infinite-{Dimensional} {Nonholonomic} and {Vakonomic} {Systems}},
	volume = {36},
	issn = {1432-1467},
	url = {https://doi.org/10.1007/s00332-026-10286-4},
	doi = {10.1007/s00332-026-10286-4},
	number = {3},
	journal = {Journal of Nonlinear Science},
	author = {Abanov, Alexander G. and Khesin, Boris},
	month = jun,
	year = {2026},
	pages = {67},
}

@article{AlMohy:2011:Exp,
	title = {Computing the {Action} of the {Matrix} {Exponential}, with an {Application} to {Exponential} {Integrators}},
	volume = {33},
	url = {https://doi.org/10.1137/100788860},
	doi = {10.1137/100788860},
	number = {2},
	journal = {SIAM Journal on Scientific Computing},
	author = {Al-Mohy, Awad H. and Higham, Nicholas J.},
	year = {2011},
	pages = {488--511}
}

@book{Arnold:2018:FEEC,
	address = {Philadelphia, PA},
	title = {Finite {Element} {Exterior} {Calculus}},
	url = {https://epubs.siam.org/doi/abs/10.1137/1.9781611975543},
	doi = {10.1137/1.9781611975543},
	publisher = {Society for Industrial and Applied Mathematics},
	author = {Arnold, Douglas N.},
	year = {2018}
}

@article{Arnold:2006:FEEC,
	title = {Finite element exterior calculus, homological techniques, and applications},
	volume = {15},
	doi = {10.1017/S0962492906210018},
	journal = {Acta Numerica},
	author = {Arnold, Douglas N. and Falk, Richard S. and Winther, Ragnar},
	year = {2006},
	pages = {1--155},
}

@article{Arnold:1966:GDG,
	title = {Sur la géométrie différentielle des groupes de {Lie} de dimension infinie et ses applications à l'hydrodynamique des fluides parfaits},
	volume = {16},
	url = {http://eudml.org/doc/73896},
	language = {fre},
	number = {1},
	journal = {Annales de l'institut Fourier},
	publisher = {Association des Annales de l'Institut Fourier},
	author = {Arnold, Vladimir},
	year = {1966},
	pages = {319--361},
}

@article{Azencot:2013:OAT,
	title = {An {Operator} {Approach} to {Tangent} {Vector} {Field} {Processing}},
	volume = {32},
	url = {https://onlinelibrary.wiley.com/doi/abs/10.1111/cgf.12174},
	doi = {https://doi.org/10.1111/cgf.12174},
	number = {5},
	journal = {Computer Graphics Forum},
	author = {Azencot, Omri and Ben-Chen, Mirela and Chazal, Frédéric and Ovsjanikov, Maks},
	year = {2013},
	pages = {73--82},
}

@article{Azencot:2015:DDV,
	address = {New York, NY, USA},
	title = {Discrete {Derivatives} of {Vector} {Fields} on {Surfaces} -- {An} {Operator} {Approach}},
	volume = {34},
	issn = {0730-0301},
	url = {https://doi.org/10.1145/2723158},
	doi = {10.1145/2723158},
	number = {3},
	journal = {ACM Trans. Graph.},
	publisher = {Association for Computing Machinery},
	author = {Azencot, Omri and Ovsjanikov, Maks and Chazal, Frédéric and Ben-Chen, Mirela},
	month = may,
	year = {2015},
}

@article{Azencot:2014:FFS,
	title = {Functional {Fluids} on {Surfaces}},
	volume = {33},
	url = {https://onlinelibrary.wiley.com/doi/abs/10.1111/cgf.12449},
	doi = {https://doi.org/10.1111/cgf.12449},
	number = {5},
	journal = {Computer Graphics Forum},
	author = {Azencot, Omri and Weißmann, Steffen and Ovsjanikov, Maks and Wardetzky, Max and Ben-Chen, Mirela},
	year = {2014},
	pages = {237--246},
}

@book{Batchelor:2000:IFD,
	series = {Cambridge {Mathematical} {Library}},
	title = {An {Introduction} to {Fluid} {Dynamics}},
	publisher = {Cambridge University Press},
	author = {Batchelor, G. K.},
	year = {2000},
}

@article{Bauer:2017:TGV,
	title = {Towards a geometric variational discretization of compressible fluids: {The} rotating shallow water equations},
	volume = {6},
	issn = {2158-2491},
	url = {https://www.aimsciences.org/article/doi/10.3934/jcd.2019001},
	doi = {10.3934/jcd.2019001},
	number = {1},
	journal = {Journal of Computational Dynamics},
	author = {Bauer, Werner and Gay-Balmaz, François},
	month = jun,
	year = {2019},
	pages = {1--37},
}

@article{Bauer:2017:VIA,
	title = {Variational integrators for anelastic and pseudo-incompressible flows},
	volume = {11},
	issn = {1941-4889},
	url = {https://www.aimsciences.org/article/doi/10.3934/jgm.2019025},
	doi = {10.3934/jgm.2019025},
	number = {4},
	journal = {Journal of Geometric Mechanics},
	author = {Bauer, Werner and Gay-Balmaz, François},
	month = nov,
	year = {2019},
	pages = {511--537},
}

@book{Bloch:2015:NMC,
  title={Nonholonomic Mechanics and Control},
  author={Bloch, Anthony M and Baillieul, John and Crouch, Peter E and Marsden, Jerrold E and Zenkov, Dmitry},
  series={Interdisciplinary Applied Mathematics},
  year={2015},
  address   = {New York},
  publisher={Springer New York}
}

@article{Brecht:2019:VIR,
	title = {Variational integrator for the rotating shallow-water equations on the sphere},
	volume = {145},
	url = {https://rmets.onlinelibrary.wiley.com/doi/abs/10.1002/qj.3477},
	doi = {https://doi.org/10.1002/qj.3477},
	number = {720},
	journal = {Quarterly Journal of the Royal Meteorological Society},
	author = {Brecht, Rüdiger and Bauer, Werner and Bihlo, Alexander and Gay-Balmaz, François and MacLachlan, Scott},
	year = {2019},
	pages = {1070--1088},
}

@article{Brunton:2022:MKT,
	title = {Modern {Koopman} {Theory} for {Dynamical} {Systems}},
	volume = {64},
	url = {https://doi.org/10.1137/21M1401243},
	doi = {10.1137/21M1401243},
	number = {2},
	journal = {SIAM Review},
	author = {Brunton, Steven L. and Budišić, Marko and Kaiser, Eurika and Kutz, J. Nathan},
	year = {2022},
	pages = {229--340},
}

@article{Buffa:2011:IGD,
	title = {Isogeometric {Discrete} {Differential} {Forms} in {Three} {Dimensions}},
	volume = {49},
	url = {https://doi.org/10.1137/100786708},
	doi = {10.1137/100786708},
	number = {2},
	journal = {SIAM Journal on Numerical Analysis},
	author = {Buffa, A. and Rivas, J. and Sangalli, G. and Vázquez, R.},
	year = {2011},
	pages = {818--844},
}

@article{Buffa:2010:IGA,
	title = {Isogeometric analysis in electromagnetics: {B}-splines approximation},
	volume = {199},
	issn = {0045-7825},
	url = {https://www.sciencedirect.com/science/article/pii/S0045782509004010},
	doi = {https://doi.org/10.1016/j.cma.2009.12.002},
	number = {17},
	journal = {Computer Methods in Applied Mechanics and Engineering},
	author = {Buffa, A. and Sangalli, G. and Vázquez, R.},
	year = {2010},
	pages = {1143--1152},
}

@article{Calvo:1997:NSI,
	title = {Numerical {Solution} of {Isospectral} {Flows}},
	volume = {66},
	issn = {00255718, 10886842},
	url = {http://www.jstor.org/stable/2153680},
	number = {220},
	urldate = {2026-07-28},
	journal = {Mathematics of Computation},
	publisher = {American Mathematical Society},
	author = {Calvo, Mari Paz and Iserles, Arieh and Zanna, Antonella},
	year = {1997},
	pages = {1461--1486},
}

@article{Calvo:1996:RKO,
	title = {Runge-{Kutta} methods for orthogonal and isospectral flows},
	volume = {22},
	issn = {0168-9274},
	url = {https://www.sciencedirect.com/science/article/pii/S0168927496000293},
	doi = {10.1016/S0168-9274(96)00029-3},
	number = {1},
	journal = {Special Issue Celebrating the Centenary of Runge-Kutta Methods},
	author = {Calvo, M.P. and Iserles, A. and Zanna, A.},
	month = nov,
	year = {1996},
	pages = {153--163},
}

@article{Cantrijn:1999:APS,
	title = {On almost-{Poisson} structures in nonholonomic mechanics},
	volume = {12},
	issn = {0951-7715},
	url = {https://doi.org/10.1088/0951-7715/12/3/316},
	doi = {10.1088/0951-7715/12/3/316},
	number = {3},
	journal = {Nonlinearity},
	author = {{F Cantrijn} and {M de León} and {D Martín de Diego}},
	month = may,
	year = {1999},
	pages = {721},
}

@article{Celledoni:2014:LGI,
	title = {An introduction to {Lie} group integrators -- basics, new developments and applications},
	volume = {257},
	issn = {0021-9991},
	url = {https://www.sciencedirect.com/science/article/pii/S0021999113000041},
	doi = {10.1016/j.jcp.2012.12.031},
	journal = {Physics-compatible numerical methods},
	author = {Celledoni, Elena and Marthinsen, Håkon and Owren, Brynjulf},
	month = jan,
	year = {2014},
	pages = {1040--1061},
}

@article{Cifani:2023:EGM,
	title = {An efficient geometric method for incompressible hydrodynamics on the sphere},
	volume = {473},
	issn = {0021-9991},
	url = {https://www.sciencedirect.com/science/article/pii/S002199912200835X},
	doi = {10.1016/j.jcp.2022.111772},
	journal = {Journal of Computational Physics},
	author = {Cifani, P. and Viviani, M. and Modin, K.},
	month = jan,
	year = {2023},
	pages = {111772},
}

@Article{Clebsch:1859:IHG,
  author = 	 {A. Clebsch},
  title = 	 {Ueber die Integration der hydrodynamischen Gleichungen},
  journal = 	{Journal f\"ur die reine und angewandte Mathematik},
  year = 	 {1859},
  volume =	 {56},
  OPTnumber = 	 {},
  pages =	 {1--10},
  note={English translation by {D.\ H.\ Delphenich}, \url{http://www.neo-classical-physics.info/uploads/3/4/3/6/34363841/clebsch_-_clebsch_variables.pdf}
},
}

@inproceedings{Crane:2013:DEC,
	address = {New York, NY, USA},
	series = {{SIGGRAPH} '13},
	title = {Digital geometry processing with discrete exterior calculus},
	isbn = {978-1-4503-2339-0},
	url = {https://doi.org/10.1145/2504435.2504442},
	doi = {10.1145/2504435.2504442},
	booktitle = {{ACM} {SIGGRAPH} 2013 {Courses}},
	publisher = {Association for Computing Machinery},
	author = {Crane, Keenan and de Goes, Fernando and Desbrun, Mathieu and Schröder, Peter},
	year = {2013},
}

@article{Elcott:2007:SCP,
	address = {New York, NY, USA},
	title = {Stable, circulation-preserving, simplicial fluids},
	volume = {26},
	issn = {0730-0301},
	url = {https://doi.org/10.1145/1189762.1189766},
	doi = {10.1145/1189762.1189766},
	number = {1},
	journal = {ACM Trans. Graph.},
	publisher = {Association for Computing Machinery},
	author = {Elcott, Sharif and Tong, Yiying and Kanso, Eva and Schröder, Peter and Desbrun, Mathieu},
	month = jan,
	year = {2007},
	pages = {4--es},
}

@article{Dupont:2003:BFECC,
	title = {Back and forth error compensation and correction methods for removing errors induced by uneven gradients of the level set function},
	volume = {190},
	issn = {0021-9991},
	url = {https://www.sciencedirect.com/science/article/pii/S0021999103002766},
	doi = {10.1016/S0021-9991(03)00276-6},
	number = {1},
	journal = {Journal of Computational Physics},
	author = {Dupont, Todd F. and Liu, Yingjie},
	month = sep,
	year = {2003},
	pages = {311--324},
}

@article{Desbrun:2014:VDR,
	title = {Variational discretization for rotating stratified fluids},
	volume = {34},
	issn = {1078-0947},
	url = {https://www.aimsciences.org/article/doi/10.3934/dcds.2014.34.477},
	doi = {10.3934/dcds.2014.34.477},
	number = {2},
	journal = {Discrete and Continuous Dynamical Systems},
	author = {Desbrun, Mathieu and Gawlik, Evan S. and Gay-Balmaz, François and Zeitlin, Vladimir},
	month = aug,
	year = {2013},
	pages = {477--509},
}

@article{Elman:1994:IPU,
	title = {Inexact and {Preconditioned} {Uzawa} {Algorithms} for {Saddle} {Point} {Problems}},
	volume = {31},
	url = {https://doi.org/10.1137/0731085},
	doi = {10.1137/0731085},
	number = {6},
	journal = {SIAM Journal on Numerical Analysis},
	author = {Elman, Howard C. and Golub, Gene H.},
	year = {1994},
	pages = {1645--1661},
}

@article{Engo:2000:CGI,
	title = {On the {Construction} of {Geometric} {Integrators} in the {RKMK} {Class}},
	volume = {40},
	issn = {1572-9125},
	url = {https://doi.org/10.1023/A:1022362117414},
	doi = {10.1023/A:1022362117414},
	number = {1},
	journal = {BIT Numerical Mathematics},
	author = {Engø, Kenth},
	month = mar,
	year = {2000},
	pages = {41--61},
}

@article{Engo:2001:NILP,
	title = {Numerical {Integration} of {Lie}--{Poisson} {Systems} {While} {Preserving} {Coadjoint} {Orbits} and {Energy}},
	volume = {39},
	url = {https://doi.org/10.1137/S0036142999364212},
	doi = {10.1137/S0036142999364212},
	number = {1},
	journal = {SIAM Journal on Numerical Analysis},
	author = {Engø, Kenth and Faltinsen, Stig},
	year = {2001},
	pages = {128--145},
}

@article{Favretti:1998:EDN,
	title = {Equivalence of {Dynamics} for {Nonholonomic} {Systems} with {Transverse} {Constraints}},
	volume = {10},
	issn = {1572-9222},
	url = {https://doi.org/10.1023/A:1022667307485},
	doi = {10.1023/A:1022667307485},
	number = {4},
	journal = {Journal of Dynamics and Differential Equations},
	author = {Favretti, Marco},
	month = oct,
	year = {1998},
	pages = {511--536},
}

@article{Fernandez:2008:EDN,
	title = {Equivalence of the dynamics of nonholonomic and variational nonholonomic systems for certain initial data},
	volume = {41},
	issn = {1751-8121},
	url = {https://doi.org/10.1088/1751-8113/41/34/344005},
	doi = {10.1088/1751-8113/41/34/344005},
	number = {34},
	journal = {Journal of Physics A: Mathematical and Theoretical},
	author = {Fernandez, Oscar E and Bloch, Anthony M},
	month = aug,
	year = {2008},
	pages = {344005},
}

@article{Gawlik:2020:CFE,
	title = {A conservative finite element method for the incompressible {Euler} equations with variable density},
	volume = {412},
	issn = {0021-9991},
	url = {https://www.sciencedirect.com/science/article/pii/S0021999120302138},
	doi = {10.1016/j.jcp.2020.109439},
	journal = {Journal of Computational Physics},
	author = {Gawlik, Evan S. and Gay-Balmaz, François},
	month = jul,
	year = {2020},
	pages = {109439},
}

@article{Gawlik:2021:SPE,
	title = {A structure-preserving finite element method for compressible ideal and resistive magnetohydrodynamics},
	volume = {87},
	doi = {10.1017/S0022377821000842},
	number = {5},
	journal = {Journal of Plasma Physics},
	author = {Gawlik, Evan S. and Gay-Balmaz, François},
	year = {2021},
	pages = {835870501},
}

@article{Gawlik:2021:VFD,
	title = {A {Variational} {Finite} {Element} {Discretization} of {Compressible} {Flow}},
	volume = {21},
	issn = {1615-3383},
	url = {https://doi.org/10.1007/s10208-020-09473-w},
	doi = {10.1007/s10208-020-09473-w},
	number = {4},
	journal = {Foundations of Computational Mathematics},
	author = {Gawlik, Evan S. and Gay-Balmaz, François},
	month = aug,
	year = {2021},
	pages = {961--1001},
}

@article{Gawlik:2022:FEM,
	title = {A finite element method for {MHD} that preserves energy, cross-helicity, magnetic helicity, incompressibility, and div {B} = 0},
	volume = {450},
	issn = {0021-9991},
	url = {https://www.sciencedirect.com/science/article/pii/S0021999121007427},
	doi = {10.1016/j.jcp.2021.110847},
	journal = {Journal of Computational Physics},
	author = {Gawlik, Evan S. and Gay-Balmaz, François},
	month = feb,
	year = {2022},
	pages = {110847},
}

@article{Gawlik:2024:VTC,
	title = {Variational and thermodynamically consistent finite element discretization for heat conducting viscous fluids},
	volume = {34},
	issn = {0218-2025},
	url = {https://doi.org/10.1142/S0218202524500027},
	doi = {10.1142/S0218202524500027},
	number = {02},
	urldate = {2026-07-28},
	journal = {Mathematical Models and Methods in Applied Sciences},
	publisher = {World Scientific Publishing Co.},
	author = {Gawlik, Evan S. and Gay-Balmaz, François},
	month = feb,
	year = {2024},
	pages = {243--284},
}

@article{Gawlik:2025:SPT,
	title = {Structure-preserving and thermodynamically consistent finite element discretization for visco-resistive {MHD} with thermoelectric effect},
	volume = {542},
	issn = {0021-9991},
	url = {https://www.sciencedirect.com/science/article/pii/S0021999125006187},
	doi = {10.1016/j.jcp.2025.114336},
	journal = {Journal of Computational Physics},
	author = {Gawlik, Evan S. and Gay-Balmaz, François and Manach-Pérennou, Bastien},
	month = dec,
	year = {2025},
	pages = {114336},
}

@article{Gawlik:2011:GVD,
	title = {Geometric, variational discretization of continuum theories},
	volume = {240},
	issn = {0167-2789},
	url = {https://www.sciencedirect.com/science/article/pii/S0167278911001989},
	doi = {10.1016/j.physd.2011.07.011},
	number = {21},
	journal = {Physica D: Nonlinear Phenomena},
	author = {Gawlik, E.S. and Mullen, P. and Pavlov, D. and Marsden, J.E. and Desbrun, M.},
	month = oct,
	year = {2011},
	pages = {1724--1760},
}

@phdthesis{Hirani:2003:DEC,
	address = {Pasadena, CA},
	type = {{PhD} {Thesis}},
	title = {Discrete {Exterior} {Calculus}},
	url = {https://thesis.caltech.edu/1885/},
	school = {California Institute of Technology},
	author = {Hirani, Anil Nirmal},
	year = {2003},
}

@article{Hughes:2005:IGA,
	title = {Isogeometric analysis: {CAD}, finite elements, {NURBS}, exact geometry and mesh refinement},
	volume = {194},
	issn = {0045-7825},
	url = {https://www.sciencedirect.com/science/article/pii/S0045782504005171},
	doi = {10.1016/j.cma.2004.10.008},
	number = {39},
	journal = {Computer Methods in Applied Mechanics and Engineering},
	author = {Hughes, T.J.R. and Cottrell, J.A. and Bazilevs, Y.},
	month = oct,
	year = {2005},
	pages = {4135--4195},
}

@book{huke1931historical,
	address = {Chicago, Illinois, USA},
	title = {An {Historical} and {Critical} {Study} of the {Fundamental} {Lemma} in the {Calculus} of {Variations}},
	publisher = {The University of Chicago Press},
	author = {Huke, Aline},
	year = {1931},
}

@article{Kleckner:2013:CDK,
	title = {Creation and dynamics of knotted vortices},
	volume = {9},
	issn = {1745-2481},
	url = {https://doi.org/10.1038/nphys2560},
	doi = {10.1038/nphys2560},
	number = {4},
	journal = {Nature Physics},
	author = {Kleckner, Dustin and Irvine, William T. M.},
	month = apr,
	year = {2013},
	pages = {253--258},
}

@article{Koom:1997:HLA,
	title = {The {Hamiltonian} and {Lagrangian} approaches to the dynamics of nonholonomic systems},
	volume = {40},
	issn = {0034-4877},
	url = {https://www.sciencedirect.com/science/article/pii/S0034487797856170},
	doi = {10.1016/S0034-4877(97)85617-0},
	number = {1},
	journal = {Reports on Mathematical Physics},
	author = {Koon, Wang Sang and Marsden, Jerrold E.},
	month = aug,
	year = {1997},
	pages = {21--62},
}

@article{Koopman:1931:HST,
	title = {Hamiltonian {Systems} and {Transformation} in {Hilbert} {Space}},
	volume = {17},
	url = {https://doi.org/10.1073/pnas.17.5.315},
	doi = {10.1073/pnas.17.5.315},
	number = {5},
	urldate = {2026-07-28},
	journal = {Proceedings of the National Academy of Sciences},
	publisher = {Proceedings of the National Academy of Sciences},
	author = {Koopman, B. O.},
	month = may,
	year = {1931},
	pages = {315--318},
}

@article{Iserles:2000:LGM,
	title = {Lie-group methods},
	volume = {9},
	issn = {0962-4929},
	url = {https://www.cambridge.org/core/product/856125FF1EAF7762DEF6E37EEBA9CA5F},
	doi = {10.1017/S0962492900002154},
	journal = {Acta Numerica},
	publisher = {Cambridge University Press},
	author = {Iserles, Arieh and Munthe-Kaas, Hans Z. and Nørsett, Syvert P. and Zanna, Antonella},
	year = {2000},
	note = {Edition: 2001/03/21},
	pages = {215--365},
}

@article{Kuiper:1965:HTU,
	title = {The homotopy type of the unitary group of {Hilbert} space},
	volume = {3},
	issn = {0040-9383},
	url = {https://www.sciencedirect.com/science/article/pii/0040938365900674},
	doi = {10.1016/0040-9383(65)90067-4},
	number = {1},
	journal = {Topology},
	author = {Kuiper, Nicolaas H.},
	month = jan,
	year = {1965},
	pages = {19--30},
}

@book{Lamb:1895:HD,
	address = {Cambridge, UK},
	title = {Hydrodynamics},
	publisher = {Cambridge University Press},
	author = {Lamb, Horace},
	year = {1895},
}

@article{Lax:1968:INE,
	title = {Integrals of nonlinear equations of evolution and solitary waves},
	volume = {21},
	issn = {0010-3640},
	url = {https://doi.org/10.1002/cpa.3160210503},
	doi = {10.1002/cpa.3160210503},
	number = {5},
	urldate = {2026-07-28},
	journal = {Communications on Pure and Applied Mathematics},
	publisher = {John Wiley \& Sons, Ltd},
	author = {Lax, Peter D.},
	month = sep,
	year = {1968},
	pages = {467--490},
}

@article{Leray:1934:SML,
	title = {Sur le mouvement d'un liquide visqueux emplissant l'espace},
	volume = {63},
	issn = {1871-2509},
	url = {https://doi.org/10.1007/BF02547354},
	doi = {10.1007/BF02547354},
	number = {1},
	journal = {Acta Mathematica},
	author = {Leray, Jean},
	month = dec,
	year = {1934},
	pages = {193--248},
}

@article{Liu:2015:MVFS,
	address = {New York, NY, USA},
	title = {Model-reduced variational fluid simulation},
	volume = {34},
	issn = {0730-0301},
	url = {https://doi.org/10.1145/2816795.2818130},
	doi = {10.1145/2816795.2818130},
	number = {6},
	journal = {ACM Trans. Graph.},
	publisher = {Association for Computing Machinery},
	author = {Liu, Beibei and Mason, Gemma and Hodgson, Julian and Tong, Yiying and Desbrun, Mathieu},
	month = nov,
	year = {2015},
}

@article{Marsden:1983:COV,
	title = {Coadjoint orbits, vortices, and {Clebsch} variables for incompressible fluids},
	volume = {7},
	issn = {0167-2789},
	url = {https://www.sciencedirect.com/science/article/pii/0167278983901343},
	doi = {10.1016/0167-2789(83)90134-3},
	number = {1},
	journal = {Physica D: Nonlinear Phenomena},
	author = {Marsden, Jerrold and Weinstein, Alan},
	month = may,
	year = {1983},
	pages = {305--323},
}

@article{Marsden:1984:SPR,
	title = {Semidirect {Products} and {Reduction} in {Mechanics}},
	volume = {281},
	issn = {00029947},
	url = {http://www.jstor.org/stable/1999527},
	doi = {10.2307/1999527},
	number = {1},
	urldate = {2026-07-28},
	journal = {Transactions of the American Mathematical Society},
	publisher = {American Mathematical Society},
	author = {Marsden, Jerrold E. and Ratiu, Tudor and Weinstein, Alan},
	year = {1984},
	pages = {147--177},
}

@article{Marsden:1997:IMS,
  title={Introduction to mechanics and symmetry},
  author={Marsden, Jerrold E and Ratiu, Tudor S and Hermann, Robert},
  journal={SIAM Review},
  volume={39},
  number={1},
  pages={152--152},
  year={1997},
  publisher={Philadelphia, Society for Industrial and Applied Mathematics.}
}

@article{Marsden:2001:DMV,
	title = {Discrete mechanics and variational integrators},
	volume = {10},
	issn = {0962-4929},
	url = {https://www.cambridge.org/product/C8F45478A9290DEC24E63BB7FBE3CEB5},
	doi = {10.1017/S096249290100006X},
	journal = {Acta Numerica},
	publisher = {Cambridge University Press},
	author = {Marsden, J. E. and West, M.},
	year = {2001},
	note = {Edition: 2003/01/09},
	pages = {357--514},
}

@article{Modin:2020:LPM,
	title = {Lie--{Poisson} {Methods} for {Isospectral} {Flows}},
	volume = {20},
	issn = {1615-3383},
	url = {https://doi.org/10.1007/s10208-019-09428-w},
	doi = {10.1007/s10208-019-09428-w},
	number = {4},
	journal = {Foundations of Computational Mathematics},
	author = {Modin, Klas and Viviani, Milo},
	month = aug,
	year = {2020},
	pages = {889--921},
}

@article{Modin:2024:TDF,
	title = {Two-{Dimensional} {Fluids} {Via} {Matrix} {Hydrodynamics}},
	volume = {250},
	issn = {1432-0673},
	url = {https://doi.org/10.1007/s00205-025-02154-4},
	doi = {10.1007/s00205-025-02154-4},
	number = {1},
	journal = {Archive for Rational Mechanics and Analysis},
	author = {Modin, Klas and Viviani, Milo},
	month = jan,
	year = {2026},
	pages = {10},
}

@book{Montgomery:2002:TSG,
	address = {Providence, Rhode Island},
	title = {A tour of subriemannian geometries, their geodesics and applications},
	number = {91},
	publisher = {American Mathematical Society},
	author = {Montgomery, Richard},
	year = {2002},
}

@article{Modin:2025:MHD,
	title = {A brief introduction to matrix hydrodynamics},
	volume = {14},
	issn = {2158-2491},
	url = {https://www.aimsciences.org/article/doi/10.3934/jcd.2026001},
	doi = {10.3934/jcd.2026001},
	number = {0},
	journal = {Journal of Computational Dynamics},
	author = {Modin, Klas and Viviani, Milo},
	month = dec,
	year = {2025},
	pages = {17--35},
}

@article{Mullen:2009:EPI,
	address = {New York, NY, USA},
	title = {Energy-preserving integrators for fluid animation},
	volume = {28},
	issn = {0730-0301},
	url = {https://doi.org/10.1145/1531326.1531344},
	doi = {10.1145/1531326.1531344},
	number = {3},
	journal = {ACM Trans. Graph.},
	publisher = {Association for Computing Machinery},
	author = {Mullen, Patrick and Crane, Keenan and Pavlov, Dmitry and Tong, Yiying and Desbrun, Mathieu},
	month = jul,
	year = {2009},
}

@article{Munthe:1999:HOR,
	address = {NLD},
	title = {High order {Runge}-{Kutta} methods on manifolds},
	volume = {29},
	issn = {0168-9274},
	number = {1},
	journal = {Appl. Numer. Math.},
	publisher = {Elsevier Science Publishers B. V.},
	author = {Munthe-Kaas, Hans},
	month = jan,
	year = {1999},
	pages = {115--127},
}

@phdthesis{Nabizadeh:2025:thesis,
	type = {{PhD} {Thesis}},
	title = {Fluid {Dynamics}: {From} {Geometric} {Formulations} to {Structure}-{Preserving} {Simulations}},
	isbn = {9798286494286},
	url = {https://escholarship.org/uc/item/26f253tx},
	language = {English},
	author = {Nabizadeh, Mohammad},
	year = {2025},
	note = {Publication Title: ProQuest Dissertations and Theses},
}

@article{Nabizadeh:2024:FIP,
	address = {New York, NY, USA},
	title = {Fluid {Implicit} {Particles} on {Coadjoint} {Orbits}},
	volume = {43},
	issn = {0730-0301},
	url = {https://doi.org/10.1145/3687970},
	doi = {10.1145/3687970},
	number = {6},
	journal = {ACM Trans. Graph.},
	publisher = {Association for Computing Machinery},
	author = {Nabizadeh, Mohammad Sina and Roy-Chowdhury, Ritoban and Yin, Hang and Ramamoorthi, Ravi and Chern, Albert},
	month = nov,
	year = {2024},
}

@article{Nabizadeh:2022:CF,
	address = {New York, NY, USA},
	title = {Covector fluids},
	volume = {41},
	issn = {0730-0301},
	url = {https://doi.org/10.1145/3528223.3530120},
	doi = {10.1145/3528223.3530120},
	number = {4},
	journal = {ACM Trans. Graph.},
	publisher = {Association for Computing Machinery},
	author = {Nabizadeh, Mohammad Sina and Wang, Stephanie and Ramamoorthi, Ravi and Chern, Albert},
	month = jul,
	year = {2022},
}

@article{Natale:2018:VFD,
	title = {A variational {$\boldsymbol{H}(\mathrm{div})$} finite-element discretization approach for perfect incompressible fluids},
	volume = {38},
	issn = {0272-4979},
	url = {https://doi.org/10.1093/imanum/drx033},
	doi = {10.1093/imanum/drx033},
	number = {3},
	urldate = {2026-07-28},
	journal = {IMA Journal of Numerical Analysis},
	author = {Natale, Andrea and Cotter, Colin J},
	month = jul,
	year = {2018},
	pages = {1388--1419},
}

@article{Pavlov:2011:SPD,
	title = {Structure-preserving discretization of incompressible fluids},
	volume = {240},
	issn = {0167-2789},
	url = {https://www.sciencedirect.com/science/article/pii/S0167278910002873},
	doi = {10.1016/j.physd.2010.10.012},
	number = {6},
	journal = {Physica D: Nonlinear Phenomena},
	author = {Pavlov, D. and Mullen, P. and Tong, Y. and Kanso, E. and Marsden, J.E. and Desbrun, M.},
	month = mar,
	year = {2011},
	pages = {443--458},
}

@article{Piccione:2001:VAG,
	title = {Variational aspects of the geodesics problem in sub-{Riemannian} geometry},
	volume = {39},
	issn = {0393-0440},
	url = {https://www.sciencedirect.com/science/article/pii/S0393044001000110},
	doi = {10.1016/S0393-0440(01)00011-0},
	number = {3},
	journal = {Journal of Geometry and Physics},
	author = {Piccione, Paolo and Tausk, Daniel V.},
	month = sep,
	year = {2001},
	pages = {183--206},
}

@article{Putnam:1952:OGH,
	title = {The {Orthogonal} {Group} in {Hilbert} {Space}},
	volume = {74},
	issn = {00029327, 10806377},
	url = {http://www.jstor.org/stable/2372068},
	doi = {10.2307/2372068},
	number = {1},
	urldate = {2026-07-28},
	journal = {American Journal of Mathematics},
	publisher = {The Johns Hopkins University Press},
	author = {Putnam, Calvin R. and Wintner, Aurel},
	year = {1952},
	pages = {52--78},
}

@article{San:2013:CGP,
	title = {A coarse-grid projection method for accelerating incompressible flow computations},
	volume = {233},
	issn = {0021-9991},
	url = {https://www.sciencedirect.com/science/article/pii/S0021999112005347},
	doi = {10.1016/j.jcp.2012.09.005},
	journal = {Journal of Computational Physics},
	author = {San, Omer and Staples, Anne E.},
	month = jan,
	year = {2013},
	pages = {480--508},
}

@book{Tritton:2012:PFD,
  title={Physical fluid dynamics},
  author={Tritton, David J},
  year={2012},
  publisher={Springer Science \& Business Media}
}

@article{vanderSchaft:1994:HFN,
	title = {On the {Hamiltonian} formulation of nonholonomic mechanical systems},
	volume = {34},
	issn = {0034-4877},
	url = {https://www.sciencedirect.com/science/article/pii/0034487794900388},
	doi = {10.1016/0034-4877(94)90038-8},
	number = {2},
	journal = {Reports on Mathematical Physics},
	author = {Van Der Schaft, A.J. and Maschke, B.M.},
	month = aug,
	year = {1994},
	pages = {225--233},
}

@article{Vankerschaver:2014:PVD,
	title = {A {Novel} {Formulation} of {Point} {Vortex} {Dynamics} on the {Sphere}: {Geometrical} and {Numerical} {Aspects}},
	volume = {24},
	issn = {1432-1467},
	url = {https://doi.org/10.1007/s00332-013-9182-5},
	doi = {10.1007/s00332-013-9182-5},
	number = {1},
	journal = {Journal of Nonlinear Science},
	author = {Vankerschaver, Joris and Leok, Melvin},
	month = feb,
	year = {2014},
	pages = {1--37},
}

@article{Scipy,
	title = {{SciPy} 1.0: {Fundamental} {Algorithms} for {Scientific} {Computing} in {Python}},
	volume = {17},
	doi = {10.1038/s41592-019-0686-2},
	journal = {Nature Methods},
	author = {Virtanen, Pauli and Gommers, Ralf and Oliphant, Travis E. and Haberland, Matt and Reddy, Tyler and Cournapeau, David and Burovski, Evgeni and Peterson, Pearu and Weckesser, Warren and Bright, Jonathan and van der Walt, Stéfan J. and Brett, Matthew and Wilson, Joshua and Millman, K. Jarrod and Mayorov, Nikolay and Nelson, Andrew R. J. and Jones, Eric and Kern, Robert and Larson, Eric and Carey, C J and Polat, İlhan and Feng, Yu and Moore, Eric W. and VanderPlas, Jake and Laxalde, Denis and Perktold, Josef and Cimrman, Robert and Henriksen, Ian and Quintero, E. A. and Harris, Charles R. and Archibald, Anne M. and Ribeiro, Antônio H. and Pedregosa, Fabian and van Mulbregt, Paul and {SciPy 1.0 Contributors}},
	year = {2020},
	pages = {261--272},
}

@article{Zeitlin:1991:FMA,
	title = {Finite-mode analogs of {2D} ideal hydrodynamics: {Coadjoint} orbits and local canonical structure},
	volume = {49},
	issn = {0167-2789},
	url = {https://www.sciencedirect.com/science/article/pii/016727899190152Y},
	doi = {10.1016/0167-2789(91)90152-Y},
	number = {3},
	journal = {Physica D: Nonlinear Phenomena},
	author = {Zeitlin, V.},
	month = apr,
	year = {1991},
	pages = {353--362},
}

@article{Yoshida:2009:CPB,
	title = {Clebsch parameterization: {Basic} properties and remarks on its applications},
	volume = {50},
	issn = {0022-2488},
	url = {https://doi.org/10.1063/1.3256125},
	doi = {10.1063/1.3256125},
	number = {11},
	urldate = {2026-07-28},
	journal = {Journal of Mathematical Physics},
	author = {Yoshida, Z.},
	month = nov,
	year = {2009},
	pages = {113101},
}

@article{Zhong:1988:LPI,
	title = {Lie-{Poisson} {Hamilton}-{Jacobi} theory and {Lie}-{Poisson} integrators},
	volume = {133},
	issn = {0375-9601},
	url = {https://www.sciencedirect.com/science/article/pii/0375960188907736},
	doi = {10.1016/0375-9601(88)90773-6},
	number = {3},
	journal = {Physics Letters A},
	author = {Zhong, Ge and Marsden, Jerrold E.},
	month = nov,
	year = {1988},
	pages = {134--139},
}

@article{Kozlov:1983:DSN,
	title = {Dynamics of systems with nonintegrable constraints. {III}},
	number = {3},
	journal = {Vestnik Moskovskogo Universiteta. Seriya 1. Matematika. Mekhanika},
	publisher = {Lomonosov Moscow State University},
	author = {Kozlov, Valery Vasil'evich},
	year = {1983},
	pages = {102--111},
}

@book{Aalembert:1743:TDD,
	address = {Paris},
	title = {Traité de dynamique},
	publisher = {David l'aîné},
	author = {d'Alembert, Jean Le Rond},
	year = {1743},
}

@incollection{Arnold:2006:MAC,
	address = {Berlin, Heidelberg},
	title = {Basic {Principles} of {Classical} {Mechanics}},
	isbn = {978-3-540-48926-9},
	url = {https://doi.org/10.1007/978-3-540-48926-9_1},
	doi = {10.1007/978-3-540-48926-9_1},
	booktitle = {Mathematical {Aspects} of {Classical} and {Celestial} {Mechanics}: {Third} {Edition}},
	publisher = {Springer Berlin Heidelberg},
	author = {Arnold, Vladimir I. and Kozlov, Valery V. and Neishtadt, Anatoly I.},
	year = {2006},
	pages = {1--60},
}

@book{Cottrell:2009:IGA,
author = {Cottrell, J. Austin and Hughes, Thomas J. R. and Bazilevs, Yuri},
title = {Isogeometric Analysis: Toward Integration of CAD and FEA},
year = {2009},
isbn = {0470748737},
publisher = {Wiley Publishing},
address   = {Chichester, UK},
edition = {1st}
}

@article{Arnold:2010:FEEC,
  title={Finite element exterior calculus: from Hodge theory to numerical stability},
  author={Arnold, Douglas and Falk, Richard and Winther, Ragnar},
  journal={Bulletin of the American mathematical society},
  volume={47},
  number={2},
  pages={281--354},
  year={2010}
}

@misc{Bloch:2026:Arxiv,
      title={Discrete variational calculus for double-bracket dissipation}, 
      author={Anthony Bloch and Sebastián J. Ferraro and David Martín de Diego and Shreyas Bharadwaj},
      year={2026},
      eprint={2604.26049},
      archivePrefix={arXiv},
      primaryClass={math.NA},
      url={https://arxiv.org/abs/2604.26049}, 
}

@book{Buffa:2016:IGA,
  title={Isogeometric analysis: a new paradigm in the numerical approximation of PDEs: Cetraro, Italy 2012},
  author={Buffa, Annalisa and Sangalli, Giancarlo},
  volume={2161},
  year={2016},
  publisher={Springer},
  address   = {Cham, Switzerland},
}

@book{Gromov:1981:SMV,
    author = {Gromov, M.},
    editor = {Lafontaine, J. and Pansu, P.},
    title = {Structures métriques pour les variétés Riemanniennes},
    volume = {1},
    year = {1981},
    publisher = {Textes Mathématiques},
    address = {Paris}
}

@InProceedings{Raviart:1977:MFE,
author="Raviart, P. A.
and Thomas, J. M.",
editor="Galligani, Ilio
and Magenes, Enrico",
title="A mixed finite element method for 2-nd order elliptic problems",
booktitle="Mathematical Aspects of Finite Element Methods",
year="1977",
publisher="Springer Berlin Heidelberg",
address="Berlin, Heidelberg",
pages="292--315",
isbn="978-3-540-37158-8"
}

@article{chetaev1932gauss,
  title = {O printsipe {G}aussa (On the {G}auss principle)},
  author = {Chetaev, Nikolai Guryevich},
  journal = {Izvestiya Kazanskogo Fiziko-Matematicheskogo Obshchestva},
  series = {3},
  volume = {6},
  pages = {323--326},
  year = {1932--1933},
  note = {In Russian}
}

\end{document}